\documentclass[a4paper,12pt]{article}

\RequirePackage{amsthm,amsmath,amssymb}
\RequirePackage[numbers]{natbib}

\RequirePackage[colorlinks,citecolor=blue,urlcolor=blue]{hyperref} 
\usepackage{graphicx}
\usepackage{authblk}
\usepackage{titlesec}
\usepackage{color}
\numberwithin{equation}{section}

\theoremstyle{plain}
\newtheorem{theorem}{Theorem}[section]
\newtheorem{lemma}[theorem]{Lemma}
\newtheorem{remark}{Remark}
\newtheorem{assumption}{Assumption}
\newtheorem{proposition}[theorem]{Proposition}
\newtheorem{corollary}[theorem]{Corollary}

\theoremstyle{definition}

\usepackage{booktabs}
\usepackage{ulem}
\usepackage{makecell}
\usepackage{enumitem}

\providecommand{\keywords}[1]
{
  \small	
  \textbf{Keywords:} #1
}
 
\usepackage[lmargin=2.5cm, rmargin=2.5cm, bmargin=2.7cm, tmargin=2.3cm]{geometry}

\begin{document}

\title{Splitting schemes and estimators for stochastic differential equations with H\"older multiplicative noise}
\author[1]{Bowen Fang}
\author[1]{Dario Span\`o}
\author[1]{Massimiliano Tamborrino}
\affil[1]{\small Department of Statistics, University of Warwick, Coventry,UK}
\maketitle

\begin{abstract}
We study parameter estimation for univariate stochastic differential equations with locally Lipschitz drift and Hölder continuous multiplicative diffusion, a class commonly arising in several applications. Existing inference methods typically rely on either the Euler–Maruyama discretisation, despite its lack of strong convergence and failure to preserve the state space, or on approximations, e.g. Gaussian approximation or truncation of Hermite's expansions, impacting on their stability and computational efficiency. We introduce the first explicit pseudo-likelihood estimators based on numerical splitting schemes that are both strong mean-square convergent and state space preserving for this class of SDEs. Our approach is based on a novel decomposition of the SDE that exploits reducibility and the Lamperti transform, leading to Lie–Trotter (LT) and Strang splitting schemes yielding explicit pseudo-likelihoods and maximum likelihood estimators based on them. We prove strong mean-square convergence, state space preservation, and improved robustness with respect to the discretisation step compared to Euler–Maruyama-based methods.
We further establish consistency and asymptotic normality of the LT estimator. Because the proposed numerical scheme couples drift and diffusion parameters in the pseudo-likelihood, the asymptotic analysis requires new proof techniques. Extensive simulations demonstrate that the proposed estimators outperform existing methods in both accuracy and computational efficiency.
\end{abstract}
\noindent\keywords{Asymptotic normality, consistency, 
discretely observed diffusions, Hölder diffusion coefficients, splitting schemes, mean square convergence, Lamperti transform}
\section{Introduction}\label{sec:intro}
In this work, we consider univariate stochastic differential equations (SDEs) with multiplicative noise of the form
\begin{equation}\label{eqn:the sde}
    dX_t=f(X_t;\theta)dt+g(X_t;\sigma)dW_t,
\end{equation}
where the drift coefficient $f$ is a locally Lipschitz function, the diffusion coefficient $g$ is a H\"older continuous function, and $W_t$ is a one-dimensional Wiener process. Our primary goal is the estimation of the underlying unknown parameter vector $\alpha=(\theta,\sigma)^T\in\Theta=\Theta_{\theta}\times\Theta_{\sigma}$, from discrete time observations of the process, where $\theta\in \Theta_{\theta}\subseteq\mathbb{R}^{d_1}$ and $\sigma\in\Theta_{\sigma}\subseteq\mathbb{R}^{d_2}$  enter in the drift and diffusion coefficients respectively. This broad class of SDEs includes several well known models, such as the Pearson diffusion class \cite{forman_pearson_2008}, the stochastic Ginzburg-Landau equation \cite{ginzburg_superconductivity_2009}, the stochastic Verhulst equation \cite{verhulst_recherches_1845}, the Feller diffusion with logistic growth \cite{lambert_branching_2005}, and the Ahn and Gao model \cite{ahn_parametric_1999}. The Pearson diffusion class includes, among others, processes like the Wright-Fisher or Jacobi diffusion, which is used in population genetics to model changes in allele frequencies \cite{casanova_duality_2017,koskela_bernoulli_2024}, and neuroscience to model single neuron dynamics \cite{donofrio_jacobi_2018}, the Ornstein–Uhlenbeck (OU) \cite{uhlenbeck_theory_1930} (also known as Va\v{s}icek
 model in the finance literature \cite{vasicek_equilibrium_1977}), and the Cox–Ingersoll–Ross (CIR) process \cite{cox_theory_1985}, both used, among other fields, in finance and econometrics to model asset prices and interest rates with features such as mean reversion, volatility clustering, and heavy-tailed distributions \cite{cox_theory_1985}, with the CIR also used to model population
growth.

We are interested in estimating the unknown parameter vector $\alpha$ in a likelihood-based framework. 
Exact likelihood-based estimators rely on the true underlying transition density of the model being available either in closed analytical form, or through a procedure that allows sampling from it. This is rarely possible for SDEs of the form \eqref{eqn:the sde} with non-additive noise,  except for a few explicitly solvable cases such as, for example, the CIR, and the geometric Brownian motion, or for sampleable ones, such as the Wright-Fisher diffusion. For the majority of processes, the transition density is instead  approximated to compute what is known as \textit{pseudo-likelihood}, which is then maximised to derive the so-called \textit{pseudo maximum likelihood estimator} (MLE).  
 The most straightforward and widely used approximation is obtained via the Euler-Maruyama (EuM) numerical scheme \cite{kloeden_numerical_1992}, yielding a Gaussian transition density \cite{florens-zmirou_approximate_1989}. While popular for its (computational) simplicity, the EuM method (as well as other canonical It\^o-Taylor expansion approaches such as Milstein \cite{milshtejn_approximate_1975}) requires globally Lipschitz and linear growth conditions on $f,g$ to ensure a strong convergence property, otherwise it may diverge in the mean-square sense \cite{hutzenthaler_strong_2011}. Several modified EuM methods have been proposed for SDEs with non-globally Lipschitz drift and superlinear growth, e.g., truncated EuM \cite{mao_truncated_2015,mao_strong_2013,yang_truncated_2020}, tamed EuM \cite{hutzenthaler_strong_2012,Sabanis_2016}, 
and implicit EuM \cite{higham_strong_2002}, all mean-square convergent schemes with an approximate Gaussian  transition density (except the implicit EuM). Nevertheless, these modified EuM-based estimators struggle with \lq\lq large\rq\rq observation time steps $h$, see e.g. \cite{,buckwar_spectral_2020,Huangetal2025} and references therein, as the underlying numerical schemes fail to preserve some desired properties of the SDE solution, such as state space of the processes, boundary behaviour, and asymptotic distribution \cite{buckwar_splitting_2022}. Kessler \cite{kessler_estimation_1997} 
 derived estimators based on approximating the transition density with a Gaussian density with true conditional moments. A different approach is proposed in the celebrated work of A\"{i}t-Sahalia
 \cite{ait-sahalia_maximum_2002}, where a pseudo MLE is derived by maximising the pseudo-likelihood obtained by applying a Lamperti transform \cite{lamperti_semi-stable_1972}
 to the SDE \eqref{eqn:the sde}, yielding a new SDE in $Y_t$ with (generally) non-linear drift and unit diffusion coefficients. $Y_t$ is then normalised to obtain a new random variable whose density  is approximated by means of a Hermite polynomial expansion, before being mapped back to the desired transition density of $X_t$. As the Lamperti transform is also crucial in our proposed numerical schemes and estimators, we recall it here for convenience. The Lamperti transform of the SDE \eqref{eqn:the sde}
 is given by
\begin{equation}\label{eqn:Lamperti transform}
    Y_t=v(X_t)=\int^{X_t}\frac{1}{g(u)}du,
\end{equation}
leading to 
  \begin{equation}\label{Lamperti}
    dY_t=
    b(Y_t)dt+dW_t:=
    \left[\frac{f\left(v^{-1}(Y_t)\right)}{g\left(v^{-1}(Y_t)\right)}-\frac{1}{2} g'\left(v^{-1}(y)\right)\right]dt+dW_t,
\end{equation}
where $g'(x)$ denotes the derivative of $g(x)$ in $x$. The transformed process $Y_t$ may be more analytically tractable than the
original process $X_t$, and may sometimes admit an explicit
solution, as is the case for the different types of reducible SDEs \cite{kloeden_numerical_1992}. 
 
Estimators based on strong convergent and property-preserving numerical schemes for SDEs with multiplicative H\"older noise are currently missing. Indeed, the previously-mentioned approaches, as well as all the main alternative inferential approaches (such as local linearisation \cite{Ozaki1985}, martingale estimation functions \cite{bibby_martingale_1995}, density-based approximation \cite{iguchi_closed-form_2025}, Markov chain Monte Carlo methods \cite{Beskosetal2006}), either use EuM, which is not convergent, or focus on the pseudo-likelihood rather than the underlying numerical scheme. Here, we fill this gap by proposing two new estimators based on the
Lie-Trotter (LT) and Strang (S)  splitting schemes derived via a clever  decomposition of the original SDE \eqref{eqn:the sde}, making use of the Lamperti transform and reducible SDEs.  Splitting schemes have been recently shown to be both strong convergent and structure-preserving for the scalar inhomogeneous geometric Brownian motion \cite{tubikanec_qualitative_2022}, and 
multivariate semi-linear SDEs with locally Lipschitz drift and constant  diffusion coefficient \cite{buckwar_splitting_2022}. 
The idea behind these schemes is to decompose 
the SDE into several solvable subequations, and then compose their solutions in a suitable way to approximate the true solution of the original system. We refer to \cite{strang_construction_1968,trotter_product_1959,vovchanskyi_quick_2023} for a thorough discussion of operator-splitting methods in differential equations, and to \cite{buckwar_splitting_2022,foster_high_2024,kelly_adaptive_2023,misawa_numerical_2000,ninomiya_new_2009,tubikanec_qualitative_2022} for more recent applications in SDEs. Crucially for our goal, splitting schemes have been recently used for inferential purposes, either to simulate trajectories within Approximate Bayesian inference (ABC) algorithms \cite{buckwar_spectral_2020,ditlevsen_network_2025,Jovanovskietal2025,STT2025}, or to perform pseudo MLE from the derived pseudo-likelihood \cite{ditlevsen_warning_2023,Huangetal2025,pilipovic_parameter_2024,pilipovic_parameter_2025}. Here, we focus on the latter approach, but  our findings can be also naturally applied to other inferential methods beyond a ``direct MLE" perspective. Indeed, our splitting schemes yield explicit one-step conditional means and variances which could be used within the generalised method of moments, or would represent the cornerstone for any methods tying simulation to inference (ABC 
\cite{sisson2018handbook}, sequential Monte Carlo \cite{del2006sequential}, Markov Chain Monte Carlo \cite{brooks2011handbook}, simulation-based-inference \cite{cranmer2020frontier}, etc.), where simulating accurate trajectories for large simulation time steps is of paramount interest to minimise the computational costs.
 In \cite{pilipovic_parameter_2024},  
the numerical splitting schemes proposed in \cite{buckwar_splitting_2022} for multivariate SDEs with locally Lipschitz drift and additive noise are used to derive pseudo-likelihood-based estimators which are consistent and asymptotically normal.  We go  beyond this: we derive estimators based on mean-square-convergent numerical splitting schemes for the  broader class of  SDEs with multiplicative noise. 

Deriving a general, suitable, splitting decomposition is the hardest task, particularly for complex SDEs such as \eqref{eqn:the sde}, as the underlying subequations would typically depend on the diffusion coefficient. 
Recently, a few methods have been proposed combining Lamperti transform and splitting schemes, with the idea of numerically solving the transformed SDE \eqref{Lamperti} for $Y_t$, and then map the solution back to the original SDE in $X_t$ via its inverse $v^{-1}(Y_t)$. 
In \cite{kelly_adaptive_2023}, the authors first Lamperti-transform the CIR process and then decompose the transformed SDE into three subequations. In \cite{ulander_boundary-preserving_2024}, the author considers the class of bounded  SDEs, Lamperti-transformed it, added and subtracted a constant to the new drift term $b(y)$, leading to a splitting into an ordinary differential equation (ODE), and a SDE with constant drift and diffusion.

Here, we propose a splitting decomposition of \eqref{eqn:the sde} taking  advantage of 
the so-called \textit{reducible SDEs} \cite{kloeden_numerical_1992}, i.e., SDEs which, with an appropriate transformation, can be reduced to (linear) solvable SDEs. By adding and subtracting a suitable, generally non-constant term in the drift (e.g., $ g(x)g'(x)/2$) of the original SDE \eqref{eqn:the sde}, we can decompose it  into an ODE and an explicitly solvable SDE (as the Lamperti SDE \eqref{Lamperti} has solution  $Y_t=W_t$), and then compose their solutions via the LT and S approach. Crucially for our inferential purposes, the derived schemes yield non-Gaussian pseudo-likelihoods via approximating transition densities which are functions of Gaussian distributions. A similar setting was recently observed in the above-mentioned paper \cite{pilipovic_parameter_2024}. 
Compared to \cite{pilipovic_parameter_2024}, the framework considered in this paper is further generalised in two ways. First, we consider a broader class of SDEs, allowing the diffusion coefficient to be H\"older continuous. Second, the derived pseudo-likelihoods are functions of terms depending jointly on the whole parameter vector $\alpha=(\theta,\sigma)^T$ rather than only $\theta$ and $\sigma$ considered separately, requiring the development of novel moment bounds and proof techniques, which can then be used in other settings where a similar framework is observed. Intuitively, the need to tackle $\theta$ and $\sigma$ jointly arises because the new splitting decomposition considered here \lq\lq brings in\rq\rq\ the parameter $\sigma$ from the diffusion coefficient into the drift of the SDE subequation.

Proving that the proposed splitting schemes are strong convergent is not immediate in our generalised setting \eqref{eqn:the sde}. Indeed, existing results are based on 
Milstein's fundamental theorem \cite{milstein_numerical_1995} 
which holds for SDEs with globally Lipschitz drift and diffusion coefficients, or on its more general version by Tretyakov and Zhang \cite{tretyakov_fundamental_2013}  for SDEs with one-sided Lipschitz drift with polynomial growth. For these SDEs, the convergence on the whole time interval is guaranteed by only checking their one-step numerical consistency (intuitively, the \lq\lq convergence\rq\rq\ in one interval, from time $t_i$ to time $t_{i+1}$), notably simplifying the proofs. However, Milstein's-type theorems cannot be applied here, due to the diffusion coefficient being, in general, only H\"older continuous. \textit{Ad-hoc} proofs have then been proposed in the literature, see, e.g., \cite{chen_strong_2023} and \cite{kelly_adaptive_2023}, for the the Wright-Fisher and the CIR, respectively. Here, we provide a comprehensive approach to prove the desired strong mean-square convergence of the LT scheme for a general splitting decomposition, suitable also for time-inhomogeneous versions of \eqref{eqn:the sde}. The proof is based on the idea of deriving a LT approximation on the whole time interval, as recently proposed for semi-discrete (SD) schemes \cite{stamatiou_numerical_2016,stamatiou_boundary_2018} and Lamperti-splitting schemes \cite{ulander_boundary-preserving_2024}. 

To summarise, for the considered class of scalar SDEs with locally Lipschitz drift and H\"older diffusion coefficients, our main contributions are:
\begin{enumerate}
    \item The derivation of two explicit and easy-to-construct numerical splitting schemes for the solution of the whole class of SDEs \eqref{eqn:the sde} which:
    \begin{itemize}
        \item[(i)] are strong mean-square convergent (cf. Theorem \eqref{thrm:LT convergence} for LT),  even for a time-inhomogeneous version of the SDE and a more general splitting decomposition based only on locally Lipschitz drift coefficient;
        \item[(ii)] are state space preserving when considering entrance boundaries;
        \item[(iii)] have explicit one-step numerical consistencies (quantifying the 
        average accuracy of the numerical scheme after one iteration, starting from the true solution at the previous step) of order 2 and 3 for LT and S respectively (cf. Proposition \ref{prop:pearson consistency});
        \item[(iv)] admit  an explicit formulation of the one-step and global conditional means (and variances), which may be used to study the bias introduced by the numerical schemes, or to perform the generalised method of moments;
        \item[(v)] may be used within simulation-based inference \cite{cranmer2020frontier} and likelihood-free approaches, where simulating accurate trajectories
for large simulation time steps is crucial to minimise the computational costs.
    \end{itemize}
    \item The derivation of two splitting-based estimators, which are: 
    \begin{itemize}
        \item[(i)] simple to compute and evaluate, especially the LT scheme;
        \item[(ii)] shown to be intrinsically linked to A\"it-Sahalia's Hermite expansions;
        \item[(iii)]  consistent and asymptotically normal under the more general assumption of locally H\"older diffusion coefficient and pseudo-likelihood with tangled $(\theta,\sigma)$, with nontrivial proofs.
    \end{itemize}
\end{enumerate}
 
The paper is structured as follows. In Section \ref{sec:set up}, we introduce the class of SDEs of interest, specifying the assumptions on the drift and diffusion coefficients, construct the proposed splitting numerical schemes and estimators, and recall alternative estimators based on Kessler's Gaussian approximation and Hermite expansion. 
In Section \ref{sec:convergence}, we prove the strong mean-square convergence of the LT scheme and derive further properties for both schemes. In Section \ref{sec:likelihood}, we use these splitting schemes to approximate the transition densities of the processes, deriving pseudo MLEs which we show being consistent and asymptotically normal for the LT estimator. A remark on the S asymptotic results is provided. Finally, in Sections \ref{sec:numerical results},  we report the empirical performance of the numerical schemes and the splitting estimators on a wide class of SDEs, comparing them with alternative numerical schemes and inferential approaches, before a final discussion in Section \ref{sec:conclusion}.
\section{Model and properties}\label{sec:set up}
Let $\left(\Omega, \mathcal{F}, \mathbb{P}\right)$ be a complete probability space with a complete and right-continuous filtration $\left(\mathcal{F}_t\right)_{t \geq 0}$, and let $(W_t)_{t\geq 0}$ be a one-dimensional Brownian motion defined on that space, and adapted to $\left(\mathcal{F}_t\right)_{t \geq 0}$. In this paper, we focus on SDEs of the form \eqref{eqn:the sde}, with the following assumptions on the drift and diffusion coefficients. 
\begin{assumption}\label{assump:main}
    For SDE \eqref{eqn:the sde}, the functions $f,g:\mathbb{R}\to \mathbb{R}$ satisfy  
    \begin{itemize}
        \item[(i)] $f(x)=f_1(x)+f_2(x)$, where $f_1,f_2$ are continuous locally Lipschitz functions, i.e.,
        for any $R>0$ and $\max(|x_1|,|x_2|)\leq R$, there exists an arbitrary constant $C_R$ that only depends on $R$ such that
        \begin{equation}
            \left|f_i(x_1)-f_i(x_2)\right|\leq C_R\left(|x_1-x_2|\right), \quad i\in \{1,2\},\label{eqn:f condition}
        \end{equation}
        \item[(ii)] The function $g$ is locally H\"older with exponent $q\in(0,1/2)$, i.e., for any $R>0$ and $\max(|x_1|,|x_2|)\leq R$, there exists $q\in (0,\frac{1}{2})$ such that
        \begin{equation}
            \left|g(x_1)-g(x_2)\right|\leq C_R \left(|x_1-x_2|+|x_1-x_2|^q\right)\label{eqn:g condition}.
        \end{equation}
    \end{itemize}
    
    \end{assumption}
 Assumption \ref{assump:main} guarantees the existence and uniqueness of a strong solution of the SDE \eqref{eqn:the sde}, see \cite{khasminskii_stochastic_2012,mao_stochastic_2011}. The locally Lipschitz condition relaxes the commonly assumed globally Lipschitz or polynomial growth conditions on the drift and diffusion terms, broadening the class of SDEs considered in this work. 

    The following assumption states all conditions needed for the proposed numerical splitting scheme to be well defined.

\begin{assumption}\label{Assumption2}
\begin{itemize}
\item[(i)] The function $g$ is differentiable.
\item[(ii)] The function $g(x)g'(x)$ is locally Lipschitz.
\item[(iii)] The function $1/g(x)$ is integrable, i.e. $v(x)=\int^x 1/g(s)ds$ is well defined (cf. \cite{kloeden_numerical_1992}).
\item[(iv)] The function $v^{-1}$ is well defined (cf. \cite{kloeden_numerical_1992}) when evaluated in $\mathbb{R}$.
\end{itemize}
\end{assumption}
Finally, the following  conditions are needed to prove the consistency and asymptotic normality of the derived splitting estimators. Below, we use the notation $f(x;\theta)$ and $g(x;\theta)$ to emphasise the dependency of $f$ and $g$ on $\theta$ and $\sigma$, respectively. 
\begin{assumption}\label{Assumption3}
\begin{enumerate}
    \item[(i)] 
The SDE \eqref{eqn:the sde} is ergodic, and for all $p\geq0$, $\sup _t \mathbb{E}\left|X_t\right|^p<\infty$.      \item[(ii)]$f(x;\theta)$ and $g(x;\sigma)$ are of polynomial growth in $x$, uniformly in $\Theta$.  $f_1(x;\alpha)$ and $g(x;\sigma)$ are twice continuously differentiable with respect to $\alpha$ and $\sigma$, uniformly in $x$. The derivatives are of polynomial growth in $x$, uniformly in $\alpha$.
    \item[(iii)] $v(x;\sigma)$ is twice continuously differentiable with respect to $\sigma$ for all $x$, and its derivative is of polynomial growth uniformly in $\sigma$. 
    \item[(iv)] $f(x;\theta)$ and $g(x;\sigma)$ are identifiable, i.e. if $f(x;\theta_1)=f(x;\theta_2)$ for all $x\in \mathbb{R}$, then $\theta_1=\theta_2$.
    \end{enumerate}
\end{assumption}
\subsection{Splitting schemes}\label{sec:splittings}
Let us consider a discretised time interval $[0,T]$ with equidistant time steps $h=t_k-t_{k-1}$, with $t_0=0, \ldots, t_n = T$. We split \eqref{eqn:the sde} as 
\begin{eqnarray}\label{eqn:LT subequations}
dX^{[1]}_t&=&f_1(X^{[1]}_t)dt,\\
   \label{eqn:LT subequations2}dX^{[2]}_t&=&f_2(X^{[2]}_t)dt+g(X^{[2]}_t)dW_t,
\end{eqnarray}
with $f_1(x)+f_2(x)=f(x)$, 
such that $f_1, f_2$ and $g$ satisfy Assumption \ref{assump:main}. An explicit solution of \eqref{eqn:LT subequations2} is generally not available. To tackle this, we consider a specific formulation of $f_1$ and $f_2$ leading to
\begin{eqnarray}
\label{eqn:ODE1}
dX_t^{[1]}&=&
\left[f(X_t^{[1]})-\frac{1}{2}g(X_t^{[1]})g'(X_t^{[1]})\right]dt,
\\
\label{eqn:SDE2}
dX_t^{[2]}&=&
\frac{1}{2}g(X_t^{[2]})g'(X_t^{[2]})dt + g(X_t^{[2]})dW_t,
\end{eqnarray}
with $g$ satisfying Assumption \ref{assump:main}(ii) and Assumption \ref{Assumption2}(ii), guaranteeing that the  corresponding $f_1(x)=f(x)-g(x)g'(x)/2$ and $f_2(x)= g(x)g'(x)/2$ satisfy Assumption \ref{assump:main}(i), with both subequations explicitly solvable.  
The proposed decomposition takes advantage of the obtained reducible SDE \eqref{eqn:SDE2} (Case 1 in Chapter 4, \cite{kloeden_numerical_1992}), with one-step solution (or $h$-step solution) from $X^{[2]}_{t_k}$ obtained via the Lamperti transform \cite{lamperti_semi-stable_1972} $Y^{[2]}_t=v(X_{t_k}^{[2]})$:
\begin{equation}\label{solSDE2}
X_{t_{k+1}}^{[2]}=\varphi_h^{[2]}(X_{t_k}^{[2]})=v^{-1}(\xi_k+v(X_{t_k}^{[2]})),\qquad k=0,\ldots, N-1,
\end{equation}
with $v(x)$ defined in \eqref{eqn:Lamperti transform} for $x$ being an arbitrary value in the state space of $X=(X_t)_{t\geq 0}$ and $\xi_k, k=0,\ldots, N-1$ being independent and identically normally distributed random variables with zero mean  and variance $h$, i.e. $\xi_k\sim \mathcal{N}(0,h)$.
\begin{remark}\label{case2}
In \cite[Chapter 4]{kloeden_numerical_1992}, two more cases of reducible SDEs are proposed with diffusion coefficient $g(x)$ and drift $d_1(x)=\alpha g(x)+g(x)g'(x)/2$ or $d_2(x)=\alpha g(x)v(x)+g(x)g'(x)/2$. Then, our approach of decomposing the SDE \eqref{eqn:the sde} via \eqref{eqn:LT subequations}-\eqref{eqn:LT subequations2}, with $f_1(x)=f(x)-d_i(x)$ and $f_2(x)=d_i(x), i=1,2$ would lead to additional novel splitting schemes for a broader class of SDEs, as $f_1, f_2$ would also be H\"older when $g(x)$ is H\"older. We illustrate the decomposition via the drift $d_1$ on an SDE with cubic drift and linear diffusion coefficient in Suppl. Mat. \ref{AppendixC}, showing how this leads to a state preserving numerical scheme.
\end{remark}

By assuming the locally Lipschitz condition \eqref{eqn:ODE1} on $f_1$, from Picard's theorem \cite{sideris_ordinary_2013}, the ODE \eqref{eqn:ODE1} has a unique global solution $\varphi_t^{[1]}:\mathbb{R}\to\mathbb{R}$ on $[0,T]$ (so it does not explode in finite time), and its one-step solution becomes
\begin{equation}\label{solODE1}
X^{[1]}_{t_{k+1}}=\varphi_h^{[1]}(X^{[1]}_{t_k}).
\end{equation}
After explicitly solving the two subequations \eqref{eqn:ODE1}-\eqref{eqn:SDE2}, we  compose their exact solutions \eqref{solSDE2}, \eqref{solODE1} to obtain an explicit numerical approximation of the SDE \eqref{eqn:the sde} via the LT    \cite{trotter_product_1959} 
\begin{equation}\label{eqn:LT subsystem}
X^{\text{LT}}_{t_{k+1}}=\left(\varphi_h^{[2]}\circ\varphi_h^{[1]}\right)(X_{t_k})=v^{-1}
    \left(\xi_k+v(\varphi^{[1]}_{h}(X_{t_k}))\right)
\end{equation}
and the S composition \cite{strang_construction_1968}
\begin{equation}\label{eqn:Strang subsystem}
        X^{\text{S}}_{t_{k+1}}=\left(\varphi^{[1]}_{\frac{h}{2}}\circ \varphi_{\frac{h}{2}}^{[2]}\circ \varphi_{\frac{h}{2}}^{[1]}\right)(X_{t_k})=
\varphi^{[1]}_{\frac{h}{2}}\left(v^{-1}({\xi_k}+v(\varphi^{[1]}_{\frac{h}{2}}(X_{t_k})))\right),
    \end{equation}
    starting in $X_{t_k}$ at time $t_k$, with $\xi_k \sim N(0,h)$. 
The composition of  solutions is not unique, as any of their permutation would yield another possible splitting approximation. However, the considered compositions bring  advantages on the inferential side, as discussed in Section \ref{sec:likelihood}. Moreover, the considered S solution is less computationally expensive than the other S composition $(\varphi^{[2]}_{\frac{h}{2}}\circ\varphi^{[1]}_{h}\circ\varphi^{[2]}_{\frac{h}{2}})$, as it only requires the simulation of one random Gaussian increment per iteration instead of two.
\begin{remark}\label{semi-discrete}
If the ODE \eqref{eqn:ODE1} (or more generally \eqref{eqn:LT subequations}) cannot be solved explicitly, a suitable numerical scheme may be used instead. If using  Euler, the resulting LT scheme will coincide with the so-called SD scheme \cite{halidias_semi-discrete_2012,stamatiou_boundary_2018}. As such scheme converges strongly to the true solution under Assumption 1 \cite{higham_strong_2002,stamatiou_boundary_2018}, the convergence of the LT scheme  with the true ODE solution is also expected, as shown in Section \ref{sec:convergence}.
\end{remark}
\subsection{Estimators}
We aim to estimate $\alpha=(\theta,\sigma)\in \Theta=\Theta_{\theta}\times\Theta_{\sigma}$, with both $\Theta_{\theta}\subseteq\mathbb{R}^{d_1}$ and $\Theta_{\sigma}\subseteq\mathbb{R}^{d_2}$  compact, from observations of the process $X$ at discrete time $t_k$ in $[0,T]$, i.e. $X_{t_k}=x_k, k=0,\ldots, N$ and $\underline{x}=(x_1,\ldots, x_k)$. In this section, we first derive the LT and S estimators, and then, we recall the Kessler and the Hermite estimators, which we will use in the simulation study for comparison.

\subsubsection{Splitting estimators}
By looking at the one-step LT and S splitting schemes derived in
\eqref{eqn:LT subsystem} and \eqref{eqn:Strang subsystem}, respectively, we have that, conditioned on $X_{t_k}=x$, the corresponding one-step transition densities are non-linear transformations of  Gaussian random variables  $\xi_k+v(\varphi^{[1]}_{h}(x))\sim \mathcal{N}(v(\varphi^{[1]}_{ h}(x)),h)$ via the functions $v^{-1}$ and $\varphi_{\frac{h}{2}}^{[1]}\circ v^{-1}$, respectively.

If the inverse of $v^{-1}$ and $(\varphi_{\frac{h}{2}}^{[1]}\circ v^{-1})$ exist and are monotone, thanks to the transformation theorem, the LT transition density becomes
\begin{equation}\label{fLT}
f^{\textrm{LT}}_{X_{t_{k+1}}|X_{t_{k}}}(x|x_0;
\alpha)=\frac{1}{g(x;\sigma)}f_{N}(v(x;\sigma);v(\varphi_h^{[1]}(x_0;\alpha);\sigma),h),
\end{equation}
where $f_N(a;b,c)$ denotes the density of a normal distribution $\mathcal{N}(b,c)$ evaluated in $a$, and similarly for the S transition density. Hence,  
the negative pseudo-log-likelihood functions $l_h$ of the LT and S splitting schemes become
\begin{eqnarray}
    l_h^{\textrm{LT}}(\underline{x};\alpha)&=&\sum_{i=1}^{N}\left[\frac{1}{2h}\left(v(x_{i};\sigma)-v(\varphi^{[1]}_h(x_{{i-1}};\alpha);\sigma)\right)^2+\log |g(x_{i};\sigma)|\right],\label{eqn:LT loglikelihood}\\ 
\nonumber l_h^{\textrm{S}}(\underline{x};\alpha)&=&\sum_{i=1}^N \left[\frac{1}{2h}\left(v(\varphi^{-1}_{\frac{h}{2}}(x_{i};\alpha);\sigma)-v(\varphi^{[1]}_{\frac{h}{2}}(x_{i-1};\alpha);\sigma)\right)^2\right.\\
&&\label{eqn:ST loglikelihood}\left.+\log |g(\varphi^{-1}_{\frac{h}{2}}(x_{i};\alpha);\sigma)|-\log\left|\frac{d\varphi^{-1}_{\frac{h}{2}}(y;\alpha)}{dy}\right|_{y=x_{i}}\right],
\end{eqnarray}
where the last term in \eqref{eqn:LT loglikelihood} and the last two terms in \eqref{eqn:ST loglikelihood} are the effect of the nonlinear nature of the transformation.
Here, we use the notations $v(x;\sigma)$ and $\varphi^{[1]}(x;\alpha)$ to emphasise the dependence of the function $v$ and $\varphi^{[1]}$ on $\sigma$ and $\alpha$, respectively, and throughout the paper, denote $\varphi^{-1}_h:=(\varphi^{[1]}_h)^{-1}$ to simplify the notation.

If $v^{-1}$ is nonmonotone, the LT negative pseudo-log-likelihood becomes 
\begin{equation}\label{logLT2}
    l_h^{\textrm{LT}}(\underline{x};\alpha)=\sum_{i=1}^{N}\sum_{j=1}^{n(x_i)}\left[
    \frac{1}{2h}\left((v^{-1})_j^{-1}(x_{i};\sigma)-v(\varphi^{[1]}_h(x_{{i-1}};\alpha);\sigma)\right)^2-\log \left|\frac{d}{dy}(v^{-1})_j^{-1}(x_{i};\sigma)\right|\right],
    \end{equation}
    where $(v^{-1})^{-1}_j(x_i), j=1,\ldots, n(x_i)$ are the solutions to $v^{-1}(y)=x_i$. A similar expression can be derived for the S case when the inverse of $(\varphi_{\frac{h}{2}}^{[1]}\circ v^{-1})$ is nonmonotone.

Then, the pseudo maximum likelihood estimates are obtained as 
\begin{equation}\label{eqn:MLE estimator}
\hat\alpha^{\bullet}_N( \underline x):=\underset{\alpha}{\arg \min}\   l_h^\bullet(\underline x;\alpha),\qquad \bullet= LT, S.
\end{equation}
\subsubsection{Kessler estimator} 
In \cite{kessler_estimation_1997}, Kessler proposed a contrast estimator based on the idea of maximising the pseudo-log-likelihood obtained  approximating the unknown transition density with a normal distribution characterised by the exact first and second one-step conditional moments given by
\begin{equation}\label{eqn:moment functions}
    m(X_{t_{k-1}};\alpha)=\mathbb{E}_{\alpha}\left[X_{t_{k}}\mid X_{t_{k-1}}\right],\quad m_2(X_{t_{k-1}};\alpha)=\mathbb{E}_{\alpha}\left[\left(X_{t_{k}}-m(X_{t_{k-1}};\alpha)\right)^2\mid X_{t_{k-1}}\right].
\end{equation}
This approach preserves the computational gain and convenience of having a Gaussian transition density, as for the EuM numerical scheme, but it improves on it by replacing the numerically approximated conditional moments with the true ones. When such moments are not available, they can be approximated by the Ito-Taylor expansion-based method introduced in \cite{florens-zmirou_approximate_1989}.

\subsubsection{Hermite estimator}\label{SecHermite} 
 $\text{Aït-Sahalia}$'s method \cite{ait-sahalia_maximum_2002}
 is to first use first the Lamperti transformation \eqref{eqn:Lamperti transform} to transform the SDE $X_t$ to $Y_t$ in \eqref{Lamperti}. Then, the normalised density function of $Z=\frac{Y_{t_{k+1}}-Y_{t_k}}{\sqrt{h}}$, conditionally on $Y_{t_k}=y_0$ for any time step $h:=t_{k+1}-t_{k}$,  is given by
\begin{equation*}
    f_{Z}(z\mid y_0;\alpha )=\phi(z)\sum_{j=0}^\infty c_j(y_0,h;\alpha)H_j(z),
\end{equation*}
where $\phi(\cdot)$ denotes the standard Gaussian density, $H_j(\cdot)$ is the Hermite polynomial of order $j$ defined by
\begin{equation*}
    H_j(z)=(-1)^j e^{z^2 / 2} \frac{d^j}{d z^j}\left(e^{-z^2 / 2}\right),
\end{equation*}
and  
\begin{equation*}
   c_j\left(y_0,h;
\alpha\right):=\frac{1}{j!}\int_\mathbb{R} H_j(z)f_{Z}(z|y_0,\alpha)dz,
\end{equation*}
where both $c_j$ and $f_Z$ can be approximated via suitable expansion/truncation. By using the Jacobian formula for the transition density of $Y_{t_{k+1}}$ given $Y_{t_k}$ first, and then of $X_{t_{k+1}}$ given $X_{t_k}$, we obtain
\begin{eqnarray}
\label{Hermite}
f^H_{X_{t_{k+1}}|X_{t_k}}(x|x_0;\alpha)&=&\frac{1}{g(x;\sigma)}f^{\textrm{H}}_{Y_{t_{k+1}}|Y_{t_k}}(v(x;\sigma)|v(x_0;\sigma);\alpha),
\end{eqnarray}
with 
\[
f^{\textrm{H}}_{Y_{t_{k+1}}|Y_{t_k}}(y|y_0;\alpha)=h^{-1}f_Z\left(\frac{1}{\sqrt{h}}(y-y_0)|y_0;\alpha\right),
\]
with \eqref{Hermite} which can be used to derive the corresponding pseudo-likelihood.
\begin{remark}\label{remarkLamperti}
If the true transition density of $Y_{t_{k+1}}|Y_{t_k}$ of the Lamperti SDE \eqref{eqn:Lamperti transform} is used instead of the Hermite expansion, then \eqref{Hermite} will become exact. If instead 
\eqref{eqn:Lamperti transform} is solved using EuM or a different numerical method, e.g. the splitting scheme  proposed in \cite{kelly_adaptive_2023}, then \eqref{Hermite} will lead to an approximated transition density which may be used to derive the corresponding pseudo MLE, as investigated  for the CIR process in Section \ref{sec:numerical results} in what we call \textit{Lamperti plus EuM}, and \textit{Lamperti plus LT splitting}, respectively. Interestingly, both approaches lead to less favourable numerical and statistical results compared to the proposed splitting approach, see  Section \ref{sec:numerical results}.
\end{remark}
\begin{remark}
The LT transition density \eqref{fLT} closely resembles \eqref{Hermite}, with the key differences that the underlying density is known, being Gaussian,  and that it starts from $v(\varphi^{[1]}(x_0))$ rather than $v(x_0)$. Thanks to the Gaussianity, evaluating \eqref{fLT} is straightforward and less computationally costly than \eqref{Hermite}, which is limited by the burden of evaluating/approximating $f_Z$.
\end{remark}
\section{Mean-square convergence and one-step accuracy of the splitting schemes}\label{sec:convergence}
{Here, we first derive a lower bound of the $L^1$ error of the LT scheme and prove its strong mean-square convergence for the 
general \eqref{eqn:LT subequations}-\eqref{eqn:LT subequations2} decomposition in Section \ref{Section3.1}, before deriving the one-step consistency (i.e. one-step conditional mean) of the LT and S schemes for the specific \eqref{eqn:ODE1}-\eqref{eqn:SDE2}  decomposition in Section \ref{Section3.2}. All proofs of the derived numerical results are reported in Suppl. Mat. \ref{sec:appendix A}.} 
\subsection{$L^1$ bound and strong mean-square convergence of the Lie-Trotter splitting}\label{Section3.1} 
As the considered SDE has a H\"older diffusion coefficient, the fundamental strong convergent results based on the numerical one-step consistency by Milstein \cite{milshtejn_approximate_1975,milstein_numerical_1995}, and Tretyakov and Zhang \cite{tretyakov_fundamental_2013} cannot be used. 
To tackle this, we prove the strong mean-square convergence  of the LT scheme (for the decomposition \eqref{eqn:LT subequations}-\eqref{eqn:LT subequations2}) directly \cite{higham_strong_2002}, showing that it  satisfies \begin{equation}\label{eqn:mean-square convergence}
         \lim_{h\rightarrow 0}\mathbb{E}\left[\sup_{0\leq t \leq T}\left|X^{\textrm{LT}}_t-X_t\right|^2\right]=0.
     \end{equation} 
This is challenging,  as it requires deriving the numerical scheme on the whole time interval $[0,T]$ instead of only  analysing its one-step behaviour 
on the time interval $[t_k,t_{k+1})$, as done for SDEs with one-sided Lipschitz coefficients. 
\begin{remark}
The mean-square convergence condition \eqref{eqn:mean-square convergence} involves  $X^{\textrm{LT}}_t$, a continuous-time extension of the discrete numerical solution $X_{t_{k+1}}, k=0,\ldots, n-1$, intuitively obtained by writing the LT solution $X_{t_{k+1}}^{\textrm{LT}}$ in a continuous-time approximation $X_t^{\textrm{LT}}$ for $t\in [t_k,t_{k+1})$ and then back-iterating it to obtain a solution $X^{\textrm{LT}}_t$ on the whole time interval, $t\in[0,T]$, with the continuous-time solution coinciding with the discrete approximation at the gridpoints $\hat t=t_k$. A formal definition and derivation is provided in Suppl. Mat. \ref{sec:appendix A}.
\end{remark} 
We start by stating two results which will be needed in the proof of the mean-square convergence of the LT scheme. The first is a lemma showing the $p$-th moment bound of $X^{LT}_{t \wedge \theta_R}$, for any $t\in [t_k,t_{k+1})$ and $\hat t=t_k$.

\begin{lemma}\label{lemma:LT_onestep_stop}
     Denote $x\wedge y=\min(x,y),  \hat{t}=t_k$ for any $ t\in [t_k,t_{k+1}), k=0,\ldots, n-1$, and define $\theta_R=\inf \left\{t \in[0, T]:\left|X^{LT}_t\right|>R\right.$ or $\left.\left|X^{LT}_{\hat{t}}\right|>R\right\}$, for $R>0$. If Assumption \ref{assump:main} holds, then, for any $p>0$ and $t\in [0,T]$, we have
    \begin{equation}\label{eqn:LT onestep bound}
\mathbb{E}\left[\left|X^{LT}_{t \wedge \theta_R}-X^{LT}_{\widehat{t \wedge \theta_R}}\right|^p\right] \leq C_{R,p} h^{p / 2},
 \end{equation}
which  implies 
\begin{equation}\label{eqn:LT onestep convergence rate}
    \lim_{h\rightarrow 0}\sup _{t \in\left[t_{k}, t_{k+1}\right]} \mathbb{E}\left[\left|X^{LT}_{t \wedge \theta_R}-X^{LT}_{\widehat{t \wedge \theta_R}}\right|^p\right]=O\left(h^{p / 2}\right), 
\end{equation}
where $C_{R,p}$ is a constant depending only on $R$ and $p$, and $O(h)$ denotes functions of order $h$.
\end{lemma}
\noindent We can then use this result to bound the $L^1$ error of the LT scheme.
\begin{proposition}[$L^1$ error bound of LT]\label{prop:LT L1}
Define $\theta_R$ as in Lemma \ref{lemma:LT_onestep_stop}. 
If Assumption \ref{assump:main} holds, then
    \begin{equation}\label{eqn:L1 bound}
        \sup_{0\leq t\leq T}\mathbb{E}\left[\left|X^{\textrm{LT}}_{t\wedge \theta_R}-X_{t\wedge \theta_R}\right|\right]\leq \left(\frac{C_Re_m}{m}+C_Rh^{\frac{1}{2}}+C_Rh^{2}\right)e^{C_R+\frac{C_Rqe_m^{(2-2/q)}}{m}T},
    \end{equation}
    where $e_m=\exp(-m(m+1)/2)$, for any arbitrary $m>1,m\in \mathbb{N}$, and
    $C_R$ is a constant depending only on $R$. 
\end{proposition}
We now state the strong mean-square convergence of the LT scheme. A similar proof leads to the same convergent result for a time-inhomogeneous version of \eqref{eqn:the sde}, as stated in the corollary below.
\begin{theorem}[Mean-square convergence of LT]\label{thrm:LT convergence}
If Assumption \ref{assump:main} holds, and for some $p>2$, there exists a constant $A>0$ such that \begin{equation}\label{assump:LT_moment_bound}
         \mathbb{E}\left(\sup _{0 \leq t \leq T}\left|X^{LT}_t\right|^p\right) \vee \mathbb{E}\left(\sup _{0 \leq t \leq T}\left|X_t\right|^p\right)<A,
     \end{equation}
     then the LT scheme is strong mean-square convergent, i.e. \eqref{eqn:mean-square convergence} holds.
\end{theorem}
\begin{corollary}
The LT scheme for a time-inhomogeneous version of the SDE \eqref{eqn:the sde} decomposed as in \eqref{eqn:LT subequations}-\eqref{eqn:LT subequations2} is  strong mean-square convergent if Assumption \ref{assump:main}, adjusted for time-inhomogeneous functions $f_1(t,x), f_2(t,x)$, and \eqref{assump:LT_moment_bound} hold.
\end{corollary}
\begin{remark}
The mean-square convergence of the S scheme is expected to be of the same order, as higher-order stochastic integrals would be required for achieving higher orders \cite{BerglundLandon2012}, as also observed in our simulation study. 
\end{remark}
\subsection{One-step consistency of the splitting schemes}\label{Section3.2}
While we expect the LT and the S schemes to share the same strong convergence order, their one-step consistency, defined as $\mathbb{E}[X_{t_{k+1}}^\bullet-X_{t_{k+1}}|X_{t_k}]$, is of different order, as stated in the Proposition below, whose proof is given in Suppl. Mat. \ref{sec:appendix A}. This proposition is particularly important for inferential purposes, as discussed in \cite{pilipovic_parameter_2024} for semi-linear SDEs with additive noise, and also observed in our simulation studies in Section \ref{sec:numerical results}.

We now recall a useful lemma 
 (Lemma 1.10 in \cite{kessler_statistical_2012}) that provides the expansion of the one-step conditional moment of $X_t$ which will be needed to prove Proposition \ref{prop:pearson consistency}, as well as Lemma \ref{lemma:Z moment bounds}, Theorem \ref{thrm:asymptotic consistency}, and Theorem \ref{thrm:asymptotic normality}. For sufficiently smooth functions $u:\mathbb{R}\mapsto \mathbb{R}$, the infinitesimal generator of SDE \eqref{eqn:the sde}, denoted by $\mathcal{L}$, is given by
\begin{equation*}
    \mathcal{L}u(x)=f(x)u^{\prime}(x)+\frac{1}{2}g^2(x)u^{\prime \prime}(x),
\end{equation*}
with $\mathcal{D}(\mathcal{L})$ denoting the domain of the generator $\mathcal{L}$, and $\mathcal{L}^k,k\in \mathbb{N}$ the $k$-fold application of $\mathcal{L}$.
\begin{lemma}[Lemma 1.10 in \cite{kessler_statistical_2012}]\label{lemma:expectation expansion}
    Assume $u, \mathcal{L} u, \mathcal{L}^2 u \in \mathcal{D}(\mathcal{L})$, and $\mathcal{L}^3u$ is of polynomial growth. Then,
\begin{equation}\label{eqn:Conditional expansion}
        \mathbb{E}[u(X_{t_{k+1}})|X_{t_k}=x]=u(x)+\mathcal{L}u(x)h+\frac{1}{2}\mathcal{L}^2(u)h^2+O(h^{3}).
    \end{equation}
\end{lemma}
\begin{proposition}[One-step consistency]
\label{prop:pearson consistency}
    Suppose the SDE \eqref{eqn:the sde} satisfies the regularity condition of Lemma \ref{lemma:expectation expansion}.
    Then, the LT \eqref{eqn:LT subsystem} and S \eqref{eqn:Strang subsystem} schemes for the decomposition \eqref{eqn:ODE1}-\eqref{eqn:SDE2}
    satisfy
    \begin{align}
       \left\| \mathbb{E}\left[X^{\textrm{LT}}_{t_{k+1}}-X_{t_{k+1}}\mid X_{t_k}\right]\right\|=O(h^2), \label{eqn:ergodic LT consistency} \\ 
       \left\| \mathbb{E}\left[X^{\textrm{S}}_{t_{k+1}}-X_{t_{k+1}}\mid X_{t_k}\right]\right\|=O(h^3). \label{eqn:ergodic Strang consistency}
    \end{align}
\end{proposition}  

\section{Parameter inference via splitting schemes}\label{sec:likelihood}
In this Section, we state the  consistency and asymptotic normality results of the proposed LT estimator, with proofs reported in Suppl. Mat. \ref{AppendixB}. Similar results hold for the S estimator, with proofs that can be extended from those of LT, as discussed in Remark \ref{remarkS} below.

Differently from the case of semi-linear SDEs with additive noise recently considered in \cite{pilipovic_parameter_2024}, where their LT estimator yielded a Gaussian pseudo-likelihood, here the LT pseudo-likelihood differs from a Gaussian, so we cannot directly apply the canonical asymptotic results from \cite{kessler_estimation_1997}, also because they were derived for globally Lipschitz coefficients. Similarly, even if the non-Gaussianity of the pseudo-likelihood, and the locally Lipschitz drift have been recently tackled in \cite{pilipovic_parameter_2024} for their S estimator, we cannot directly use their proofs either, as our diffusion coefficient is H\"older continuous, and the pseudo-log-likelihood contains terms depending jointly on $(\theta,\sigma)$, which needs to be handled differently.

Let us denote  $l_N(\alpha)=l_N(\theta,\sigma):=l^{\textrm{LT}}_h(\underline{x};\alpha)$
the LT negative pseudo-log-likelihood to explicitly stress the dependence on $N, \theta$ and $\sigma$ for some given observations $\underline{x}$ of length $N$. By looking at \eqref{eqn:LT loglikelihood}, we see that it depends on
\begin{equation}\label{Z}
Z_{t_k}(\theta,\sigma)=v(X_{t_k};\sigma)-v(\varphi^{[1]}_h(X_{t_{k-1}};\alpha);\sigma).
\end{equation}
By adding and subtracting $Z_{t_k}^2(\theta_0,\sigma)$ into $l_N(\alpha)$, we can decompose it as  \begin{equation}\label{eqn:U decomposition}
l_N(\alpha)=\sum_{k=1}^{N}\log |g(x_{t_k};\sigma)|+\frac{1}{2}\left(\tilde U_1+\tilde U_2+2\tilde U_3\right),
    \end{equation}
    where 
    \begin{align*}
        \tilde U_1=&\frac{1}{h} \sum_{k=1}^N Z_{t_k}^2\left(\theta_0,\sigma\right)
        ,\\
        \tilde U_2=&\frac{1}{h} \sum_{k=1}^N \left(v\left(\varphi^{[1]}_h(X_{t_{k-1}};\alpha) ; \sigma \right)-v\left(\varphi^{[1]}_h(X_{t_{k-1}};\theta_0,\sigma) ; \sigma \right)\right)^2        ,\\
        \tilde U_3=&\frac{1}{N h} \sum_{k=1}^NZ_{t_k}(\theta_0,\sigma)\left(v\left(\varphi^{[1]}_h(X_{t_{k-1}};\theta_0,\sigma) ; \sigma\right)-v\left(\varphi^{[1]}_h(X_{t_{k-1}};\alpha) ; \sigma \right)\right).
    \end{align*}
\begin{remark}\label{RemarkZ}
Except for $g(x_{t_k};\sigma)$, all other terms entering into \eqref{eqn:U decomposition} depend on $\alpha=(\theta,\sigma)$ jointly, unlike the Kessler estimator and the Strang estimator from \cite{pilipovic_parameter_2024}, which have separation of drift and diffusion parameter in the pseudo-likelihood.
\end{remark}
For the consistency and asymptotic normality of the LT estimator, it is crucial that $l_N(\alpha)$ in \eqref{eqn:U decomposition}, properly rescaled, converges to the desired quantities. To prove this, we will extensively use the following lemma on moment bounds for $Z_{t_k}(\theta,\sigma)$. 
 
\begin{lemma}\label{lemma:Z moment bounds}
    Suppose $v(x)$ defined by \eqref{eqn:Lamperti transform} and an arbitrary function $u(x;\alpha)$ are sufficiently smooth such that Lemma \ref{lemma:expectation expansion} can be applied. Then, the followings hold:
    \begin{enumerate}
        \item[(i)]
            $\mathbb{E}_{\alpha_0}[Z_{t_k}(\theta_0,\sigma)|X_{t_{k-1}}=x]=\frac{h}{2}g^\prime(x;\sigma)(1-\frac{g(x;\sigma_0)^2}{g(x;\sigma)^2})+\mathcal{R}(h^2;x)$;
        \item[(ii)]
            $\mathbb{E}_{\alpha_0}[Z_{t_k}(\theta_0,\sigma)u(X_{t_k};\alpha)|X_{t_{k-1}}=x]$\\ $=h\left[\frac{1}{2}u(x;\alpha)g^\prime (x;\sigma)\left(1-\frac{g(x;\sigma_0)^2}{g(x;\sigma)^2}\right)+\frac{g(x;\sigma_0)^2}{g(x;\sigma)}u^\prime(x;\alpha)\right]+\mathcal{R}(h^2;x)$;
       \item[(iii)]
            $\mathbb{E}_{\alpha_0}[Z_{t_k}^2(\theta_0,\sigma)|X_{t_{k-1}}=x]=h\frac{g(x;\sigma_0)^2}{g(x;\sigma)^2}+\mathcal{R}(h^2;x)$.
    \end{enumerate}
\end{lemma}
\begin{remark}\label{remarkS}
A similar lemma holds for Strang too, with proofs for this and other asymptotic results resembling those of LT. Indeed, the S transition density is still a non-linear transformation of a Gaussian random variable with a different 
$Z^{\text{S}}(\alpha)=v(\varphi^{-1}_{\frac{h}{2}}(x_{i};\alpha);\sigma)-v(\varphi^{[1]}_{\frac{h}{2}}(x_{i-1};\alpha);\sigma)$, see \eqref{eqn:LT loglikelihood}, with extra terms coming from the Jacobian. Nevertheless, a similar decomposition of $l^{\textrm{S}}_N$ as the one above can be done, with proofs mimicking those of LT with extra assumptions on $\partial_x \varphi^{-1}_h$ and $\partial_\alpha \varphi^{-1}_h$, where $\partial_y $ denotes the partial derivative with respect to $y$. 
\end{remark}
\subsection{Consistency of the estimators}

\begin{theorem}[Consistency of the LT estimator]\label{thrm:asymptotic consistency}
Suppose the SDE \eqref{eqn:the sde} satisfies Assumption \ref{Assumption3} (i)-(iv), and $\Theta$ is compact. Let $\hat{\alpha}_N$ be the estimator defined by \eqref{eqn:MLE estimator}. Then, if $h\to0$ and $Nh\to \infty$, we have
    \begin{equation*}
        \hat{\theta}_N \stackrel{\mathbb{P}_{\alpha_0}}{\to} \theta_0,\quad \hat{\sigma}_N \stackrel{\mathbb{P}_{\alpha_0}}{\to} \sigma_0,
    \end{equation*}
    where $\stackrel{\mathbb{P}}{\to}$ denotes convergence in probability, and
    $\alpha_0$ is the true parameter of the SDE.
\end{theorem}

\subsection{Asymptotic normality of the estimators}
We are now ready to state the asymptotic normality of the splitting pseudo MLEs under Assumption \ref{Assumption3} (i)-(iv), compactness of $\Theta$,  and $\alpha_0$ in the interior of $\Theta$. Let $\rho>0$ and $\mathcal{B}_\rho(\alpha_0)=\{\alpha\in \Theta|\left\|\alpha-\alpha_0\right\|\leq \rho\}$ be a ball around $\alpha_0$ of radius $\rho$. Then, we know that $\mathcal{B}_\rho(\alpha_0)\in \Theta$ for sufficiently small $\rho$, since $\alpha_0$ is in the interior of $\Theta$. For $\hat{\alpha}_N\in \mathcal{B}_\rho(\alpha_0)$, by the definition of the MLE estimator and Taylor's expansion, we have
\begin{equation}\label{eqn:gradient expansion}
    \left(\int_0^1\mathbf{H}_{l_N}\left(\alpha_0+t(\hat{\alpha}_N-\alpha_0)\right)dt\right)(\hat{\alpha}_N-\alpha_0)=-\nabla l_N(\alpha_0),
\end{equation}
where gradient ($\nabla$) and Hessian ($\mathbf{H}$) are the differential operators taken with respect to $\alpha$. Define 
\begin{align}
    \mathbf{C}_N(\alpha):=&\left[\begin{array}{cc}
\frac{1}{N h} \partial_{\theta\theta} l_N(\alpha) & \frac{1}{N \sqrt{h}} \partial_{\theta\sigma} l_N(\alpha) \\
\frac{1}{N \sqrt{h}} \partial_{\sigma\theta} l_N(\alpha) & \frac{1}{N} \partial_{\sigma\sigma} l_N(\alpha)
\end{array}\right],\label{eqn:C_N alpha}\\
\mathbf{S}_N:=&\left[\begin{array}{c}
\sqrt{N h}\left(\hat{\theta}_N-\theta_0\right) \\
\sqrt{N}\left(\hat{\sigma}_N-\sigma_0\right)
\end{array}\right],\quad \mathbf{\lambda}_N:=\left[\begin{array}{c}
-\frac{1}{\sqrt{N h}} \partial_\theta l_N\left(\alpha_0\right) \\
-\frac{1}{\sqrt{N}} \partial_\sigma l_N\left(\alpha_0\right)
\end{array}\right],\label{eqn:S_N lambda_N}
\end{align}
and $D_N:=\int_0^1\mathbf{C}(\alpha_0+t(\hat{\alpha}_N-\alpha_0))dt$. Then, by \eqref{eqn:gradient expansion}, we have $\mathbf{D}_N \mathbf{S}_N=\mathbf{\lambda}_N$. We also define
\begin{equation}\label{eqn:C alpha0}
    \mathbf{C}(\alpha_0):=\left[\begin{array}{cc}
\mathbf{C}_\theta\left(\alpha_0\right) & \mathbf{0}_{d_1 \times d_2} \\
\mathbf{0}_{d_2 \times d_1} & \mathbf{C}_\sigma\left(\alpha_0\right)
\end{array}\right],
\end{equation}
where 
\begin{align*}
    \left[\mathbf{C}_{\theta}(\alpha_0)\right]_{i_1,i_2}:=&\int \frac{\partial_{\theta_i}f(x;\theta_0)\partial_{\theta_j}f(x;\theta_0)}{g(x;\sigma_0)^2}d\mu_0(x),\quad 1\leq i_1,i_2\leq d_1,\\
    \left[\mathbf{C}_{\sigma}(\alpha_0)\right]_{j_1,j_2}:=&2\int\frac{\partial_{\sigma_i} g(x;\sigma_0)\partial_{\sigma_j}g(x;\sigma_0)}{g(x;\sigma_0)^2}d\mu_0(x),\quad 1\leq j_1,j_2\leq d_2.
\end{align*}
\begin{theorem}[Asymptotic normality of the LT estimator]\label{thrm:asymptotic normality}
    Under Assumption \ref{Assumption3}(i)-(iv), if $\Theta$ is compact and $\alpha_0$ is in the interior of $\Theta$, then, for $h\to 0$, $Nh\to \infty$ and $Nh^2\to 0$, under $\mathbb{P}_{\alpha_0}$, we have that
    \begin{equation}
        \left[\begin{array}{l}
\sqrt{N h}\left(\hat{\theta}_N-\theta_0\right) \\
\sqrt{N}\left(\hat{\sigma}_N-\sigma_0\right)
\end{array}\right]\xrightarrow{d}\mathcal{N}(0,\mathbf{C}^{-1}(\alpha_0)).  
    \end{equation}
\end{theorem}

\section{Simulation study}\label{sec:numerical results}
In this section, we test the properties and accuracy of the proposed LT and S splitting schemes and estimators at the numerical and statistical level on a variety of examples belonging to the SDE \eqref{eqn:the sde}, with diffusion coefficients which are either H\"older (such as the CIR, and the Ahn and Gao model), or  Lipschitz (such as the Student diffusion). Further experiments (on all members of the Pearson diffusion class, the stochastic Ginzburg-Landau equation, and the stochastic Verhulst equation) are reported in the Supplementary Material \cite{Fangetal26}. In Table \ref{TablePearson} (and Table S1 in  the Supplementary Material \cite{Fangetal26}), we report some of the quantities needed to perform simulation/parameter inference (e.g. $\varphi_h^{[1]}, \varphi^{[2]}_h, v, v^{-1}, \varphi_h^{[1]}$).
\begin{remark}\label{remarkPearson}
For Pearson diffusion SDEs, we explicitly compute the one-step conditional means $\mathbb{E}[X^\bullet_{t_{k+1}}|X_{t_k}]$, and the  conditional means $\mathbb{E}[X^\bullet_{t}|X_{t_k}], \bullet=\textrm{LT}, \textrm{S}$ 
of the derived LT and S schemes, see Table \ref{TablePearson} and Table S1 in the Supplementary Material \cite{Fangetal26}, with (one-step) conditional variances which can be similarly derived for some models. This may be used in two manners. From a numerical point of view,
by comparing these quantities to the exact ones, we would be able to quantify the bias/error introduced by the numerical schemes, both at finite time and asymptotically as $t\to\infty$, following the approach for (different types of) splitting schemes for the inhomogeneous geometric Brownian motion in \cite{tubikanec_qualitative_2022}. From a statistical point of view, these quantities can be used to perform the generalised method of moments. Both are out of the scope of this manuscript, and  will be considered in future work.
\end{remark}
To investigate the strong mean-square convergence order of our proposed splitting schemes, and compare it with alternative discretisation methods, we estimate the mean-squared error 
\begin{equation}\label{eqn:MSE}
S_N:=\left(\frac{1}{M} \sum_{k=1}^M\left|\widetilde X_{N,k}-X_{N,k}\right|^2\right)^{\frac{1}{2}},
\end{equation}
where $\widetilde X_{N,k}$ and $X_{N,k}$ denote the $k$th realisation of the numerical discretisation scheme and of the true SDE solution at time $T=Nh$, respectively, expressed as a function of the number of discretisation steps $N$ in $[0,T]$, with $M$ being the total number of simulated paths. When an explicit solution of $X$ is unavailable in closed form, $\widetilde X^{\rm{fine}}$ refers to a reference numerical scheme with a time step $h^{\textrm{fine}}$ much smaller than that $h$ for $\widetilde{X}$,  while using the same simulated Brownian paths. Here, we only report an illustration of the mean-square convergent order for the CIR process, referring to the Supplementary Material \cite{Fangetal26} for all other processes.

We compare our proposed LT and S estimators against established benchmarks like Kessler and  Hermite estimators, as well as ad-hoc estimators  derived for some of the considered models, testing their performance with respect to the number of observations $N$, the time horizon $T$, the observation time step $h_\textrm{obs}$, and the computational time to evaluate the underlying pseudo-likelihoods.

\paragraph*{Implementation details} We use the Python package \textit{pymle} \cite{kirkby2025pymle} to calculate the transition densities of Kessler and Hermite estimators. We modify the code for the latter to allow for higher order approximation (up to $j=3$ instead of the implemented $j=1$). We use the \textit{Nelder-Mead} method within the \textit{scipy.optimize} package to perform the MLE optimisation, using $\alpha_0=(1.0,1.0)$ as initial guess for CIR and student diffusion, and $\alpha_0=(5.0,5.0)$ as the initial guess for Ahn-Gao diffusion. For each value of $N$ and $T$, we generate $M=1000$ trajectories with time step $h=h_\textrm{fine}$, and subsample according to an observation time step $h_\textrm{obs}>h$ to obtain  the desired datasets, and visualise the $M$ estimates via density plots and violin plots, obtained using the Python \textit{sns.kdeplot} and \textit{sns.violinplot}, respectively. In all experiments, $X_0=1$. To evaluate the computational cost of the estimators, we report the runtime needed to perform $M=1000$ pseudo MLEs from simulated data up to time $T$ with $h_\textrm{obs}=0.001$. All
experiments are performed using the Warwick
University HPC facilities on Avon, with
Dell PowerEdge C6420 compute nodes each with 2 x Intel Xeon Platinum 8268 (Cascade Lake) 2.9 GHz 24-core processors; 48 cores per node; 180 nodes; 8640 compute cores; 192 GB DDR4-2933 RAM per node; 4 GB RAM per core. 

\begin{table}
\tiny
    \centering
    \begin{tabular}{ccc}
    \toprule
     & {\bf CIR - $g(x)=\sqrt{2\theta b x}$ } & {\bf Student diffusion  -  $g(x)=\sqrt{2a\theta(x^2+1)}$} \\ \midrule
$(\tilde \mu,\tilde\theta)$ & $(\mu-\frac{1}{2}b, \bullet )$ & $(\frac{\mu}{1+a},(1+a)\theta)$, \\ \midrule
Extra&$A^{\textrm{LT}}=1, \quad A^{\textrm{S}}=e^{-\theta h/2}$&$A^{\textrm{LT}}=\tilde\mu(e^{\theta ah}-e^{-\theta} h)$
    \\
    &&   $A^{\textrm{S}}=\tilde\mu\left[1-e^{-\theta h}-e^{-\frac{\tilde{\theta}h}{2}}(1-e^{\theta ah})\right]$
\\ \midrule
$X^{[2]}_{t_{k+1}}=\varphi^{[2]}_{h}(x)$&$\left(\sqrt{x}+\frac{\sqrt{2\theta}}{2}
\xi_k\right)^2$&$\sinh\left(\sqrt{2\theta a}{\xi_k}+\operatorname{arcsinh}(X_{t_k})\right)$ \\ \midrule
$\mathbb{E}[X^{\bullet}_{t_{k+1}}|X_{t_k}]$&$\tilde{\mu}+e^{-\theta h}\left(X_{t_k}-\tilde{\mu}\right)+\frac{\theta bh}{2}A^{\bullet}$&$e^{-\theta h}X_{t_k}+A^{\bullet}$\\ \midrule
$\mathbb{E}[X^{\bullet}_t|X_{0}]$& $e^{-\theta t}X_0+\tilde{\mu}\left(1-e^{-\theta t}\right)+A^{\bullet}\frac{\theta h(1-e^{-\theta t})}{2(1-e^{-\theta h})}$ &$e^{-\theta t}X_0+A^{\bullet}\frac{1-e^{-\theta t}}{1-e^{-\theta h}}$ \\ \midrule
$v(x); v^{-1}(x)$& $\sqrt{\frac{2x}{\theta b}}; \quad \frac{\theta b}{2}x^2 $&$\frac{1}{\sqrt{2\theta a}}\operatorname{arcsinh}(x); \quad \sinh(\sqrt{2\theta a}x)$\\ \midrule
\bottomrule
    \end{tabular}
\vspace{.5cm}
      \begin{tabular}{ccc}
    \toprule
    & {\bf Ahn and Gao model } - $g(x)=\sigma x^{3/2}$\\
    Extra & $X^{[1]}_{t_{k+1}}=\varphi_h^{[1]}(x)$&$X^{[2]}_{t_{k+1}}=\varphi_h^{[2]}(x)$\\ \midrule
   $A=\kappa \theta,B=\kappa+\frac{3}{4}\sigma^2$& $\frac{Ae^{Ah}X_{t_k}}{A+B X_{t_k}(e^{Ah}-1)}$&$\frac{4}{(\sigma\xi_k-\frac{2}{\sqrt{x}})^2}$\\ \midrule
   $v(x); v^{-1}(x)$& $\varphi_h^{-1}(x)$ &$\partial x\varphi^{-1}(x)$\\ \midrule
   $-\frac{2}{\sigma\sqrt{x}}; \frac{4}{\sigma^2x^2}$&$\frac{Ax}{Ae^{Ah}+xB(1-e^{Ah})}$& $\frac{A(Ae^{Ah}+xB(1-e^{Ah}))-B(1-e^{Ah})Ax}{(Ae^{Ah}+xB(1-e^{Ah}))^2}$\\
    \bottomrule
    \end{tabular}     \caption{Characterisation of the derived LT  \eqref{eqn:LT subsystem} and S \eqref{eqn:Strang subsystem} splitting schemes for the three models considered here (more in the Supplementary Material \cite{Fangetal26}), with $\varphi^{[1]}_h$ and $\varphi^{[2]}_h$ representing the one-step solutions to the subequations \eqref{eqn:ODE1}-\eqref{eqn:SDE2}, and $\varphi^{-1}_h$ denoting the inverse of $\varphi^{[1]}_h$. The one-step conditional means are also reported, together with the whole conditional mean. Here, $\xi_k\sim N(0,h), k=0,\ldots, N-1$ are iid normal increments. Note that \eqref{eqn:LT loglikelihood} and \eqref{eqn:ST loglikelihood} can only be used for the Student diffusion.}
    \label{TablePearson}
\end{table}

\subsection{Pearson diffusion class}\label{resultsPearson}
The Pearson diffusion class includes solutions to the following SDE
\begin{equation*}\label{eqn:Pearson diffusion}
d X_t=-\theta\left(X_t-\mu\right) d t+\sqrt{2 \theta\left(a X_t^2+b X_t+c\right)} d W_t,
\end{equation*}
where $\theta>0$, and $a,b,c$ are such that the square root is well defined when $X_t$ is in the state space \cite{forman_pearson_2008}. 
Our derived 
splitting schemes \eqref{eqn:LT subsystem} and \eqref{eqn:Strang subsystem} work particularly well for this class, as 
    $g(x)g^{\prime}(x)/2=\theta( a x+ b/2)$, 
so the decomposition \eqref{eqn:ODE1}-\eqref{eqn:SDE2} drastically simplifies, with $f_1$ and $f_2$ given by
\[
f_1(x)=-\tilde\theta(x-\tilde\mu),\quad f_2(x)=\theta (ax+b/2), \quad \tilde\theta=\theta(a+1), \quad \tilde\mu=(\mu-b/2)/(a+1),
\]
trivially fulfilling Assumption \ref{assump:main}. In particular, \eqref{eqn:ODE1} becomes a first-order linear ODE with one-step solution starting from $X_{t_k}$ given by 
\begin{equation}\label{ODEsolPearson}
X^{[1]}_{t_{k+1}}=\varphi^{[1]}_h(X_{t_k})=\tilde\mu+e^{-\tilde\theta h}(X_{t_k}-\tilde\mu),
\end{equation}
while \eqref{eqn:SDE2} becomes an SDE with linear drift, with solution $\varphi^{[2]}_t$ depending on $g(x)$. This allows us to compute the one-step conditional means of the derived schemes (as well as one-step conditional variances, results not shown), and, for all but the $F$-diffusion process, their conditional means, see Remark \ref{remarkPearson}. Moreover, the following result holds, see the Supplementary Material \cite{Fangetal26} for the proof for each Pearson diffusion model.
 \begin{proposition}
The derived \textrm{LT} and \textrm{S} splitting schemes preserve the state space $\mathcal{X}=[\alpha,\beta]$, $\alpha,\beta \in [-\infty,\infty]$ of all processes belonging to the Pearson diffusion class when the boundaries $\alpha$ and $\beta$ are of entrance type (and thus not attainable when $X_0\in\mathcal{X}$).  
 \end{proposition}
Adaptive time steps may be be considered for attainable  (i.e. reachable with positive probability) boundaries, as discussed below. In the following, we focus on the CIR and the Student diffusion.
\subsubsection{Cox-Ingersoll-Ross (CIR) process - $g(x)=\sqrt{2\theta b x}$}

The CIR is probably one of the most commonly used models when proposing a numerical scheme for SDEs with H\"older diffusion coefficients, as it can be simulated exactly and has a non-negative state space, whose numerical preservation is of interest, as canonical Ito-Taylor schemes (such as EuM and Milstein) do not preserve it \cite{kelly_adaptive_2023}. 

In particular, the CIR state space is $[0,\infty)$, with the boundary $0$ being of entrance type for $\mu\geq b$, and attainable for $\mu\in(0,b)$, see \cite{forman_pearson_2008}.  The transition density is known to be 
\begin{equation}\label{fCIR}
X_t|X_0 = \frac{(1-e^{-\theta t})bY}{2},\qquad Y\sim \chi^2\left(2\frac{\mu}{b},\frac{2}{(1-e^{-\theta t})b}X_0e^{-\theta t}\right),
\end{equation}
where $\chi^2(x,y)$ denotes a non-central $\chi^2$ distribution with $x$ degrees of freedom and non-central parameter $y$. 
It also admits a unique ergodic solution  when $\mu>0$, with invariant Gamma distribution with shape parameter $\mu/b$ and rate parameter $1/b$. Using \eqref{ODEsolPearson} and the subequation' solution $\varphi^{[2]}_h$ reported in Table \ref{TablePearson}, the LT  and S compositions are then
\begin{eqnarray}
\label{eqn:CIR LT}X^{\textrm{LT}}_{t_{k+1}}&=& \left(\sqrt{\tilde\mu+e^{-\theta h}(X_{t_k}-\tilde\mu)}+\frac{\sqrt{2\theta b}}{2}\xi_k\right)^2,\\
X^{\textrm{S}}_{t_{k+1}}\label{eqn:CIR Strang}&=& \tilde\mu +e^{-\theta h/2}\left(\left(\sqrt{\tilde\mu+e^{-\theta h/2}(X_{t_k}-\tilde\mu)}+\frac{\sqrt{2\theta b}}{2}\xi_k\right)^2-\tilde \mu\right),
\end{eqnarray}
leading to the one-step $\mathbb{E}[X^\bullet_{t_{k+1}}|X_{t_k}]$ and  global conditional means $\mathbb{E}[X^\bullet_{t}|X_{0}], \bullet=\textrm{LT}, \textrm{S}$ reported in Table \ref{TablePearson}. From \eqref{eqn:CIR LT} and \eqref{eqn:CIR Strang}, we immediately see that both schemes are positive-preserving when $0$ is an entrance boundary, i.e. 
$X_{t_k}>0$ and $\mu\geq b$, as then $\tilde\mu>0$.  When 0 is attainable, i.e. $\mu<b$, the $X^{[1]}_{t_{k+1}}$ solution may be negative, making thus the LT and S schemes not defined. A remedy for this is using an adaptive step size: given $X_{t_k}$, choosing
\begin{equation*}
    h^{\textrm{LT}}_{t_k}=\min\left(h,-\frac{1}{\theta}\ln{\frac{-\tilde{\mu}}{X_{t_k}-\tilde{\mu}}}\right),\qquad  h^{\textrm{S}}_{t_k}=\min\left(h,-\frac{2}{\theta}\ln{\frac{-\tilde{\mu}}{X_{t_k}-\tilde{\mu}}}\right)
\end{equation*}
guarantees that $X^{[1]}_{t_{k+1}}$ stays non-negative, leading to well defined schemes.

In Figure \ref{fig:CIR strong convergence}, we illustrate the mean-squared error \eqref{eqn:MSE} of our splitting schemes, and compare it with other competing numerical methods, such as the commonly used EuM and Milstein schemes \cite{kloeden_numerical_1992},  the SD scheme \cite{halidias_semi-discrete_2012}, the Lamperti plus EuM and Lamperti plus LT splitting approach \cite{kelly_adaptive_2023},  see Remark \ref{remarkLamperti}. We fix $\mu=6.0, \theta=2.0, b=0.2, T=1.0,  h^{\textrm{fine}}=2^{-13}$, and evaluate $S_N$ for different $h$ on $M=1000$ simulated paths. All schemes but EuM achieve a strong order one convergence rate, with S outperforming them, also for different parameter values (figures nor reported). This is also observed when computing the Wasserstein distance between the true invariant density and that simulated with the numerical schemes, with the lowest distance being obtained for the S splitting scheme, followed by LT and Lamperti plus LT (see the Supplementary Material \cite{Fangetal26}).

\begin{figure}
    \centering    \includegraphics[width=.9\textwidth]{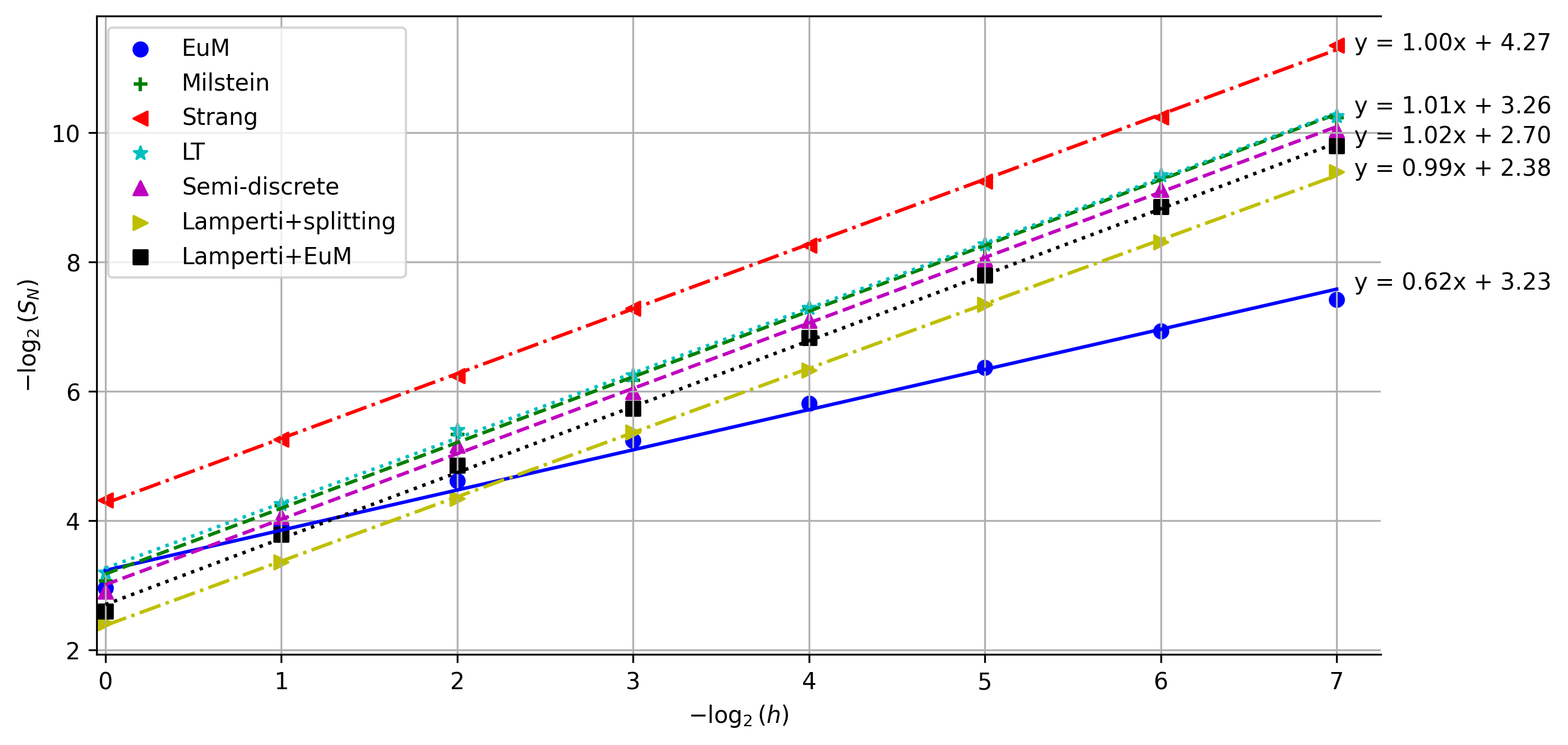}
    \caption{Illustration of the mean-square convergence order on the CIR process via the mean-squared error \eqref{eqn:MSE}, with $\mu=6.0,\theta=2.0,b=0.2, T=1.0$. 
   }
    \label{fig:CIR strong convergence}
\end{figure}

As $v^{-1}$ is not monotone, we cannot use \eqref{eqn:LT loglikelihood} and \eqref{eqn:ST loglikelihood}. Instead, we derive the log-likelihoods directly. The conditional distribution of the LT scheme \eqref{eqn:LT subsystem} is given by
\begin{equation}
    \mathbb{P}(X^{\textrm{LT}}_{t_{k+1}}\leq y|X_{t_{k}}=x)=\mathbb{P}(Y^2\leq y)=\mathbb{P}(-\sqrt{y}\leq Y\leq \sqrt{y}),
\end{equation}
with $Y\sim \mathcal{N}(\sqrt{\tilde\mu+e^{-\theta h}(x-\tilde\mu)},\theta bh/2)$, leading to the LT transition density
\begin{equation*}
f^{\textrm{LT}}_{X_{t_{k+1}}|X_{t_k}}(y|x)=\frac{1}{2\sqrt{y}}\left(f_{Y}(\sqrt{y})+f_{Y}(-\sqrt{y})\right),
\end{equation*}
where $f_Y(y)$ denotes the probability density function of $Y$ evaluated in $y$. 
For the S scheme \eqref{eqn:CIR Strang}, denote $X^{[1]}_{h/2}=\tilde{\mu}+e^{\frac{-\theta h}{2}}(X_{t_k}-\tilde{\mu})$, and  $Y^{\textrm{S}}\sim \mathcal{N}(\sqrt{X^{[1]}_{h/2}},\theta bh/2)$. 
Then, the S transition density becomes
\begin{equation*}
    f^{\textrm{S}}_{X_{t_{k+1}}|X_{t_k}}(y|x)=\frac{e^{\frac{\theta h}{2}}\left(f_{Y^{\textrm{S}}}(\sqrt{\tilde{\mu}+e^{\frac{\theta h}{2}}(y-\tilde{\mu})})+f_{Y^{\textrm{S}}}(-\sqrt{\tilde{\mu}+e^{\frac{\theta h}{2}}(y-\tilde{\mu})})\right)}{2\sqrt{e^{\frac{\theta h}{2}}(y-\tilde{\mu})+\tilde{\mu}}}.
\end{equation*}

To evaluate the performance of our splitting estimators, we fix $\theta$, and infer the unknown parameter $(\mu,b)$ from the $1000$ paths simulated from the true CIR transition density with simulation time step $h^{\textrm{fine}}=10^{-4}$ and then subsampled according to different observation time steps $h_\textrm{obs}$, and thus different number of observations $N$. 
In Figure \ref{fig:CIR inference}, we report the MLE densities obtained via our splitting estimators, the true MLE and several competing estimators, such as Kessler, Hermite, Lamperti plus splitting and Lamperti plus EuM. The pseudo-likelihoods for the Lamperti-based estimators were derived similarly to the LT and S cases, see Remark \ref{remarkLamperti}.   

Figure \ref{fig:CIR inference} shows that our
splitting estimators are robust with respect to the observation time step $h_\textrm{obs}$, performing satisfactorily (especially for $b$, and particularly S) with $h_\textrm{obs}$ as large as $0.5$, when the other approaches fail. As observed in other classes of SDEs, and expected from the higher one-step consistency order, the S estimator performs better than LT, especially for estimating $b$. All estimators lead to almost indistinguishable inference (especially for $\mu$) as $h_\textrm{obs}$ decreases, closely resembling  the true MLE. Note how considering a higher order Hermite expansion (j=3) does not lead to a higher accuracy due to the introduced numerical instability.

In terms of computation cost, on the cluster, it takes 39.72s to perform 1000 pseudo MLEs using LT, followed by S (40.37s), Lamperti plus splitting (43.28s), Lamperti plus EuM (46.17s), Kessler (168.02s), true density (183.70s), Hermite (1327.6s, j=2), and  Hermite (35284.35s, j=3).
\begin{figure}
    \centering    \includegraphics[width=.9\linewidth]{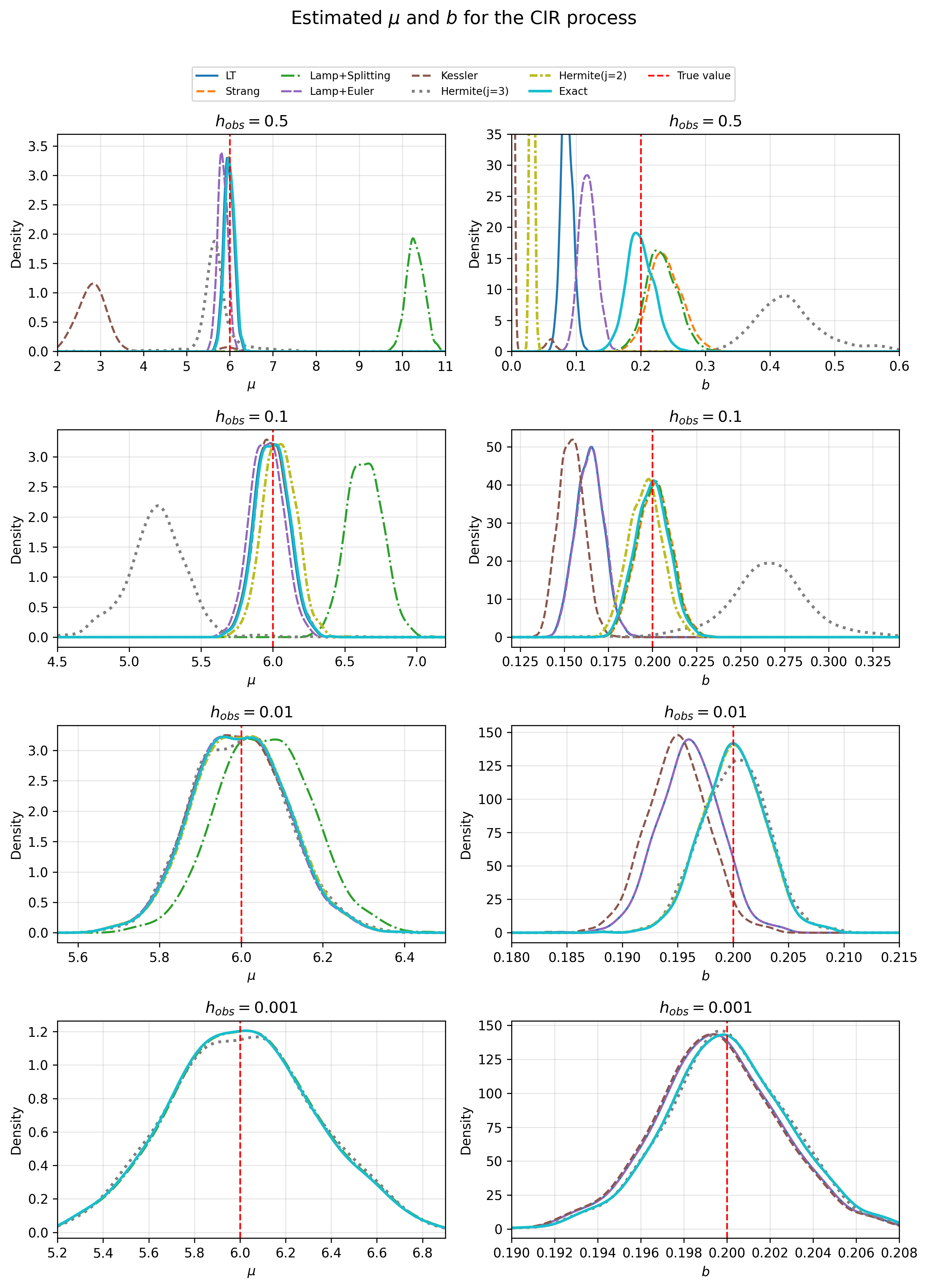}
    \caption{
    Inference results for the CIR process: densities of the pseudo MLEs derived under different estimators. We fix $\theta=2.0$ and infer the unknown $\mu$ (left column) and $b$ (right column) using $M=1000$ simulated paths,  generated exactly using $h^{\textrm{fine}}=10^{-3}$, and then subsampled according to the  observation time step $h_\textrm{obs}$, here $0.5, 0.1, 0.01, 0.001$, from top to bottom.}
    \label{fig:CIR inference}
\end{figure}

\subsubsection{Student diffusion - $g(x)=\sqrt{2\theta a(x^2+1)},  a>0$}
We consider this process because it is the only one within the Pearson diffusion class to have monotone inverses of $v^{-1}$ and $(\varphi_{h/2}^{[1]}\circ v^{-1})$, allowing the use of \eqref{eqn:LT loglikelihood} and \eqref{eqn:ST loglikelihood} as negative pseudo-log-likelihoods. This process is defined in $\mathbb{R}$, and admits a unique ergodic solution for all $a>0$. 
If $\mu=0$, the invariant distribution is a scaled $t$-distribution with $\nu=1+1/a$ degrees of freedom and scaled parameter $\nu^{-\frac{1}{2}}$, while if $\mu \neq 0$, the invariant distribution is a skew $t$-distribution, also known as  Pearson's type IV distribution \cite{forman_pearson_2008}. Using \eqref{ODEsolPearson} and $\varphi^{[2]}_h$ reported in Table \ref{TablePearson}, we immediately obtain the LT and S splitting schemes
\begin{align*}
X^{\textrm{LT}}_{t_{k+1}}=&\sinh\left(\sqrt{2\theta a }
\xi_k+\operatorname{arcsinh}\left(\tilde{\mu}+e^{-\tilde{\theta}h}\left(X^{\textrm{LT}}_{t_k}-\tilde{\mu}\right)\right)\right),\\
    X^{\textrm{S}}_{t_{k+1}}=&\tilde{\mu}+e^{-\frac{\tilde{\theta}h}{2}}\left(\sinh\left(\sqrt{2\theta a }\xi_k+\operatorname{arcsinh}\left(\tilde{\mu}+e^{-\frac{\tilde{\theta}h}{2}}\left(X^{\textrm{S}}_{t_k}-\tilde{\mu}\right)\right)\right)-\tilde{\mu}\right),
\end{align*}
with one-step conditional means given in Table \ref{TablePearson}.
In our experiment, we fix $\theta=2.0$ and $T=100.0$, and infer the unknown parameter $(\mu,b)$ from the $1000$ paths simulated with $h^\textrm{fine}=10^{-5}$ and $\widetilde X^{\textrm{fine}}$ being the SD scheme, using different level of observation time step $h_\textrm{obs}$. As shown in Figure \ref{fig:T inference}, the Hermite estimators (both $j=2$ and $3$) and the Kessler estimator \lq\lq collapse\rq\rq\ for all but $h_\textrm{obs}=0.001$, with notably less accurate inference though. Meanwhile, the two splitting estimators and the Lamperti plus EuM estimator (obtained by solving the underlying Lamperti SDE \eqref{Lamperti} with EuM) can be used even for larger observation time steps (including $h_\textrm{obs}=1$, figures not shown), with the S estimator performing better overall, followed by LT and Lamperti plus EuM. 

In terms of computation costs, on the cluster, it takes 70.16 seconds to perform 1000 pseudo MLEs  paths using LT estimator, followed by S (72.02s), Lamperti plus EuM (86.16s), Kessler (246.64s), Hermite (264.05s, j=3), Hermite (272.51s, j=2).

\begin{figure}
    \centering
\includegraphics[scale=.5]{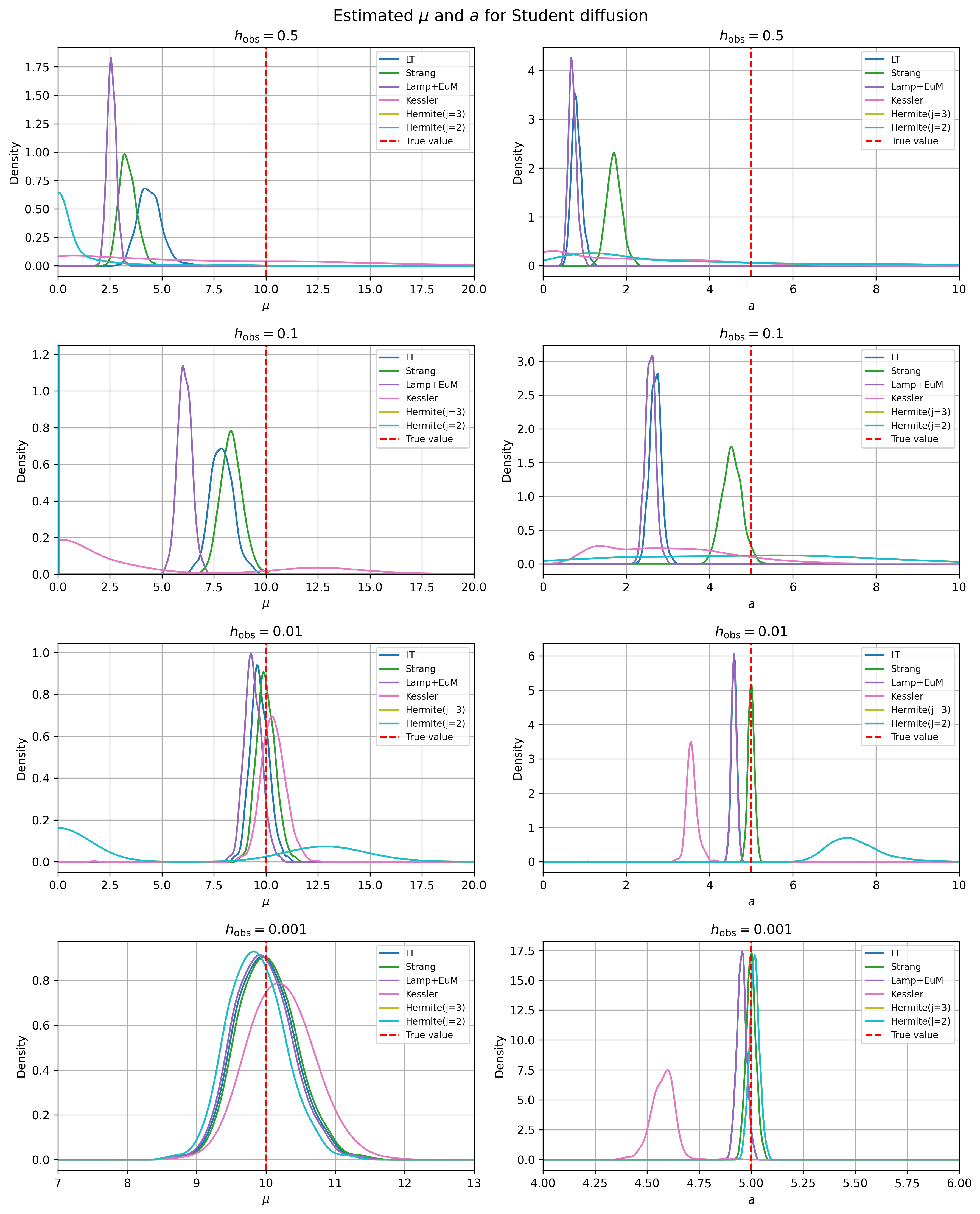}
    \caption{Inference results for the Student diffusion: densities of the pseudo MLEs derived under different estimators. We fix $\theta=2.0$, and infer the unknown $\mu$ (left column) and $a$ (right column) using $M=1000$ paths generated using $h^{\textrm{fine}}=10^{-5}$ with $\widetilde X^{\textrm{fine}}$ being the SD scheme, and then subsampled according to the  observation time step $h_\textrm{obs}$, here $1.0,  0.5, 0.1, 0.01, 0.001$, from top to bottom. }
    \label{fig:T inference}
\end{figure}
\subsection{Ahn and Gao model}
Finally, we consider the Ahn and Gao model \cite{ahn_parametric_1999}, which solves the following SDE
\begin{equation}\label{eqn:AG}
    \mathrm{d} X_t=\kappa(\theta-X_t)X_t\mathrm{~d}t+\sigma X_t^{\frac{3}{2}} \mathrm{~d} W_t,
\end{equation}
with $\kappa,\theta,\sigma>0$ to guarantee that the state space is positive, and the boundary $0$ is of entrance type. This model is a generalisation of the CIR process, with an explicit transition density given by \cite{iacus_simulation_2008}
\[
f_{X_t|X_0}(y|x_0)=\frac{1}{y^2}f^{\textrm{CIR}}_{X_t|X_0}\left(\frac{1}{y},\frac{1}{x_0}\right),
\]
where  $f^{\textrm{CIR}}_{X_t|X_0}$ denotes the true transition density of the CIR given in \eqref{fCIR}. The desired LT \eqref{eqn:LT subsystem} and S \eqref{eqn:Strang subsystem} schemes are obtained using the subequations' solutions reported in Table \ref{TablePearson}. From there, we see that $X^{[1]}_{t_{k+1}}$ is always positive, and thus the $X^{[2]}_{t_{k+1}}$ is well defined, and the state space is preserved by both splitting schemes. For the inference part, we fix $\theta=2.0, \sigma=0.5, \kappa=0.2, h^{\textrm{fine}}=10^{-3}, T=100.0$, and simulate exactly via the true transition density. We then fix $\kappa$ and infer the unknown parameter $(\theta,\sigma)$ for different observation times steps $h_{\textrm{obs}}=0.001, 0.01, 0.1$ using subsampling. Among all estimators, the LT is the most robust (with respect to the observation time step) and accurate (compared to the true MLE). Moreover, for the first time, S estimator does not yield satisfactory inference, collapsing for $h_\textrm{obs}=1.0$, and returning estimates which either centred around a wrong value (see e.g. the mass around $\theta=5$) for $h_\textrm{obs}=0.5$. While the S scheme is well defined from a numerical point of view, outperforming alternative schemes in terms of mean-square error, the inverse of $\varphi^{[1]}\circ v^{-1}$, and thus the S transition density, may not be well defined for the observation time step $h_\textrm{obs}$ from a statistical point of view. The lack of invertibility was also encountered for partially observed SDEs with locally Lipschitz drift and constant diffusion coefficients in \cite{Huangetal2025}, where the authors proposed using bridge processes (dealt in the underlying pseudo-likelihood via controlled sequential Monte Carlo) to obtain discretisation steps small enough to make the quantity invertible. Alternatively, in order to take advantage of the numerical robustness of S, one could perform simulation-based inference based on it, bypassing the invertibility issues. It is important to stress that all estimators but LT lead to not-well-defined transition densities, especially for large time steps. This can be observed in Figure \ref{fig:AG heatmap}, where we report the one-step transition density heatmaps $f_{X_{t_{k+1}}|X_{t_k}}(y|x)$ for $x,y\in[0,5]$, rather than only likelihood values. For coarse observations, $h=1.0$ and $h=0.5$, the approximate densities differ substantially from the true density, expect LT. In particular, S and Hermite produce large non-evaluable regions, while Kessler places probability mass in visibly distorted regions. As $h$ decreases, the LT and S heatmaps become much closer to the true one, whereas the other approximations still show noticeable support and shape artefacts. This illustrates how similar inference performance may mask substantial off-path differences in the transition density approximation.

Finally, the runtimes for performing 1000 pseudo MLEs are 63.4s (LT), 109.43s (true MLE), 121.28s (S), 134.23s (Kessler), 238.79s (Hermite, j=2), and 239.28s (Hermite, j=3). 

\begin{figure}[t]
    \centering
    \makebox[\textwidth][c]{        \includegraphics[width=0.8\textwidth]{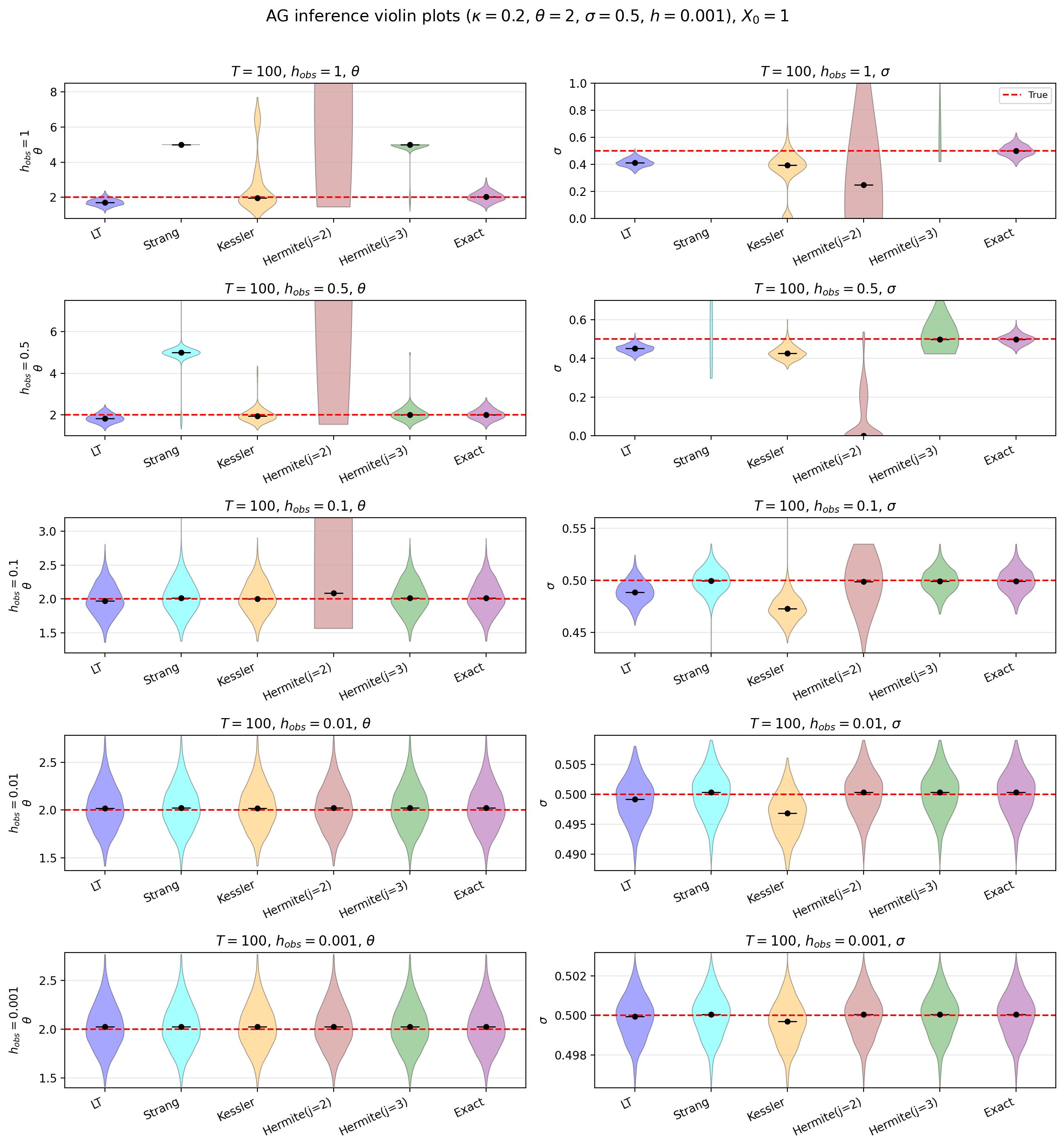}
    }
    \caption{Inference result for the Ahn-Gao model: violin plots of the estimates derived under different estimators. We fix $\kappa=0.2$, and infer the unknown $\theta$ (left columns) and $\sigma$ (right columns) using $M=1000$ simulated paths,  generated exactly using $h^{\textrm{fine}}=10^{-3}$, and then subsampled according to the  observation time step $h_\textrm{obs}$, here from $1.0$ to $0.001$.}
    \label{fig:AG-inference}
\end{figure}

\begin{figure}
    \centering
    \includegraphics[width=1.0\linewidth]{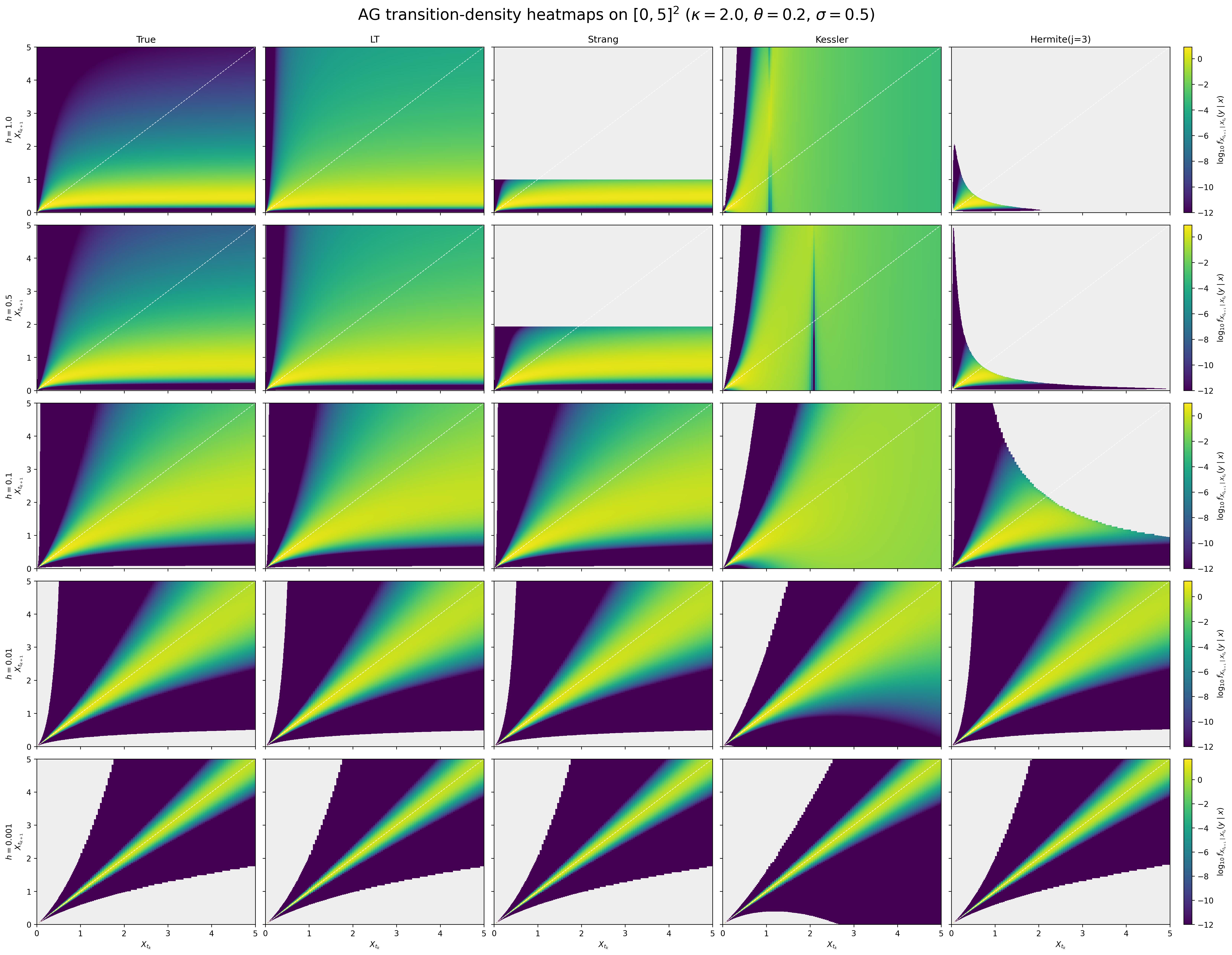}
    \caption{
One-step transition density heatmaps of $\log_{10} f_{X_{t_{k+1}}|X_{t_k}}(y|x)$ for the AG model on $(x,y)\in [0,5]^2$ with
$(\kappa,\theta,\sigma)=(2.0,0.2,0.5)$. The columns compare the true density with the LT, Strang, Kessler, and Hermite approximations (left to right), while rows reports the heatmaps for different  $h=1.0,0.5,0.1,0.01$ (top to bottom). Grey regions
indicate non-positive or non-evaluable density values, and the dashed line marks
$y=x$.
}
    \label{fig:AG heatmap}
\end{figure}
\section{Conclusion}\label{sec:conclusion}
In this paper, we tackled both numerical and statistical challenges by deriving, for the first time, mean-square convergent splitting schemes for scalar SDEs with locally Lipscthiz drift and H\"older diffusion coefficient, leading to consistent and asymptotically normal estimators (with rate of convergence $\sqrt{N}$ for diffusion parameters and $\sqrt{Nh}$ for drift parameters). The proposed splitting schemes are obtained via an innovative SDE decomposition based on  reducible SDEs, allowing the exact solution of the underlying SDE subequation with multiplicative noise: this was traditionally deemed to be the hardest task for splitting schemes, something  bypassed by applying them on the Lamperti-transformed SDEs in the very few existing results.

From a numerical point of view, the mean-square convergent splitting schemes derived here are easy to construct,  preserve the state space of the models and are shown, via numerical simulations, to be more robust than competing discretisation schemes for large time steps $h$. All features are highly desirable for simulation and modelling purposes, e.g., when they are  used to investigate the model properties, or embedded in simulation-based approaches. From a statistical point of view,  we showed, via an extensive simulation study, that the proposed splitting estimators perform better than state-of-art alternatives, such as Kessler or Hermite, especially for large time steps. As expected from the proved higher one-step accuracy, the S estimator performed better than LT when available, that is, when the underlying function mapping a Gaussian distribution to the distribution of the one-step numerical S scheme is invertible. If it is not invertible, LT
may be used instead, while S may be used for simulation-based inference.

While we derived and analysed our numerical schemes and estimators with a focus on scalar SDEs, the proposed approach can be extended to multivariate SDEs with H\"older diffusion coefficients by decomposing the original SDE via multivariate reducible SDEs (through the multivariate Lamperti transformation). Still confining ourselves within one-dimensional models, our approach could also be extended to scalar SDEs with H\"older drift coefficients, allowing other types of decompositions into reducible SDEs than those considered here.

Pilipovic et al. \cite[April 17, 2026]{pilipovic2026} independently introduced a S splitting estimator for multivariate SDEs with Pearson-type multiplicative noise (i.e., SDEs with a multivariate version of the scalar diffusion coefficient $g(x)=\sqrt{ax^2+b x+c}$), proving consistency and asymptotic normality of the estimators. Despite its similarity at a first sight, there are some substantial differences between our approaches. One first straightforward remark is that they deal with multivariate SDEs, while we deal with scalar SDEs but with a more general diffusion coefficient. Further, their adopted splitting scheme, whose convergence is not investigated, decomposes the SDE into a multivariate Pearson-type SDE and an ODE. Since the former cannot be explicitly solved, they rely on Kessler's approach, approximating the SDE solution with a Gaussian distribution
 with derived exact moments. 
If applied to scalar Pearson-type SDEs such as some of the models considered here, their approach will coincide with Kessler's estimator, which we have shown to perform less favourably compared to our derived estimators, and which furthermore fails to preserve 
 the state space domain of the processes (e.g., the univariate Cox-Ingersoll-Ross and Wright-Fisher models are defined on $\mathbb{R}_+$ and $(0,1)$, respectively), owed to
 the Gaussianity on which Kessler's (splitting) scheme relies. Similar numerical considerations may hold in the multivariate case. Ad-hoc adjustments (e.g. time-adaptive time steps or alternative numerical methods, as the extended splitting scheme mentioned above) may then be adopted for simulation purposes, or when there is an interest in embedding these schemes within a simulation-based-inference framework, for partially observed SDEs.  These considerations go beyond \cite{pilipovic2026}, whose proposed numerical scheme has the merit of giving rise to an explicit pseudo-likelihood in the fully observed setting. 
From a statistical point of view, their derived pseudo-likelihood can be decomposed in functions which depend only on either the drift parameters, or the diffusion parameters. Hence, their theoretical framework for proving consistency and asymptotic normality cannot be applied in our setting, where the presence of terms with tangled parameters require different proof techniques.

Finally, it is worth stressing that the statistical benefit arising from the proposed splitting schemes goes beyond the  maximum likelihood estimation considered here. Indeed, the one-step conditional means and variances derived from the splitting schemes may be used within the generalised method of moments, while the splitting schemes may be used in simulation-based inference methods like ABC, for example when the S transition density is not available due to the lack of invertibility.

\section*{Acknowledgement}
We are grateful to the Scientific Computing Research Technology Platform (SCRTP) at the University of Warwick for the provision of computational resources. Bowen Fang gratefully acknowledges funding from the Statistics Centre for Doctoral Training at the University of Warwick. For the purpose of open access, the authors have applied a Creative Commons Attribution (CC BY) licence to any Accepted
Manuscript version arising from this submission.
  \bibliographystyle{abbrvnat}
\bibliography{references}     
\begin{appendix}
\section{Proofs of numerical results} 
\label{sec:appendix A}
\subsection{Auxiliary results}
 We now recall some auxiliary results that we will use frequently in our proofs of the numerical results of the splitting schemes from Section \ref{sec:convergence}. The first is a lemma whose proof is immediate from H\"older's inequality, see for example \cite{schurz_axiomatic_2006}.

\begin{lemma}\label{lemma:open_p_bracket}
    Suppose $\left(\mathbb{B},\|\cdot\|_{\mathbb{B}}\right)$ is a Banach space with respect to the norm $\|\cdot\|_{\mathbb{B}}$. Let $X_i\in \mathbb{B},\; i\in \{1,\dots,n\}$. Then, for $n\in \mathbb{N},p\geq 1$, 
    \begin{equation}\label{equation:open pth bracket}
        \left\|\sum_{i=1}^n X_i\right\|_{\mathbb{B}}^p \leq n^{p-1} \sum_{i=1}^n\left\|X_i\right\|_{\mathbb{B}}^p.
    \end{equation}
\end{lemma}
\noindent The next is Burkholder-Davis-Gundy's  (BDG) inequality, see Theorem 1.7.3 in \cite{mao_stochastic_2011} for its proof.
\begin{lemma}[Burkholder-Davis-Gundy's inequality]\label{lemma:BDG_inequality}
    Assume $W_t$ is a $m-$dimensional Brownian motion defined on a complete probability space $(\Omega,\mathcal{F},\mathbb{P})$ adapted to the filtration $\{\mathcal{F}_t\}_{t\geq0}$. Let $g\in L^p(\mathbb{R}^+,\mathbb{R}^{d\times m})$. 
    Now define $x(t)=\int_0^t g(s) d W_s$ and $y(t)=\int_0^t|g(s)|^2 d s$. Then, for every $p>0$, there exist constants $c_p, C_p$ depending only on $p$, such that for all $t\geq 0$ we have
    \begin{equation}\label{eqn:BDG inequality}
        c_p \mathbb{E}[|y(t)|^{\frac{p}{2}}] \leq \mathbb{E}\left(\sup _{0 \leq s \leq t}|x(s)|^p\right) \leq C_p \mathbb{E}[|y(t)|^{\frac{p}{2}}].
    \end{equation}
\end{lemma}
We also need Young's inequality. 
\begin{lemma}[Young's inequality]\label{lemma:Youngs inequality}
If $u>1$ and $v>1$ are real numbers such that $1/u+1/v=1$, we have
    \begin{equation}\label{eqn:Young's inequality}
        a b \leq \frac{\delta}{u} a^u+\frac{1}{v \delta^{v / u}} b^v \quad \forall a, b, \delta>0 .
    \end{equation}
\end{lemma}
We then recall the generalised Gronwall's inequality \cite{mao_stochastic_2011,tian2020nonlinear} that has been extensively used  in the analysis of differential equations. 
\begin{lemma}[Gronwall's inequality]\label{lemma:Gronwall}
    Suppose $q\geq 1,b\geq 0$ are constants, $u(t),a(t)\in C(\mathbb{R}^+,\mathbb{R}^+)$, and
    \begin{equation*}
        u(t)\leq a(t)+b\int_0^t u^q(s)ds, \quad \forall t\geq 0.
    \end{equation*}
    Then 
    \begin{equation}
        u(t)\leq\begin{cases}
             a(t)+b\int_0^t e^{b(t-s)} a(s) ds,\quad & \text { if } q=1 ,\\
             a(t)+\left(\kappa^{1-q}(t)-(q-1)2^{q-1}bt\right)^{\frac{1}{1-q}},\quad & \text { if } q>1,
        \end{cases}\label{eqn:Gronwall inequality}
    \end{equation}
    where $\kappa(t):=2^{q-1}b\int_0^t a^q(s)ds$, and $\kappa^{1-q}(t)-(q-1)2^{q-1}bt>0$.
\end{lemma}

\subsection{Derivation of the LT scheme over the entire time interval}
\label{AuxiliaryNum}
The first step to prove the numerical properties of the derived LT scheme is writing its solution over the entire time interval, as described below. First, writing the solution to the ODE \eqref{eqn:LT subequations} in integral form, and using the nontrivial Taylor expansion 
(see Chapter 5 of \cite{kloeden_numerical_1992})
\begin{equation}\label{eqn:ODE expansion}
X_t=X_0+f_1(X_0)t+\frac{t^2}{2}f_1(X_0)f^{\prime}_1(X_0)+O(t^3),
\end{equation}
when solving in one-step starting from $X_{t_k}^{[1]}$, we get
\begin{eqnarray}
    \nonumber X^{[1]}_{t_{k+1}}&=&X^{[1]}_{t_k}+\int_{t_k}^{t_{k+1}}\left[f_1(X^{[1]}_{t_k})+\int_{t_k}^s f'_1(X^{[1]}_z)
    f_1(X^{[1]}_z)dz\right]ds\\
    \label{eqn:ODE one step}&=&X_{t_k}^{[1]}+f_1(X^{[1]}_{t_k})h+\mathcal{R}(h^2;X^{[1]}_{t_k}),
\end{eqnarray}
where 
\begin{equation*}
    \mathcal{R}(h^2;X^{[1]}_{t_k})=\int_{t_k}^{t_{k+1}}\int_{t_k}^s 
    f'_1(X^{[1]}_z)
    f_1(X^{[1]}_z)dzds,
\end{equation*}
with the right-hand side in \eqref{eqn:ODE one step} corresponding to the Euler approximation of the ODE plus the remainder term $\mathcal{R}$. 
Then, the one-step LT scheme (for generic $f_1, f_2$) starting from $X_{t_k}$ becomes 
\begin{equation}\label{eqn:LT one step}
    X^{\textrm{LT}}_{t_{k+1}}=X_{t_k}+f_1(X_{t_k})h+\mathcal{R}(h^2;X_{t_k})+\int_{t_k}^{t_{k+1}}f_2(X^{LT}_s)ds+\int_{t_k}^{t_{k+1}}g(X^{LT}_s)dW_s,
\end{equation} 
coinciding with the SD scheme \cite{halidias_semi-discrete_2012,stamatiou_numerical_2016,stamatiou_boundary_2018} when the remainder is ignored, see Remark \ref{semi-discrete}. 
To prove the strong mean square convergence of the LT scheme, we follow the framework of \cite{higham_strong_2002,stamatiou_boundary_2018}. First, we write a continuous-time extension of \eqref{eqn:LT one step} at time $t\in[t_k,t_{k+1})$ as 
\begin{equation}\label{eqnLTt}
    X^{\textrm{LT}}_{t}=X_{t_k}+\int_{t_k}^t\left(f_1(X_{t_k})+\frac{1}{h}\mathcal{R}(h^2;X_{t_k})+f_2(X^{LT}_s)\right)ds+\int_{t_k}^{t}g(X^{LT}_s)dW_s,
\end{equation} 
with $f_1(X_{t_k})=f_1(X_{\hat t})$ and $\mathcal{R}(h^2;X_{t_k})=\mathcal{R}(h^2;X_{\hat t})$ as $\hat t=t_k$. Then, we derive the (continuous-time extension of the) LT solution on the time interval $[0,t]$ for $t\leq T$ starting from $X_0$ by  back iterating \eqref{eqnLTt} in $X_{t_k}=X_{t_k}^\textrm{LT}$, obtaining
\begin{equation}\label{eqn:LT full step}
    X^{\textrm{LT}}_t=X_0+\int_0^t \left(f_1(X^{LT}_{\hat{s}})+f_2(X^{LT}_s)+\frac{1}{h}\mathcal{R}(h^2;X^{LT}_{\hat{s}})\right)ds+\int_0^t g(X^{LT}_s)dW_s,
\end{equation}
 where $\hat s=t_l$ for any $s\in [t_l,t_{l+1}), l=0,\ldots, k-1$ and $\hat s=t_k$ for $s\in[t_k,t)$.
 
\subsection{Proofs of Lemma \ref{lemma:LT_onestep_stop} and Proposition \ref{prop:LT L1}}
In both proofs, the arbitrary constants $C_{\cdot}$ may vary in each line.
\begin{proof}[Proof of Lemma \ref{lemma:LT_onestep_stop}]
If $t\leq \theta_R$, i.e. $t\wedge\theta_R=t$,  by Assumption \ref{assump:main}(i) on $f_1(x)$, we have
    \begin{equation*}
        \mathcal{R}(h^2;X^{\textrm{LT}}_{t\wedge \theta_R})\leq C_Rh^2.
    \end{equation*}
If $t>\theta_R$, i.e. $t\wedge \theta_R=\theta_R\in [t_k,t_{k+1})$ for some $k\in \{0,1,\dots,N-1\}$, and $p\geq1$, we have
    \begin{align*}
         \left|X^{\textrm{LT}}_{t \wedge \theta_R}-X^{\textrm{LT}}_{\widehat{t \wedge \theta_R}}\right|^p \leq &\left|\int_{t_k}^{t\wedge \theta_R} \left[f_1(X^{\textrm{LT}}_{\hat{s}})+f_2(X^{\textrm{LT}}_s)+\frac{1}{h}\mathcal{R}(h^2;X^{\textrm{LT}}_{\hat{s}})\right]ds+\int_{t}^{t\wedge \theta_R}g(X^{\textrm{LT}}_s)dW_s\right|^p\\
         \leq& 4^{p-1}\left(C_R h^p+\left|\int_{t_k}^{t\wedge \theta_R}f_2(X^{\textrm{LT}}_s)ds\right|^p+C_Rh^{2p}+\left|\int_{t}^{t\wedge \theta_R}g(X^{\textrm{LT}}_s)dW_s\right|^p\right)\\
         \leq&C_{R,p}h^p+4^{p-1}h^{p-1}\int_{t_k}^{t\wedge \theta_R}\left|f_2(X^{\textrm{LT}}_s)\right|^pds+C_Rh^{2p}+4^{p-1}\left|\int_{t}^{t\wedge \theta_R}g(X^{\textrm{LT}}_s)dW_s\right|^p\\
         \leq&C_{R,p}h^p+C_Rh^p+C_Rh^{2p}+C_p\left|\int_{t}^{t\wedge \theta_R}g(X^{\textrm{LT}}_s)dW_s\right|^p,
    \end{align*}
    where we use Lemma \ref{lemma:open_p_bracket} in the second inequality, H\"older's inequality (on $f_1$ and $f_2$) in the third step, and Assumption \ref{assump:main}(i) on $f_2(x)$ using the fact that $t\wedge\theta_r\leq t_{k+1}$ and $f_2(x)$ is bounded in $[t_k,t_{k+1})$ for $|x|\leq R$. Now, if we take the expectation of the last term and use BDG's inequality \eqref{eqn:BDG inequality} on $x(s)=\int_t^{t\wedge\theta_R}g(s)dW_s$, we have 
    \begin{equation*}
        \mathbb{E}\left[C_p\left|\int_{t}^{t\wedge \theta_R}g(X^{\textrm{LT}}_s)dW_s\right|^p\right]\leq C_p\mathbb{E}\left|\int_{t_k}^{t\wedge \theta_R}|g(X^{\textrm{LT}}_s)|^{ 2}ds\right|^{\frac{p}{2}}\leq C_{R,p}h^{\frac{p}{2}}.
    \end{equation*}
    Combining, we get \eqref{eqn:LT onestep bound}, and \eqref{eqn:LT onestep convergence rate} follows immediately.\\
    Finally, if $p\in (0,1)$, using Jensen's inequality we have
    \begin{equation}
        \mathbb{E}\left|X^{\textrm{LT}}_{t\wedge \theta_R}-X^{\textrm{LT}}_{\widehat{t\wedge \theta_R}}\right|^p\leq  \left(\mathbb{E}\left|X^{\textrm{LT}}_{t\wedge \theta_R}-X^{\textrm{LT}}_{\widehat{t\wedge \theta_R}}\right|^2\right)^{\frac{p}{2}}\leq C_{R,p}h^{\frac{p}{2}},
    \end{equation}
    which concludes the proof.
\end{proof}

\begin{proof}[Proof of Proposition \ref{prop:LT L1}]
    We use the method of Yamada and Watanabe \cite{yamada_uniqueness_1971} to approximate the function $f(x)=|x|$ by sequence of smooth functions $\phi_m(x), m\in \mathbb{N}$:
    \begin{equation*}
        \phi_m(x)=\int_0^{|x|} d y \int_0^y \psi_m(u) d u,
    \end{equation*}
    where $\psi_m(x)$ is a continuous function defined on $(e_m,e_{m-1}), e_m=e^{-m(m+1)/2}$ with $0 \leq \psi_m(x) \leq 2 /(m x)$. Then, the  following relations hold
    \begin{align}
        &|x|-e_{m-1} \leq \phi_m(x) \leq|x|, \quad x \in \mathbb{R},\label{eqn:phi_rel_1} \\
        &\left|\phi_m^{\prime}(x)\right| \leq 1, \quad x \in \mathbb{R},\label{eqphi2}\\
        &\left|\phi_m^{\prime \prime}(x)\right| \leq \frac{2}{m|x|}, \quad \text { when } e_m<|x|<e_{m-1} \quad \text { and } \quad\left|\phi_m^{\prime \prime}(x)\right|=0 \quad \text { otherwise. }\label{eqn:phi_rel_2}
    \end{align}
    Denoting
    \[
    \mathcal{E}_t:=X^{\textrm{LT}}_t
    -X_t,
    \]
    from \eqref{eqn:phi_rel_1} on $x=\mathcal{E}_{t\wedge\theta_R}$, we have
    \begin{equation}
\left|\mathcal{E}_{t \wedge \theta_R}\right| \leq e_{m-1}+\phi_m\left(\mathcal{E}_{t \wedge \theta_R}\right) .
\end{equation}
Now, if we apply Ito's lemma to $\phi_m\left(\mathcal{E}_{t \wedge \theta_R}\right)$,  then
\begin{align*}
    \phi_m(\mathcal{E}_{t\wedge \theta_R})=&\int_0^{t\wedge \theta_R}\phi^{\prime}_m(\mathcal{E}_s)\left(f_1(X^{LT}_{\hat{s}})-f_1(X_s)+f_2(X^{LT}_s)-f_2(X_s)+\frac{1}{h}\mathcal{R}(h^2;X^{LT}_{\hat{s}})\right)ds\\
    &+\int_0^{t\wedge \theta_R}\phi^{\prime}_m(\mathcal{E}_s)\left[g(X^{LT}_s)-g(X_s)\right]dW_s+\int_0^{t\wedge \theta_R}\frac{\phi^{\prime\prime}_m(\mathcal{E}_s)}{2}\left[g(X^{LT}_s)-g(X_s)\right]^2ds\\
    \leq&C_R \int_0^{t\wedge \theta_R}\left|\mathcal{E}_s\right|ds+C_R\int_0^{t\wedge \theta_R}\left|X^{LT}_s-X^{LT}_{\hat{s}}\right|ds+C_Rh^2+\frac{C_R}{m}\int_0^{t\wedge \theta_R}\left|\mathcal{E}_{s}\right|^{2q-1}ds\\
    &+\int_0^{t\wedge \theta_R}\phi^{\prime}_m(\mathcal{E}_s)\left[g(X^{LT}_s)-g(X_s)\right]dW_s,
\end{align*}
where we use 
Assumption \ref{assump:main}, \eqref{eqphi2}, \eqref{eqn:phi_rel_2} 
and the fact that
\begin{equation*}
    \left|X^{LT}_{\hat{s}}-X_{s}\right|\leq \left|X^{LT}_{\hat{s}}-X^{LT}_s\right|+\left|X^{LT}_s-X_s\right|.
\end{equation*}
Now, using Young's inequality \eqref{eqn:Young's inequality} and \eqref{eqn:phi_rel_2} for non-zero $\phi''(x)$ on the last term in the inequality above, we have
\begin{align*}
    \frac{C_R}{m} \int_0^{t \wedge \theta_R}\left|\mathcal{E}_s\right|^{2 q-1} d s \leq& \frac{C_R}{m} \int_0^{t \wedge \theta_R}\frac{\left(q\left|\mathcal{E}_s\right|^2\left(e_m\right)^{\frac{2 q-2}{q}}+(1-q)\left(e_m\right)^2\right)}{\left|\mathcal{E}_s\right|} d s\\
    \leq& \frac{C_R q \left(e_m\right)^{\frac{2 q-2}{q}}}{m} \int_0^{t \wedge \theta_R}\left|\mathcal{E}_s\right| d s+\frac{C_R}{m} e_m.
\end{align*}
Also, from Theorem 1.5.8 in \cite{mao_stochastic_2011}, the fact that 
\begin{equation*}
    \mathbb{E}\int_0^{t\wedge \theta_R}\left|\phi^{\prime}_m(\mathcal{E}_s)\left(g(X^{LT}_s)-g(X_s)\right)\right|^2ds< \infty
\end{equation*}
implies that $\mathbb{E}\int_0^{t\wedge \theta_R}\phi^{\prime}_m(\mathcal{E}_s)\left[g(X^{LT}_s)-g(X_s)\right]dW_s=0$. Now, combining everything together and using \eqref{lemma:LT_onestep_stop}, we have 
\begin{align*}
   \mathbb{E}\left|\mathcal{E}_{t\wedge \theta_R}\right|\leq& \left(\frac{C_Re_m}{m}+C_Rh^{\frac{1}{2}}+C_Rh^2\right)+\left(C_R+\frac{C_Rqe_m^{2-2/q}}{m}\right)\int_0^{T}\mathbb{E}\left|\mathcal{E}_{l\wedge \theta_R}\right|ds\\
    \leq&\left(\frac{C_Re_m}{m}+C_Rh^{\frac{1}{2}}+C_Rh^2\right)e^{\left(C_R+\frac{C_Rqe_m^{2-2/q}}{m}\right)T},
\end{align*}
where we use Lemma \ref{lemma:LT_onestep_stop} to the second term in the first inequality, and 
Gronwall's inequality \eqref{eqn:Gronwall inequality} for $q=1$ in the second inequality. Taking the supremum over $t\in[0,T]$ concludes the proof.
\end{proof}
\subsection{Proof of Theorem \ref{thrm:LT convergence}}
We are now ready to prove the strong mean square convergence of the LT scheme for the generic decomposition \eqref{eqn:LT subequations}.
\begin{proof}[Proof of Theorem \ref{thrm:LT convergence}] Under the assumptions of Theorem \ref{thrm:LT convergence}, we denote 
\begin{equation*}
    \Omega_R:=\{\omega\in \Omega: \; \sup_{0\leq t\leq T}|X^{\textrm{LT}}_t|\leq R,\; \sup_{0\leq t \leq T}|X_t|\leq R\},
\end{equation*}
and $\theta_R,\;\mathcal{E}_t$ same as before. Then, we have
    \begin{align}
        \mathbb{E}\left[ \sup _{0 \leq t \leq T}\left|\mathcal{E}_t\right|^2\right]=&\mathbb{E} \left[\sup _{0 \leq t \leq T}\left|\mathcal{E}_t\right|^2 \mathbb{I}_{\Omega_R}\right]+\mathbb{E} \left[\sup _{0 \leq t \leq T}\left|\mathcal{E}_t\right|^2 \mathbb{I}_{\Omega_R^c}\right]\nonumber\\
        \leq& \mathbb{E}\left[ \sup _{0 \leq t \leq T}\left|\mathcal{E}_{t \wedge \theta_R}\right|^2\right]+\left(\mathbb{E}\left[\sup _{0 \leq t \leq T}\left|\mathcal{E}_t\right|^p\right]\right)^{2 / p}\left[\mathbb{P}\left(\Omega_R^c\right)\right]^{(p-2) / p}\nonumber\\
        \leq& \mathbb{E}\left[\sup _{0 \leq t \leq T}\left|\mathcal{E}_{t \wedge \theta_R}\right|^2\right]+\left(2^{p-1} \mathbb{E} \left[\sup _{0 \leq t \leq T}\left(\left|X^{\textrm{LT}}_t\right|^p+\left|X_t\right|^p\right)\right]\right)^{2 / p}\left(\mathbb{P}\left(\Omega_R^c\right)\right)^{(p-2) / p}\nonumber\\
        \leq& \mathbb{E}\left[\sup _{0 \leq t \leq T}\left|\mathcal{E}_{t \wedge \theta_R}\right|^2\right]+4 A^{2 / p}\left(\mathbb{P}\left(\Omega_R^c\right)\right)^{(p-2) / p}\label{eqn:L2 bound decompose},
    \end{align}
    where we use H\"older's inequality in the second step on the second term, Lemma \ref{lemma:open_p_bracket} in the third step, and the moment bound assumption \eqref{assump:LT_moment_bound} in the last step. Now, from Markov's inequality and again \eqref{assump:LT_moment_bound}, we have 
    \begin{align}
    \mathbb{P}\left(\Omega_R^c\right) \leq& \mathbb{P}\left(\sup _{0 \leq t \leq T}\left|X^{\textrm{LT}}_t\right|>R\right)+\mathbb{P}\left(\sup _{0 \leq t \leq T}\left|X_t\right|>R\right)\nonumber\\
    \leq&\left(\mathbb{E} \left[\sup _{0 \leq t \leq T}\left|X^{\textrm{LT}}_t\right|^p\right]\right) R^{-p}+\left(\mathbb{E} \left[\sup _{0 \leq t \leq T}\left|X_t\right|^p\right]\right) R^{-p}\nonumber\\
    \leq& 2 A R^{-p}\label{eqn:L2 decompose 1}.
\end{align}
For $\mathcal{E}_{t\wedge \theta_R}^2$, using Ito's lemma,  we get
\begin{align*}
    \mathcal{E}_{t\wedge \theta_R}^2=&\int_0^{t\wedge \theta_R}2|\mathcal{E}_s|\left(f_1(X^{\textrm{LT}}_{\hat{s}})-f_1(X_s)+f_2(X^{\textrm{LT}}_{s})-f_2(X_s)+\frac{1}{h}\mathcal{R}(h^2;X^{\textrm{LT}}_{\hat{s}})\right)ds\\
    &+\int_0^{t\wedge \theta_R}2|\mathcal{E}_s|\left(g(X^{\textrm{LT}}_{s})-g(X_s)\right)dW_s+\int_0^{t\wedge \theta_R}\left(g(X^{\textrm{LT}}_{s})-g(X_s)\right)^2 ds.
\end{align*}
If we denote the three integrals on the right-hand side of the equation above as $\varepsilon^{(1)}_t,\varepsilon^{(2)}_t,\varepsilon^{(3)}_t$, then, by Assumption \ref{assump:main} (i) and Lemma \ref{lemma:LT_onestep_stop}, we have 
\begin{align}
    \mathbb{E}\left[\sup_{0\leq t\leq T}\varepsilon^{(1)}_t\right]\leq& C_R\int_0^{T}\mathbb{E}\left[\sup_{0\leq l \leq s}\left|\mathcal{E}_{l\wedge \theta_R}\right|^2\right]ds+C_R\mathbb{E}\left[\sup_{0\leq t\leq T}\int_0^{t\wedge \theta_R}\left|X^{\textrm{LT}}_s-X^{\textrm{LT}}_{\hat{s}}\right|^2 ds\right]+C_Rh^2\nonumber\\
    \leq&C_R\int_0^{T}\mathbb{E}\left[\sup_{0\leq l \leq s}\left|\mathcal{E}_{l\wedge \theta_R}\right|^2\right]ds+C_Rh+C_Rh^2\label{eqn:varepsilon_1}.
\end{align}
Using \eqref{eqn:BDG inequality}, we have 
\begin{align}
    \mathbb{E}\left[\sup_{0\leq t\leq T}\varepsilon^{(2)}_t\right]\leq&C\mathbb{E}\sqrt{\int_0^{T\wedge \theta_R}|\mathcal{E}_s|^2\left(g(X^{\textrm{LT}}_{s})-g(X_s)\right)^2ds}\nonumber\\
    \leq & C\mathbb{E}\sqrt{\sup_{0\leq t \leq T}|\mathcal{E}_{t\wedge \theta_R}|^2\int_0^{T\wedge \theta_R}\left(g(X^{\textrm{LT}}_{s})-g(X_s)\right)^2ds}\nonumber\\
    \leq& C\mathbb{E}\left[\sup_{0\leq t\leq T}\left|\mathcal{E}_{t\wedge \theta_R}\right|^2\right]+ C\mathbb{E}\left[\int_0^{T\wedge \theta_R}\left(g(X^{\textrm{LT}}_{s})-g(X_s)\right)^2ds\right].\label{eqn:varepsilon_2}
\end{align}
For $\varepsilon^{(3)}_t$, using Assumption \ref{assump:main} (ii), the fact that $2q<1$ and Jensen's inequality, we have
\begin{align}
    \mathbb{E}\sup_{0\leq t \leq T}\varepsilon^{(3)}_t\leq& C_R\mathbb{E}\int_0^{T\wedge \theta_R}\left(X^{\textrm{LT}}_s-X_s\right)^{2q}ds\nonumber\\
    \leq&C_R\int_0^{T}\left(\mathbb{E}\left|\mathcal{E}_{l\wedge \theta_R}\right|\right)^{2q}ds\nonumber\\
    \leq&C_R(K_{r,h,m,q,T})^{2q}\label{eqn:varepsilon_3}.
\end{align}
Here $K_{r,h,m,q,T}=\left(\frac{C_Re_m}{m}+C_Rh^{\frac{1}{2}}+C_Rh^{2}\right)e^{C_R+\frac{C_Rqe_m^{(2-2/q)}}{m}T}$ is from \eqref{eqn:L1 bound}. Now if we combine \eqref{eqn:varepsilon_1}, \eqref{eqn:varepsilon_2} and \eqref{eqn:varepsilon_3}, we have
\begin{align}
    \mathbb{E} \sup _{0 \leq t \leq T}\left|\mathcal{E}_{t \wedge \theta_R}\right|^2\leq& C_R\left(h+h^2+(K_{r,h,m,q,T})^{2q}\right)+C_R\int_{0}^T\mathbb{E}\sup_{0\leq l\leq s}\left|\mathcal{E}_{l\wedge \theta_R}\right|^2ds\nonumber\\
    \leq& C_R\left(h+h^2+(K_{r,h,m,q,T})^{2q}\right)e^{C_RT},\label{eqn:L2 decompose 2}
\end{align}
where we used Gronwall's inequality \eqref{eqn:Gronwall inequality}.
By plugging \eqref{eqn:L2 decompose 1} and \eqref{eqn:L2 decompose 2} into \eqref{eqn:L2 bound decompose}, we have 
\begin{equation}\label{calc}
    \mathbb{E}\sup_{0\leq t\leq T}\left|\mathcal{E}_t\right|^2\leq C_R\left(h+h^2+(K_{r,h,m,q,T})^{2q}\right)e^{C_RT}+CAR^{2-p}.
\end{equation}
Now, remembering that $p>2$ by assumption, for any $\epsilon>0$, we can first choose $R$ large enough such that $CAR^{2-p}\leq \epsilon/2$, ad then $h$ small enough and $m$ large enough to bound the first term in \eqref{calc}
by $\epsilon/2$, concluding the proof.
\end{proof}

\subsection{Proof of Proposition \ref{prop:pearson consistency}}\label{ProofProp4.5}
We now focus on the one-step numerical consistency of the LT \eqref{eqn:LT subsystem} and S \eqref{eqn:Strang subsystem} schemes for the specific decomposition \eqref{eqn:ODE1}-\eqref{eqn:SDE2} instead of the more general one. To simplify the notation, we sometimes use $f$ in place of $f(x)$ in the proof below.
\begin{proof}[Proof of Proposition \ref{prop:pearson consistency}] 
   
Using Lemma \ref{lemma:expectation expansion}, we have
\begin{equation}\label{eqn:x expectation}
    \mathbb{E}[X_{t_{k+1}}|X_{t_k}=x]=x+hf(x)+\frac{h^2}{2}(f(x)f^{\prime}(x)+\frac{1}{2}g^2(x)f^{\prime \prime}(x))+O(h^3).
\end{equation}
By considering the SDE decomposition \eqref{eqn:ODE1}-\eqref{eqn:SDE2}, we have that
\begin{align}
    \frac{1}{2}ff^{\prime}&=\frac{1}{2}(f_1+\frac{1}{2}gg^{\prime})(f_1^{\prime}+\frac{1}{2}((g^{\prime})^2+gg^{\prime\prime}))\nonumber\\
    &=\frac{f_1f_1^{\prime}}{2}+\frac{gg^{\prime}f_1^{\prime}}{4}+\frac{f_1(g^{\prime})^2}{4}+\frac{g(g^{\prime})^3}{8}+\frac{f_1gg^{\prime\prime}}{4}+\frac{g^2g^{\prime}g^{\prime\prime}}{8}\label{eqn:f fprime},\\
    \frac{1}{4}g^2f^{\prime\prime}&=\frac{1}{4}g^2(f_1^{\prime \prime}+g^{\prime} g^{\prime \prime}+\frac{1}{2} g^{\prime} g^{\prime}+\frac{1}{2} g g^{\prime \prime \prime})=\frac{1}{4} g^2 f_1^{\prime \prime}+\frac{3}{8} g^2 g^{\prime} g^{\prime \prime}+\frac{1}{8} g^3 g^{\prime \prime \prime}.\label{eqn:g square fpp}
\end{align}
For the LT splitting \eqref{eqn:LT subsystem}, denote $\mathcal{V}_1=v\left(\varphi_h^{[1]}(X_{t_k})\right)$. By expanding it around $\xi_k$, we have
\begin{align}
X^{\textrm{LT}}_{t_{k+1}}&=v^{-1}(\xi_{k}+\mathcal{V}_1)\nonumber\\
    &=\varphi_h^{[1]}(X_{t_{k}})+\xi_k\frac{d v^{-1}}{d x}(\mathcal{V}_1)+\frac{\xi_k^2}{2}\frac{d^2v^{-1}}{dx^2}(\mathcal{V}_1)+\frac{\xi_k^3}{6}\frac{d^3 v^{-1}}{d x^3} (\mathcal{V}_1)+\frac{\xi_k^4}{24}\frac{d^4 v^{-1}}{d x^4} (\mathcal{V}_1)+O(\xi_k^5).\label{eqn:LT expand 1}
\end{align}
Since $v(x)=\int^x 1/{g(s)}ds$, if its inverse $v^{-1}$ is well defined on the state space $\mathcal{X}$, then
\begin{align}
    \frac{d v^{-1}}{d x}&=g_v,\quad \frac{d^2 v^{-1}}{d x}=g^{\prime}_vg_v,\quad \frac{d^3 v^{-1}}{dx^3}=g^{\prime \prime}_vg^2_v+(g^{\prime}_v)^2g_v\label{eqn:v derivatives}\\
    \frac{d^4 v^{-1}}{dx^4}&=g^{\prime\prime\prime}_vg^3_v+4g^{\prime\prime}_vg^2_vg^{\prime}_v+(g^{\prime}_v)^3g_v\nonumber,
\end{align}
where we denoted $g_v(x):=g(v^{-1}(x))$. 
Since given $X_{t_k}=x$, the Brownian increment $\xi_k$ is independent on $X_{t_k}$, by the Gaussianity of $\xi_k$, \eqref{eqn:ODE expansion},  \eqref{eqn:LT expand 1} and \eqref{eqn:v derivatives}, we have
\begin{align*}
    &\mathbb{E}[X^{\textrm{LT}}_{t_{k+1}}|X_{t_k}=x]=\varphi_h(x)+\frac{h}{2}g^{\prime}_vg_v+\frac{h^2}{8}(g^{\prime\prime\prime}_vg^3_v+4g^{\prime\prime}_vg^2_vg^{\prime}_v+(g^{\prime}_v)^3g_v)+O(h^3)\\
    &=x+f_1h+\frac{h^2}{2}f_1f_1^{\prime}+\frac{h}{2}\left[g(x+hf_1+O(h^2))g^{\prime}(x+hf_1+O(h^2))\right]\\
    &\qquad +\frac{h^2}{8}(g^{\prime\prime\prime}_vg^3_v+4g^{\prime\prime}_vg^2_vg^{\prime}_v+(g^{\prime}_v)^3g_v)+O(h^3)\\
    &=x+h(f_1+\frac{gg^{\prime}}{2})+\frac{h^2}{2}(f_1f_1^{\prime}+f_1gg^{\prime\prime}+f_1(g^{\prime})^2)+\frac{h^2}{8}(g^{\prime\prime\prime}g^3+4g^{\prime\prime}g^2g^{\prime}+(g^{\prime})^3g)+O(h^3).
\end{align*}
Comparing this expression with \eqref{eqn:x expectation}, \eqref{eqn:f fprime} and \eqref{eqn:g square fpp} lead to \eqref{eqn:ergodic LT consistency}.

For the S \eqref{eqn:Strang subsystem} scheme, if we denote $\mathcal{V}_2=v(\varphi_{\frac{h}{2}}^{[1]}(X_{t_k}))$, then again using \eqref{eqn:ODE expansion}
\begin{equation*}
    X^{S}_{t_{k+1}}=\varphi_{\frac{h}{2}}^{[1]}(v^{-1}(\xi_k+\mathcal{V}_2))=U_1+U_2+U_3+O(h^3),
\end{equation*}
where
\begin{equation*}
    U_1=v^{-1}(\xi_k+\mathcal{V}_2),\qquad U_2=\frac{h}{2}f_1(U_1),\qquad U_3=\frac{h^2}{8}f_1(U_1)f^{\prime}_1(U_1).
\end{equation*}
Using a similar argument as for LT, we have
\begin{equation}\label{eqn:U1 expectation}
    \mathbb{E}[U_1|X_{t_k}=x]=x+\frac{h}{2}(f_1+gg^{\prime})+\frac{h^2}{8}(f_1f_1^{\prime}+2f_1gg^{\prime\prime}+2f_1(g^{\prime})^2+g^{\prime\prime\prime}g^3+4g^{\prime}g^2g^{\prime\prime}+g(g^{\prime})^3)+O(h^3).
\end{equation}
For $U_2$, we have
\begin{align*}
    U_2&=\frac{h}{2}f_1(U_1)=\frac{h}{2}f_1\left(x+\frac{h}{2}f_1+\xi_k\frac{dv^{-1}}{dx}(\mathcal{V}_2)+\frac{\xi_k^2}{2}\frac{d^2v^{-1}}{dx^2}(\mathcal{V}_2)+O(\xi_k^3)\right)\\
    &=\frac{h}{2}\left(f_1+f^{\prime}_1(\frac{h}{2}f_1+\xi_k\frac{dv^{-1}}{dx}(\mathcal{V}_2)+\frac{\xi_k^2}{2}\frac{d^2v^{-1}}{dx^2}(\mathcal{V}_2))+\frac{f^{\prime\prime}_1}{2}\xi_k^2(\frac{dv^{-1}}{dx}(\mathcal{V}_2))^2+\mathcal{R}(
h^3;\mathcal{V}_2)\right),
\end{align*}
where $\mathcal{R}(h^3;\mathcal{V}_2)$ contains all terms that either become $0$ or have order $h^3$ after taking the conditional expectation. Then we have
\begin{equation}\label{eqn:U2 expectation}
    \mathbb{E}[U_2|X_{t_k}=x]=\frac{h}{2}f_1+\frac{h^2}{4}(f_1^{\prime}f_1+f_1^{\prime}gg^{\prime}+f_1^{\prime\prime}g^2)+O(h^3).
\end{equation}
For $U_3$, it is trivial that
\begin{equation}\label{eqn:U3 expectation}
    \mathbb{E}[U_3|X_{t_k}=x]=\frac{h^2}{8}f_1^{\prime}f_1+O(h^3).
\end{equation}
Combining \eqref{eqn:U1 expectation}, \eqref{eqn:U2 expectation}, and  \eqref{eqn:U3 expectation} to compute the one-step conditional expectation of $X^{\textrm{S}}_{t_{k+1}}$, and compare it with \eqref{eqn:x expectation}, \eqref{eqn:f fprime} and \eqref{eqn:g square fpp} leads to \eqref{eqn:ergodic Strang consistency}.
\end{proof}
\section{Proofs of consistency and asymptotic normality}\label{AppendixB}
\subsection{Auxiliary results} We start by recalling some auxiliary results which will be needed in our statistical proofs.

\begin{lemma}[Lemma 9, \cite{jacod_estimation_1993}]\label{lemma:Jacob pointwise convergence}
    Let $\left(X_k^N\right)_{N \in \mathbb{N}, 1 \leq k \leq N}$ be a triangular array with each row $N$ adapted to a filtration $(\mathcal{G}^N_k)_{1\leq k\leq N}$, and let $U$ be a random variable. If
    \begin{equation*}
        \sum_{k=1}^N \mathbb{E}\left[X_k^N \mid \mathcal{G}_{k-1}^N\right] \xrightarrow[N \rightarrow \infty]{\mathbb{P}} U, \quad \sum_{k=1}^N \mathbb{E}\left[\left(X_k^N\right)^2 \mid \mathcal{G}_{k-1}^N\right] \xrightarrow[N \rightarrow \infty]{\mathbb{P}} 0,
    \end{equation*}
    then $\sum_{k=1}^N X_k^N \xrightarrow[N \rightarrow \infty]{\mathbb{P}} U$.
\end{lemma}
\begin{lemma}[Proposition A1, \cite{gloter2006parameter}]\label{lemma:Gloter uniform convergence}
    Let $(\Omega, \mathcal{F}, \mathbb{P})$ be a probability space,  and let $\Theta$ be a product of compact intervals of $\mathbb{R}$. Let $S_N(\omega,\theta):\Omega\times \Theta\mapsto \mathbb{R}$ be a sequence of measurable maps that converges to a constant $C$ in probability for all $\theta\in \Theta$. Assume that there exists on open neighbourhood of $\Theta$ such that $S_N(\omega,\cdot)$ is continuously differentiable for all $\omega\in \Omega$. Suppose
\begin{equation*}
    \sup _{N \in \mathbb{N}} \mathbb{E}\left[\sup _{\boldsymbol{\theta} \in \Theta}\left|\nabla_\theta S_N(\theta)\right|\right]<\infty .
\end{equation*}
Then, $S_N(\theta) \underset{N \rightarrow \infty}{\overset{\mathbb{P}}\longrightarrow} C$ uniformly in $\theta$.
\end{lemma}

\begin{lemma}[Lemma 8, \cite{kessler_estimation_1997}]\label{lemma:Kessler convergence}
Let $X$ be the solution of \eqref{eqn:the sde},
and 
$u:\mathbb{R}\times \Theta\mapsto \mathbb{R}$ a differentiable function with respect to $x$ an $\alpha$, with derivative of polynomial growth uniformly in $x$, uniformly in $\Theta$. 
    Under Assumption \ref{assump:main} and Assumption \ref{Assumption3}(iii), if $h\rightarrow0$ and $Nh\rightarrow\infty$, then
    \begin{equation*}
        \frac{1}{N}\sum_{k=1}^{N}u(X_{t_k},\alpha)\xrightarrow[\substack{N h\rightarrow\infty \\ h \rightarrow 0}]{\mathbb{P}_{\alpha_0}}\int u(x,\alpha)d\mu_0(x),
    \end{equation*}
    uniformly in $\alpha$.
\end{lemma}

\begin{lemma}[Proposition 3.1. in \cite{crimaldi2005convergence}]\label{prop:triangular asymptotic normal}
    Let $\left(\mathbf{X}_{N,k}\right)_{N\in \mathbb{N},1\leq k\leq N}$ be a triangular array of $d$-dimensional random vectors such that, for each $N$, the finite sequence $\left(\mathbf{X}_{N,k}\right)_{1\leq k\leq N}$ is a martingale difference array w.r.t a given filtration $(\mathcal{G})_{1\leq k\leq N}$. Define
    \begin{equation*}
        \mathbf{S}_N=\sum_{k=1}^N \mathbf{X}_{N,k},\quad N\in \mathbb{N}.
    \end{equation*}
    Suppose:
    \begin{itemize}
        \item[(i)] $\mathbb{E}\left[\sup _{1 \leq k \leq N}\left\|\mathbf{X}_{N, k}\right\|_1\right] \underset{N \rightarrow \infty}{\longrightarrow} 0 ;$
        \item[(ii)] $\sum_{k=1}^N \mathbf{X}_{N, k} \mathbf{X}_{N, k}^{\top} \xrightarrow[N \rightarrow \infty]{\mathbb{P}} \mathbf{U}$ for some non-random positive semi-definite matrix $\mathbf{U}$.
    \end{itemize}
    Then, $\mathbf{S}_N \underset{N \rightarrow \infty}{\stackrel{d}{\longrightarrow}} \mathcal{N}_d(0, \mathbf{U})$.
\end{lemma}
\subsection{Proof of Lemma \ref{lemma:Z moment bounds}}\label{proofLemmaZ}
\begin{proof}[Proof of Lemma \ref{lemma:Z moment bounds}]
    To simplify the notation, we sometimes drop $x$, and use $\mathcal{R}(h)$ and $f(\alpha)$ to denote $\mathcal{R}(h;x)$ and $f(x;\alpha)$, respectively.
    
     Proof of (i): By the definition of $Z_{t_k}(\theta_0,\sigma)$ in \eqref{Z}, we have
    \begin{equation*}
        \mathbb{E}_{\alpha_0}\left[Z_{t_k}\left(\theta_0,\sigma\right)|X_{t_{k-1}}=x\right]=\mathbb{E}_{\alpha_0}\left[v\left(X_{t_k} ; \sigma\right) \mid X_{t_{k-1}}=x\right]-v\left(\varphi_h^{[1]}\left(x ; \theta_0,\sigma\right) ; \sigma\right).
    \end{equation*}
    By the Taylor-type ODE expansion \eqref{eqn:ODE expansion} on $\varphi^{[1]}_h$ and $v(\varphi^{[1]}_h)$, we have 
    \begin{align}
        \varphi_h^{[1]}(x;\theta_0,\sigma)&=x+hf_1(\theta_0,\sigma)+\frac{h^2}{2}f_1^{\prime}(\theta_0,\sigma)f_1(\theta_0,\sigma)+\mathcal{R}(h^3),\nonumber\\
        v(\varphi_h^{[1]}(x;\theta_0,\sigma);\sigma)&=v(x;\sigma)+h\frac{f_1(\theta_0,\sigma)}{g(\sigma)}+\mathcal{R}(h^2).\label{eqn:vvarphi expand}
    \end{align}
   
    Using Lemma \ref{lemma:expectation expansion}, we can again expand the one-step conditional expectation as
    \begin{equation}
        \mathbb{E}_{\alpha_0}[v(X_{t_k};\sigma)|X_{t_{k-1}}=x]=v(x;\sigma)+h\mathcal{L}_{\alpha_0}v(x;\sigma)+\mathcal{R}(h^2) \label{eqn:v expectation expansion}.
    \end{equation}
    Now, by definition of $v,\; \mathcal{L}_{\alpha_0}$, and the fact that $f_1(x;\theta_0,\sigma)=f(x;\theta_0)-\frac{1}{2}g(x;\sigma)g^\prime(x;\sigma)$, we have
    \begin{align}
        \mathcal{L}_{\alpha_0}v(x;\sigma)=&f(x;\theta_0)v^\prime(x;\sigma)+\frac{1}{2}g(x;\sigma_0)v^{\prime\prime}(x;\sigma)\nonumber\\
        =&\frac{f_1(x;\theta_0,\sigma)}{g(x;\sigma)}+\frac{1}{2}g^\prime(x;\sigma)\left(1-\frac{g(x;\sigma_0)^2}{g(x;\sigma)^2}\right).\label{eqn:L1 v}
    \end{align}
    Plugging  \eqref{eqn:L1 v} into \eqref{eqn:v expectation expansion}, and comparing it with \eqref{eqn:vvarphi expand} gives us (i).
    
Proof of (ii): First, note that by the product rule
    \begin{align}
        \mathcal{L}_{\alpha_0}u(\alpha)v(\sigma)&=f(\theta_0)\left(u^{\prime}(\alpha)v(\sigma)+v^\prime(\sigma_0)u(\alpha)\right)+\frac{1}{2}g^2(\sigma_0)\left(u^{\prime\prime}(\alpha)v(\sigma)+v^{\prime\prime}(\sigma)u(\alpha)+2v^\prime(\sigma)u^\prime(\alpha)\right)\nonumber\\
        &=u(\alpha)\mathcal{L}_{\alpha_0}v(\sigma)+v(\sigma)\mathcal{L}_{\alpha_0}u(\alpha)+\frac{g^2(\sigma_0)}{g(\sigma)}u^\prime(\alpha).\label{eqn:product rule}
    \end{align}
    By definition of $Z_{t_k}(\theta,\sigma)$, 
    \begin{equation}\label{eqn:Zu expectation}
        \mathbb{E}_{\alpha_0}[Z_{t_k}(\theta_0,\sigma)u(X_{t_k};\alpha)|X_{t_{k-1}}=x]=\mathbb{E}_{\alpha_0}[v(X_{t_k};\sigma)u(X_{t_k};\alpha)|x]-v(\varphi^{[1]}_h(\theta_0,\sigma);\sigma)\mathbb{E}[u(X_{t_k};\alpha)|x].
    \end{equation}
    Now, using \eqref{eqn:Conditional expansion} from Lemma \ref{lemma:expectation expansion} and \eqref{eqn:ODE expansion}, we have
    \begin{align*}
        &\mathbb{E}_{\alpha_0}[v(X_{t_k};\sigma_0)u(X_{t_k};\alpha)|X_{t_{k-1}}=x]=v(\sigma)u(\alpha)+h\mathcal{L}_{\alpha_0}u(\alpha)v(\sigma)+\mathcal{R}(h^2)\\
       & v(\varphi_h^{[1]}(x;\theta_0,\sigma);\sigma)\mathbb{E}_{\alpha_0}[u(X_{t_k};\alpha)|X_{t_{k-1}}=x]\\
        &\qquad\qquad=(v(\sigma)+hv^\prime(\sigma)f_1(\theta_0,\sigma)+\mathcal{R}(h^2))(u(\alpha)+h\mathcal{L}_{\alpha_0}u(\alpha)+O(h^2))\\   &\qquad\qquad=v(\sigma)u(\alpha)+h(u(\alpha)v^\prime(\sigma)f_1(\theta_0,\sigma)+v(\sigma)\mathcal{L}_{\alpha_0}u(\alpha))+O(h^2).
    \end{align*}
    Plugging the above two equations into \eqref{eqn:Zu expectation}, and using \eqref{eqn:L1 v} and \eqref{eqn:product rule} 
    gives us (ii).

    Proof of (iii): If we set $u(X_{t_k};\alpha)=v(X_{t_k};\sigma)$, and use (i), we obtain
    \begin{eqnarray*}
        &&\mathbb{E}_{\alpha_0}\left[Z_{t_k}(\theta_0,\sigma)v(X_{t_k};\sigma)\mid X_{t_{k-1}}=x\right]=h\left[\frac{1}{2}v(\sigma)g^\prime(\sigma)\left(1-\frac{g(\sigma_0)^2}{g(\sigma)^2}\right)+\frac{g(\sigma)^2}{g(\sigma)^2}\right]+\mathcal{R}(h^2),\\
&&\mathbb{E}_{\alpha_0}\left[Z_{t_k}(\theta_0,\sigma)v(\varphi_h^{[1]}(X_{t_{k-1}};\theta_0,\sigma);\sigma)\mid X_{t_{k-1}}=x\right]=h\left[\frac{1}{2}v(\sigma)g^\prime(\sigma)\left(1-\frac{g(\sigma_0)^2}{g(\sigma)^2}\right)\right]+\mathcal{R}(h^2).
    \end{eqnarray*}
    Taking the difference of the  two expressions above gives us (iii).
\end{proof}
\begin{lemma}\label{lemma:X one step bound}
    Under Assumption \ref{Assumption3}(ii), for $p\geq 1, t\in [t_k, t_{k+1}]$,  $h< 1$, we have
    \begin{equation*}
        \mathbb{E}\left[\left|X_{t}-X_{t_k}\right|^p\mid X_{t_k}\right] \leq C_{p} h^{p/2}(1+|X_{t_k}|)^p,
    \end{equation*}
    where $C_p$ is a constant that only depends on $p$.
\end{lemma}
\begin{proof}
    Here, we use the expressions $f(x)$ and $g(x)$ without stating the parameters, as the polynomial growth condition  in Assumption \ref{Assumption3} (ii) depends only on $x$, and it holds uniformly on $\Theta$. By the definition of the SDE \eqref{eqn:the sde}, we have
    \begin{eqnarray*}
       && \left|X_t-X_{t_k}\right|^p=\left|\int_{t_k}^t f(X_s)ds+\int_{t_k}^t g(X_s)dW_s\right|^p\\
       &\leq&2^{p-1}\left(\left|\int_{t_k}^t f(X_s)ds\right|^p+\left|\int_{t_k}^t g(X_s)dW_s\right|^p\right)\\
        &\leq&2^{p-1}\left|\int_{t_k}^t(1+|X_s-X_{t_k}|+|X_{t_k}|)^{C_1}ds\right|^p+2^{p-1}\left|\int_{t_k}^t g(X_s)dW_s\right|^p\\
        &\leq&C_p(t-t_k)^{p-1}\left(\int_{t_k}^t|X_s-X_{t_k}|^{pC_1}ds+(t-t_k)(1+|X_{t_k}|)^{pC_1}\right)+2^{p-1}\left|\int_{t_k}^t g(X_s)dW_s\right|^p,
    \end{eqnarray*}
    where we use Lemma \ref{lemma:open_p_bracket} in the first inequality, polynomial growth condition in Assumption \ref{Assumption3} in the second inequality, and Hölder's inequality in the third one. Now, by the BDG's inequality \eqref{eqn:BDG inequality}, Jensen's inequality and the polynomial growth condition again, using the similar argument as above, we have
    \begin{align*}
        &\mathbb{E}\left[\left|\int_{t_k}^t g(X_s)dW_s\right|^p\mid X_{t_k}\right]\leq C_p\mathbb{E}\left[\left(\int_{t_k}^t g(X_s)^2 ds\right)^{p/2}\mid X_{t_k}\right]\\
        \leq& C_p(t-t_k)^{\frac{p-2}{2}}\left(\int_{t_k}^t \mathbb{E}\left[|X_s-X_{t_k}|^{pC_1}\mid X_{t_k}\right]ds+(t-t_k)(1+|X_{t_k}|)^{pC_1}\right)\\
        \leq& C_p(t-t_k)^{p/2}(1+|X_{t_k}|)^{pC_1}+C_p(t-t_k)^{\frac{p-2}{2}}\int_{t_k}^t \mathbb{E}\left[|X_s-X_{t_k}|^{pC_1}\mid X_{t_k}\right]ds.
    \end{align*}
    Combining the two inequalities above, using H\"older's inequality  $\mathbb{E}[|C_s-X_{t_k}|^{pC_1}|X_{t_k}]\leq (\mathbb{E}[|C_s-X_{t_k}|^{p}|X_{t_k}])^{C_1}$, and Gronwall's inequality \eqref{eqn:Gronwall inequality} with
$u(t)=\mathbb{E}\left[\left|X_t-X_{t_k}\right|^p\mid X_{t_k}\right] , a=C_ph^{p/2}(1+|X_{t_k}|)^{pC_1}, b=C_ph^{\frac{p-2}{2}}$, and $q=C_1$, we get
\[
\mathbb{E}\left[\left|X_t-X_{t_k}\right|^p\mid X_{t_k}\right]\leq a+(\kappa^{1-C_1}-(C_1-1)2^{C_1-1}bt)^{\frac{1}{1-C_1}}.
\]
    Now, by definition of $a,b,\kappa$, the second term in the inequality above is of the form
    \begin{equation*}
        \left(h^{\frac{p(1+C_1)-2}{2}\frac{1}{1-C_1}}-h^{\frac{p-2}{2}}\right)^{\frac{1}{1-C_1}}.
    \end{equation*}
    As $C_1$ is the constant arising from the polynomial growth condition, Assumption \ref{Assumption3}(ii), we can always choose a large enough $C_1>1$ such that this term is dominated by the $h^{\frac{p}{2}}$ term in $a$ when $h<1$. This concludes the proof. 
    \end{proof}
\subsection{Proof of Theorem \ref{thrm:asymptotic consistency}}\label{ProofTheorem5.2}
To prove Theorem \ref{thrm:asymptotic consistency}, 
we first scale the negative pseudo-log-likelihood function \eqref{eqn:U decomposition}
by $1/N$ to study the consistency of the diffusion parameter $\sigma$, and then scale it by $1/(Nh)$ to obtain the consistency of the drift parameter $\theta$. We decompose the scaled negative pseudo-log-likelihood into multiple triangular arrays, which we explicitly derive, and we prove their uniform convergence using existing results in \cite{jacod_estimation_1993,gloter2006parameter,yoshida1990asymptotic}. Besides weaker assumption on the diffusion coefficient (here H\"older continuous, in contrast with globally Lipschitz in \cite{kessler_estimation_1997} and constant in \cite{pilipovic_parameter_2024}), the main challenge when proving this result  is that, unlike the other estimators, which have separation of drift and diffusion parameter in the contrast function, here we have a mixture of drift parameter and diffusion parameter in $Z_{t_k}(\theta,\sigma)$, see \eqref{Z} and Remark \ref{RemarkZ}. Lemma \ref{lemma:Z moment bounds} plays an important role in tackling this new problem. 

\begin{proof}[Proof of Theorem \ref{thrm:asymptotic consistency}]
To study the consistency of the diffusion parameter $\sigma$, we first study the limit of $l_N(\theta,\sigma)/N$. From \eqref{eqn:U decomposition}, we get 
\begin{equation}\label{eqn:U decomposition2}
        \frac{1}{N}l_N(\theta,\sigma)=\frac{1}{N}\sum_{k=1}^{N}\log |g(x_{t_k};\sigma)|+\frac{1}{2}\left(U_1+U_2+2U_3\right),
        \end{equation}
      with $U_i = \tilde U_i/N$, with $\tilde U_i$ defined in Section \ref{sec:likelihood}.

   We now study the limit of each term, showing that only the first two will not convergence to zero. By Lemma \ref{lemma:Kessler convergence}, we know that
    \begin{equation}\label{eqn:Consistency log g}
        \frac{1}{N}\sum_{k=1}^{N}\log |g(x_{t_k};\sigma)|\xrightarrow[\substack{N h\rightarrow\infty \\ h \rightarrow 0}]{\mathbb{P}_{\alpha_0}}\int\log g(x;\sigma)d\mu_0(x).
    \end{equation}
    For $U_1$, we first use Lemma \ref{lemma:Jacob pointwise convergence} to prove the pointwise convergence. By Lemma \ref{lemma:Z moment bounds} (iii) and Lemma \ref{lemma:Kessler convergence},
    \begin{align*}
       &\frac{1}{Nh}\sum_{k=1}^N\mathbb{E}_{\alpha_0}\left[Z_{t_k}^2(\theta_0,\sigma)\mid X_{t_{k-1}}\right]\\
        &=\frac{1}{N}\sum_{k=1}^N \frac{g(X_{t_{k-1}};\sigma_0)^2}{g(X_{t_{k-1}};\sigma)^2}+\mathcal{R}(h;X_{t_{k-1}})\xrightarrow[\substack{N h\rightarrow\infty \\ h \rightarrow 0}]{\mathbb{P}_{\alpha_0}}\int\frac{g(x;\sigma_0)^2}{g(x;\sigma)^2}d\mu_0(x).
    \end{align*}
    Also, by Lemma \ref{eqn:Conditional expansion} and \eqref{eqn:vvarphi expand},
    \begin{equation*}
        \sum_{k=1}^{N}\mathbb{E}_{\alpha_0}\left[\frac{Z^2_{t_k}(\theta_0,\sigma)}{N^2h^2}\mid X_{t_{k-1}}\right]=\frac{1}{N^2h^2}\sum_{k=1}^N \mathcal{R}(h;X_{t_{k-1}})\xrightarrow[\substack{N h\rightarrow\infty \\ h \rightarrow 0}]{\mathbb{P}_{\alpha_0}}0.
    \end{equation*}
    For uniform convergence in $\Theta_{\sigma}$, by Assumption \ref{Assumption3}(iii) and Lemma \ref{lemma:X one step bound}, we have
    \begin{align*}
        \sup_{N\in \mathbb{N}}\mathbb{E}\left[\sup_{\sigma\in \Theta_{\sigma}}\left|\nabla_{\sigma} \sum_{k=1}^{N}(v(X_{t_k};\sigma)-v(X_{t_{k-1}};\sigma))^2+\frac{\mathcal{R}(h;X_{t_k})}{N}\right|\right]
        <\infty,
    \end{align*}
    and by Lemma \ref{lemma:Gloter uniform convergence}, we know that $U_1\xrightarrow[\substack{N h\rightarrow\infty \\ h \rightarrow 0}]{\mathbb{P}_{\alpha_0}}\int\frac{g(x;\sigma_0)^2}{g(x;\sigma)^2}d\mu_0(x)$ uniformly in $\Theta_{\sigma}$.

    For $U_2$, from \eqref{eqn:vvarphi expand}, we have
    \begin{align*}
        U_2
        =&\frac{h}{N}\sum_{k=1}^N\frac{\left(f_1(X_{t_{k-1}};\theta,\sigma)-f_1(X_{t_{k-1}};\theta_0,\sigma)\right)^2}{g(X_{t_{k-1}};\sigma)^2}+\mathcal{R}(h^2;X_{t_{k-1}}).
    \end{align*}
    By Lemma \ref{lemma:Kessler convergence}, and Assumption \ref{Assumption3}(ii) on $f_1$ and $g$,
    \begin{equation}\label{eqn:U_2 limit}
        U_2\xrightarrow[\substack{N h\rightarrow\infty \\ h \rightarrow 0}]{\mathbb{P}_{\alpha_0}}h\int\frac{\left(f_1(x;\theta,\sigma)-f_1(x;\theta_0,\sigma)\right)^2}{g(x;\sigma)^2}d\mu_0(x)=0,
    \end{equation}
    uniformly in $\alpha$.\\
     For $U_3$, by Lemma \ref{lemma:Z moment bounds} (i) and \eqref{eqn:vvarphi expand},
    \begin{align*}
        &\frac{1}{Nh}\sum_{k=1}^N\mathbb{E}_{\alpha_0}\left[Z_{t_k}(\theta_0,\sigma)\left(v\left(\varphi^{[1]}_h(X_{t_{k-1}};\theta_0,\sigma) ; \sigma\right)-v\left(\varphi^{[1]}_h(X_{t_{k-1}};\alpha) ; \sigma \right)\right)|X_{t_{k-1}}\right]\\
        &=\frac{1}{N}\sum_{k=1}^N\mathbb{E}_{\alpha_0}\left[Z_{t_k}(\theta_0,\sigma)\left(\frac{\left(f_1(X_{t_{k-1}};\theta,\sigma)-f_1(X_{t_{k-1}};\theta_0,\sigma)\right)}{g(X_{t_{k-1}};\sigma)}+\mathcal{R}(h;X_{t_{k-1}})\right)\mid X_{t_{k-1}}\right]\\
        &=\frac{1}{N}\sum_{k=1}^N\mathcal{R}(h;X_{t_{k-1}})\xrightarrow[\substack{N h\rightarrow\infty \\ h \rightarrow 0}]{\mathbb{P}_{\alpha_0}}0.
    \end{align*}
    For second moment, again by Lemma \ref{lemma:Z moment bounds} (iii),
    \begin{align*}
         &\frac{1}{N^2h^2}\sum_{k=1}^N\mathbb{E}_{\alpha_0}\left[Z^2_{t_k}(\theta_0,\sigma)\left(v\left(\varphi^{[1]}_h(X_{t_{k-1}};\theta_0,\sigma) ; \sigma\right)-v\left(\varphi^{[1]}_h(X_{t_{k-1}};\alpha) ; \sigma \right)\right)^2\right]\\
         &=\frac{1}{N^2}\sum_{k=1}^N\mathbb{E}_{\alpha_0}\left[Z_{t_k}(\theta_0,\sigma)^2\left(\frac{\left(f_1(X_{t_{k-1}};\theta,\sigma)-f_1(X_{t_{k-1}};\theta_0,\sigma)\right)^2}{g(X_{t_{k-1}};\sigma)^2}+\mathcal{R}(h;X_{t_{k-1}})\right)\mid X_{t_{k-1}}\right]\\
        &=\frac{1}{N}\sum_{k=1}^N\mathcal{R}(h;X_{t_{k-1}})\xrightarrow[\substack{N h\rightarrow\infty \\ h \rightarrow 0}]{\mathbb{P}_{\alpha_0}}0,
    \end{align*}
    so that we have the pointwise convergence by Lemma \ref{lemma:Jacob pointwise convergence}. By Assumption \ref{Assumption3} (ii)(iii) 
    on the derivatives of $v(x;\sigma)$ and $v(\varphi^{[1]}_h(x;\alpha);\sigma)$, we immediately have the uniform convergence using Lemma \ref{lemma:Gloter uniform convergence}. Now, combining the uniform convergence results for $U_1,U_2,U_3$, we have that
\begin{equation}\label{eqn:sigma convergence}
        \frac{1}{N}l_N(\theta,\sigma)\xrightarrow[\substack{N h\rightarrow\infty \\ h \rightarrow 0}]{\mathbb{P}_{\alpha_0}}G_1(\sigma,\sigma_0):=\int \left(\log|g(x;\sigma)|+\frac{g(x;\sigma_0)^2}{2g(x;\sigma)^2}\right)d\mu_0(x)
    \end{equation}
    uniformly in $\alpha$. Since the convergence in probability is equivalent to the existence of a subsequence converging almost surely, then there exists a subsequence $N_m, m\in \mathbb{N}$ 
    such that $N_m^{-1}l_{N_m}(\theta,\sigma)$ converges almost surely to the limit in \eqref{eqn:sigma convergence} a.s..
For a fixed $\omega\in \Omega$, thanks to the compactness of $\Theta$, there exists a subsequence $\hat{\theta}_{N_m},\hat{\sigma}_{N_m}$ tending to a limit $\alpha_\infty=(\theta_\infty,\sigma_\infty)$ as $m\to\infty$. Therefore, by \eqref{eqn:sigma convergence} and the continuity of the map $\sigma\mapsto G_1(\sigma,\sigma_0)$, we have
    \begin{equation*}
        \frac{1}{N_m}l_{N_m}(\hat{\theta}_{N_m},\hat{\sigma}_{N_m})(\omega)\mapsto G_1(\sigma_\infty,\sigma_0)
    \end{equation*}
    as $h\to0,Nh\to \infty$. By definition of $\hat{\alpha}_{N_k}$, we have
    \begin{eqnarray*}
        \frac{1}{N_m}l_{N_m}(\hat{\theta}_{N_m},\hat{\sigma}_{N_m})\leq \frac{1}{N_m}l_{N_m}(\hat{\theta}_{N_m},\sigma_0),
    \end{eqnarray*}
    so we get $G_1(\sigma_\infty,\sigma_0)\leq G_1(\sigma_0,\sigma_0)$. But on the other hand, since for all $y>0,y_0>0$, $\log y+\frac{y_0^2}{2y^2}\geq \log y_0+\frac{1}{2}$, by the identifiability assumption of the log-likelihood function we must have $\sigma_\infty=\sigma_0$. We therefore proved that any convergent subsequence of $\hat{\sigma}_{N}$ tends to $\sigma_0$, hence $\hat{\sigma}_N\xrightarrow[\substack{N h\rightarrow\infty \\ h \rightarrow 0}]{\mathbb{P}_{\alpha_0}}\sigma_0$.
    
    For the consistency of the drift, we focus on the uniform limit with respect to $\alpha$ of the following expression
    \begin{equation}
        \frac{1}{Nh}\left(l_N(\theta,\sigma)-l_N(\theta_0,\sigma)\right).
    \end{equation}
    Using the decomposition \eqref{eqn:U decomposition}, we have 
    \begin{equation*}
        \frac{1}{Nh}\left(l_N(\theta,\sigma)-l_N(\theta_0,\sigma)\right)=\frac{1}{2h}(U_2+2U_3).
    \end{equation*}
    Now, by \eqref{eqn:U_2 limit}, we immediately have that
    \begin{equation*}
         \frac{U_2}{h}\xrightarrow[\substack{N h\rightarrow\infty \\ h \rightarrow 0}]{\mathbb{P}_{\alpha_0}}\int\frac{\left(f_1(x;\theta,\sigma)-f_1(x;\theta_0,\sigma)\right)^2}{g(x;\sigma)^2}d\mu_0(x)
    \end{equation*}
    uniformly in $\alpha$. For $U_3/h$, by Lemma \ref{lemma:Z moment bounds} (i) and \eqref{eqn:vvarphi expand}
    \begin{align*}
        &\sum_{k=1}^N\mathbb{E}_{\alpha_0}\left[\frac{1}{Nh^2}Z_{t_k}(\theta_0,\sigma)\left(v\left(\varphi^{[1]}_h(X_{t_{k-1}};\theta_0,\sigma) ; \sigma \right)-v\left(\varphi^{[1]}_h(X_{t_{k-1}};\alpha) ; \sigma \right)\right)\mid X_{t_{k-1}}\right]\\
        &=\frac{1}{Nh^2}\sum_{k=1}^N \frac{h^2}{2}g^{\prime}(X_{t_{k-1}} ; \sigma)\left(1-\frac{g^2\left(X_{t_{k-1}} ; \sigma_0\right)}{g^2(X_{t_{k-1}} ; \sigma)}\right)\left(\frac{f_1\left(X_{t_{k-1}} ; \theta_0, \sigma\right)-f_1(X_{t_{k-1}} ; \theta, \sigma)}{g(X_{t_{k-1}} ;\sigma)}\right)+\mathcal{R}(h^3;X_{t_{k-1}})\\
    &\xrightarrow[\substack{N h\rightarrow\infty \\ h \rightarrow 0}]{\mathbb{P}_{\alpha_0}} \int g^{\prime}(x ; \sigma)\left(1-\frac{g^2\left(x ; \sigma_0\right)}{g^2(x ; \sigma)}\right)\left(\frac{f_1\left(x ; \theta_0, \sigma\right)-f_1(x ; \theta, \sigma)}{g(x ;\sigma)}\right)d\mu_0(x).
    \end{align*}
    For the second moment,
    \begin{align*}
         &\sum_{k=1}^N\mathbb{E}_{\alpha_0}\left[\frac{1}{N^2h^4}Z^2_{t_k}(\theta_0,\sigma)\left(v\left(\varphi^{[1]}_h(X_{t_{k-1}};\theta_0,\sigma) ; \sigma \right)-v\left(\varphi^{[1]}_h(X_{t_{k-1}};\alpha) ; \sigma \right)\right)^2\mid X_{t_{k-1}}\right]\\
         &=\frac{1}{N^2h^2}\sum_{k=1}^N\mathbb{E}_{\alpha_0}\left[Z_{t_k}(\theta_0,\sigma)^2\left(\frac{\left(f_1(X_{t_{k-1}};\theta,\sigma)-f_1(X_{t_{k-1}};\theta_0,\sigma)\right)^2}{g(X_{t_{k-1}};\sigma)^2}^+\mathcal{R}(h;X_{t_{k-1}})\right)\mid X_{t_{k-1}}\right]\\
        &=\frac{1}{N^2h^2}\sum_{k=1}^N\mathcal{R}(h;X_{t_{k-1}})\xrightarrow[\substack{N h\rightarrow\infty \\ h \rightarrow 0}]{\mathbb{P}_{\alpha_0}}0.
    \end{align*}
    Thus, by Lemma \ref{lemma:Jacob pointwise convergence}, we have the pointwise convergence for $U_3/h$. Again, using Assumption \ref{Assumption3} (ii)(iii) 
    on $v$ and $v(\varphi^{[1]}_h)$, and Lemma \ref{lemma:Gloter uniform convergence}, we have the uniform convergence with respect to $\alpha$. Therefore, we have the following uniform limit with respect to $\alpha$
    \begin{equation*}
        \frac{1}{Nh}\left(l_N(\theta,\sigma)-l_N(\theta_0,\sigma)\right)\xrightarrow[\substack{N h\rightarrow\infty \\ h \rightarrow 0}]{\mathbb{P}_{\alpha_0}}G_2(\theta_0,\sigma_0,\theta,\sigma),
    \end{equation*}
    where
    \begin{align*}
        G_2(\theta_0,\sigma_0,\theta,\sigma):=&\frac{1}{2}\int\frac{\left(f_1(x;\theta,\sigma)-f_1(x;\theta_0,\sigma)\right)^2}{g(x;\sigma)^2}d\mu_0(x)\\
        &+\int g^{\prime}(x ; \sigma)\left(1-\frac{g^2\left(x ; \sigma_0\right)}{g^2(x ; \sigma)}\right)\left(\frac{f_1\left(x ; \theta_0, \sigma\right)-f_1(x ; \theta, \sigma)}{g(x ;\sigma)}\right)d\mu_0(x).
    \end{align*}
    Again, for arbitrary $\omega\in \Omega$, there exists a subsequence $N_l$ such that $(\hat{\theta}_{N_l},\hat{\sigma}_{N_l})$ converges to a limit $(\theta_\infty,\sigma_\infty)$ almost surely. By the continuity of the mapping $(\theta,\sigma)\mapsto G_2(\theta_0,\sigma_0,\theta,\sigma)$, as $h\to0,Nh\to \infty$,
    \begin{equation*}
        \frac{1}{N_lh}\left(l_N(\hat{\theta}_{N_l},\hat{\sigma}_{N_l})-l_N(\theta_0,\hat{\sigma}_{N_l})\right)(\omega)\to G_2(\theta_0,\sigma_0,\theta_\infty,\sigma_\infty).
    \end{equation*}
    Since we have already proved that $\sigma_\infty=\sigma_0$, then the second integral in $G_2$ becomes zero, which means $G_2(\theta_0,\sigma_0,\theta,\sigma)\geq0$. On the other hand, by the definition of the estimator, $l_N(\hat{\theta}_{N_l},\hat{\sigma}_{N_l})-l_N(\theta_0,\hat{\sigma}_{N_l})\leq 0$. Therefore, by the identifiability condition in Assumption \ref{Assumption3} (vi), we have the consistency of the drift parameter, which concludes the proof.
\end{proof}
\subsection{Proof of Theorem \ref{thrm:asymptotic normality}} The proof of asymptotic normality relies on Lemma \ref{lemma:C alpha0 limit} and Lemma \ref{lemma:lambda limit} below. Establishing these Lemmas will allow us to prove Theorem \ref{thrm:asymptotic normality} following the classical route of \cite{kessler_estimation_1997,sorensen2003small}. As it is the case for Lemma \ref{lemma:Z moment bounds} and Theorem \ref{thrm:asymptotic consistency}, proving these lemmas is not immediate, as the underlying pseudo-likelihood arising from the LT splitting gives rise to terms depending jointly on $\sigma$ and $\theta$, requiring a different treatment than that considered in \cite{pilipovic_parameter_2024}, where a separation was possible. 
\begin{lemma}\label{lemma:C alpha0 limit}
    For $\mathbf{C}_N(\alpha_0)$ and $\mathbf{C}(\alpha_0)$ defined as \eqref{eqn:C_N alpha} and \eqref{eqn:C alpha0}, if $h\to0,\;Nh\to\infty,\;\rho_N\to0$, then we have 
    \begin{itemize}
        \item[(i)] $\mathbf{C}_N\left(\alpha_0\right) \xrightarrow{\mathbb{P}_{\alpha_0}}  \mathbf{C}\left(\alpha_0\right)$,
        \item[(ii)] $\sup _{\|\alpha\| \leq \rho_N}\left\|\mathbf{C}_N\left(\alpha_0+\alpha\right)-\mathbf{C}_N\left(\alpha_0\right)\right\| \xrightarrow{\mathbb{P}_{\alpha_0}} 0$.
    \end{itemize}
\end{lemma}
\begin{proof}
    Here, we use again the negative pseudo-log-likelihood $l_N(\theta,\sigma)$ defined earlier, and its decomposition \eqref{eqn:U decomposition} to analyse the partial derivatives that do not vanish when $h\to0, Nh\to \infty$. We also drop $x$ in the function $f(x;\alpha)$ when there is no ambiguity. We start with the first block of $\mathbf{C}_N$:
    \begin{equation*}
        \frac{1}{Nh}\partial_{\theta_i\theta_j}l_N(\theta,\sigma)=\frac{\partial_{\theta_i\theta_j}}{h}\left(U_2+2U_3\right).
    \end{equation*}
    By the definition of $U_2$ and \eqref{eqn:vvarphi expand}, 
    \begin{align*}
        \frac{U_2}{h}=&\frac{1}{N}\sum_{k=1}^N\frac{\left(f_1(X_{t_{k-1}};\theta,\sigma)-f_1(X_{t_{k-1}};\theta_0,\sigma)\right)^2}{g(X_{t_{k-1}};\sigma)^2}+\mathcal{R}(h;X_{t_{k-1}}),\nonumber\\
        \frac{\partial_{\theta_i}U_2}{2h}=&\frac{1}{N}\sum_{k=1}^N \frac{\left(f_1(\theta,\sigma)-f_1(\theta_0,\sigma)\right)\partial_{\theta_i}f_1(\theta,\sigma)}{g(\sigma)^2}+\mathcal{R}(h;X_{t_{k-1}}),\\
        \frac{\partial_{\theta_i\theta_j}U_2}{2h}=&\frac{1}{N}\sum_{k=1}^N\frac{\partial_{\theta_i}f_1(\theta,\sigma)\partial_{\theta_j}f_1(\theta,\sigma)+\left[f_1(\theta,\sigma)-f_1(\theta_0,\sigma)\right]\partial_{\theta_i\theta_j}f_1(\theta,\sigma)}{g(\sigma)^2}+\mathcal{R}(h;X_{t_{k-1}}).
    \end{align*}
    Then, by Lemma \ref{lemma:Kessler convergence} and the fact that $f_1(x;\theta,\sigma)=f(x;\theta)-\frac{1}{2}g(x;\sigma)g^\prime(x;\sigma)$, we have 
    \begin{equation}\label{eqn:U_2 theta theta}
        \left.\frac{\partial_{\theta_i\theta_j}U_2}{2h}\right|_{\alpha=\alpha_0}=\frac{1}{N}\sum_{k=1}^N \frac{\partial_{\theta_i}f(\theta_0)\partial_{\theta_j}f(\theta_0)}{g(\sigma)^2}\xrightarrow[\substack{N h\rightarrow\infty \\ h \rightarrow 0}]{\mathbb{P}_{\alpha_0}}\int\frac{\partial_{\theta_i}f(x;\theta_0)\partial_{\theta_j}f(x;\theta_0)}{g(x;\sigma)^2}d\mu_0(x).
    \end{equation}
    For $U_3$, we have 
    \begin{align*}
        \frac{U_3}{h}=&\frac{1}{Nh}\sum_{k=1}^N Z_{t_k}(\theta_0,\sigma)\left(\frac{f_1(X_{t_{k-1}};\theta,\sigma)-f_1(X_{t_{k-1}};\theta_0,\sigma)}{g(\sigma)}\right)+\mathcal{R}(h;X_{t_{k-1}})\\
        =&\frac{1}{Nh}\sum_{k=1}^N Z_{t_k}(\theta_0,\sigma)\left(\frac{f(\theta,\sigma)-f(\theta_0,\sigma)}{g(\sigma)}\right)+\mathcal{R}(h;X_{t_{k-1}})
        \end{align*}
        and
         \begin{align*}\frac{\partial_{\theta_i\theta_j}U_3}{h}=&\frac{1}{Nh}\sum_{k=1}^NZ_{t_k}(\theta_0,\sigma)\frac{\partial_{\theta_i\theta_j}f(\theta)}{g(\sigma)}.
    \end{align*}
We now use Lemma \ref{lemma:Jacob pointwise convergence} to establish the convergence. From Lemma \ref{lemma:Z moment bounds}, when $\sigma=\sigma_0$ and $\partial_{\theta_i\theta_j}f(\theta)$ is bounded in $\Theta$, we have 
    \begin{equation*}
        \mathbb{E}_{\alpha_0}\left[Z_{t_k}(\theta_0,\sigma_0)\frac{\partial_{\theta_i\theta_j}f(X_{t_{k-1}};\theta_0)}{g(X_{t_{k-1}};\sigma_0)}|X_{t_{k-1}}\right]=\mathcal{R}(h^2;X_{t_{k-1}}).
    \end{equation*}
    Then
    \begin{equation*}
        \frac{1}{Nh}\sum_{k=1}^N \mathbb{E}_{\alpha_0}\left[Z_{t_k}(\theta_0,\sigma_0)\frac{\partial_{\theta_i\theta_j}f(X_{t_{k-1}};\theta_0)}{g(X_{t_{k-1}};\sigma_0)}|X_{t_{k-1}}\right]\xrightarrow[\substack{N h\rightarrow\infty \\ h \rightarrow 0}]{}0.
    \end{equation*}
    Also, again by Lemma \ref{lemma:Z moment bounds}, we immediately  see that 
    \begin{equation*}
         \frac{1}{N^2h^2}\sum_{k=1}^N \mathbb{E}_{\alpha_0}\left[Z^2_{t_k}(\theta_0,\sigma_0)\left.\left(\frac{\partial_{\theta_i\theta_j}f(X_{t_{k-1}};\theta_0)}{g(X_{t_{k-1}};\sigma_0)}\right)^2\right|X_{t_{k-1}}\right]\xrightarrow[\substack{N h\rightarrow\infty \\ h \rightarrow 0}]{}0.
    \end{equation*}
    Then, we have that $\frac{\partial_{\theta_i\theta_j}U_3}{h}\xrightarrow[\substack{N h\rightarrow\infty \\ h \rightarrow 0}]{\mathbb{P}_{\alpha_0}}0$, which concludes the study of the limiting behaviour of the first block of $\mathbf{C}_N(\alpha_0)$.
    
    We now study $\partial_{\sigma_i\sigma_j}l_N(\theta,\sigma)/N$. Since $U_2$ and $U_3$ converge uniformly to $0$ in $\Theta$,   we only need to study the behaviour of the partial derivatives of the remaining two terms in \eqref{eqn:U decomposition}. 
    \begin{align}
        \frac{1}{N}\sum_{k=1}^N\partial_{\sigma_i\sigma_j}\log |g(X_{t_k};\sigma)|=&\frac{1}{N}\sum_{k=1}^N\frac{(\partial_{\sigma_i\sigma_j}g(\sigma))g(\sigma)-\partial_{\sigma_i}g(\sigma)\partial_{\sigma_j}g(\sigma)}{g^2(\sigma)}\nonumber\\
        &\xrightarrow[\substack{N h\rightarrow\infty \\ h \rightarrow 0}]{\mathbb{P}_{\alpha_0}}\int \frac{(\partial_{\sigma_i\sigma_j}g(x;\sigma_0))g(x;\sigma_0)-\partial_{\sigma_i}g(x;\sigma_0)\partial_{\sigma_j}g(x;\sigma_0)}{g^2(x;\sigma_0)}d\mu_0(x),\label{eqn:logg sigma sigma}
    \end{align}
    when evaluated at $\alpha_0$. For $U_1$, we have 
    \begin{equation*}
        \frac{\partial_{\sigma_i\sigma_j}U_1}{2}=\frac{1}{Nh}\sum_{k=1}^N\partial_{\sigma_i}Z_{t_{k}}(\theta_0,\sigma)\partial_{\sigma_j}Z_{t_{k}}(\theta_0,\sigma)+Z_{t_k}(\theta_0,\sigma)\partial_{\sigma_i\sigma_j}Z_{t_k}(\theta_0,\sigma).
    \end{equation*}
    From \eqref{eqn:vvarphi expand}, we know that 
    \begin{align*}
        \partial_{\sigma_i}Z_{t_k}(\theta_0,\sigma)\partial_{\sigma_j}Z_{t_k}(\theta_0,\sigma)=&\left(\partial_{\sigma_i}v(X_{t_k};\sigma)-\partial_{\sigma_i}v(X_{t_{k-1}};\sigma)\right)\left(\partial_{\sigma_j}v(X_{t_k};\sigma)-\partial_{\sigma_j}v(X_{t_{k-1}};\sigma)\right)\\
        &+\mathcal{R}(h,X_{t_{k-1}})\left(\partial_{\sigma_i}v(X_{t_k};\sigma)-\partial_{\sigma_i}v(X_{t_{k-1}};\sigma)\right)\\
        &+\mathcal{R}(h,X_{t_{k-1}})\left(\partial_{\sigma_j}v(X_{t_k};\sigma)-\partial_{\sigma_j}v(X_{t_{k-1}};\sigma)\right)+\mathcal{R}(h^2;X_{t_{k-1}}).
    \end{align*}
    To simplify the notation in the above expression, we introduce the notation $\gamma_1(X_{t_k};\sigma)=\partial_{\sigma_i}v(X_{t_k};\sigma)-\partial_{\sigma_i}v(X_{t_{k-1}};\sigma)$ and $\gamma_2(X_{t_k};
    \sigma)=\partial_{\sigma_j}v(X_{t_k};\sigma)-\partial_{\sigma_j}v(X_{t_{k-1}};\sigma)$. 
Note that, when we take the one-step conditional mean, the components with $X_{t_{k-1}}$ in $\gamma_1, \gamma_2$ become constant. We also need the following two identities,
\begin{equation}\label{eqn:v sigma derivatives}
        \partial_{\sigma_i x}v(x;\sigma)=-\frac{\partial_{\sigma_i}g(x;\sigma)}{g^2(x;\sigma)},\quad \partial_{\sigma_i\sigma_jx}=2\frac{\partial_{\sigma_i}g(x;\sigma)\partial_{\sigma_j}g(x;\sigma)}{g^3(x;\sigma)}-\frac{\partial_{\sigma_i\sigma_j}g(x;\sigma)}{g^2(x;\sigma)}
    \end{equation}
   Using Lemma \ref{lemma:expectation expansion}, the product formula \eqref{eqn:product rule} and \eqref{eqn:v sigma derivatives}, we get
    \begin{align*}
        &\mathbb{E}_{\alpha_0}\left[\left.\partial_{\sigma_i}Z_{t_k}(\theta_0,\sigma)\partial_{\sigma_j}Z_{t_k}(\theta_0,\sigma)\right|X_{t_{k-1}}\right]\\
        =&\mathbb{E}_{\alpha_0}\left[\left.\left(\partial_{\sigma_i}v(X_{t_k};\sigma)-\partial_{\sigma_i}v(X_{t_{k-1}};\sigma)\right)\left(\partial_{\sigma_j}v(X_{t_k};\sigma)-\partial_{\sigma_j}v(X_{t_{k-1}};\sigma)\right)\right|X_{t_{k-1}}\right]\\
        &+\mathcal{R}(h^2,X_{t_{k-1}})\\
        =&h\left[\gamma_1(X_{t_{k-1}};\sigma)\mathcal{L}_{\alpha_0}\gamma_2(X_{t_{k-1}};\sigma)+\gamma_2(X_{t_{k-1}};\sigma)\mathcal{L}_{\alpha_0}\gamma_1(X_{t_{k-1}};\sigma)\right.\\
        &\left.+g^2(X_{t_{k-1}};\sigma)\gamma^\prime_i(X_{t_{k-1}};\sigma)\gamma^\prime_{j}(X_{t_{k-1}};\sigma)\right]+\mathcal{R}(h^2,X_{t_{k-1}})\\
        =&h\frac{\partial_{\sigma_i}g(X_{t_{k-1}};\sigma)\partial_{\sigma_j}g(X_{t_{k-1}};\sigma)}{g^2(X_{t_{k-1}};\sigma)}+\mathcal{R}(h^2,X_{t_{k-1}}).
    \end{align*}
We then have 
    \begin{equation*}
        \frac{1}{N^2h^2}\sum_{k=1}^N\mathbb{E}_{\alpha_0}\left[ \left.\left(\partial_{\sigma_i}Z_{t_k}(\theta_0,\sigma)\partial_{\sigma_j}Z_{t_k}(\theta_0,\sigma)\right)^2\right| X_{t_{k-1}}\right]=\frac{1}{N^2h^2}\sum_{k=1}^N \mathcal{R}(h,X_{t_{k=1}})\xrightarrow[\substack{N h\rightarrow\infty \\ h \rightarrow 0}]{\mathbb{P}_{\alpha_0}}0.
    \end{equation*}
    Then by Lemma \ref{lemma:Jacob pointwise convergence},
    \begin{equation}\label{eqn:Z sigma Z sigma}
        \left.\frac{1}{Nh}\sum_{k=1}^N\partial_{\sigma_i}Z_{t_{k}}(\theta_0,\sigma)\partial_{\sigma_j}Z_{t_{k}}(\theta_0,\sigma)\right|_{\alpha=\alpha_0}\xrightarrow[\substack{N h\rightarrow\infty \\ h \rightarrow 0}]{\mathbb{P}_{\alpha_0}}\int \frac{\partial_{\sigma_i}g(x;\sigma_0)\partial_{\sigma_j}g(x;\sigma_0)}{g^2(x;\sigma_0)}d\mu_0(x).
    \end{equation}
   If we denote $\beta_1(X_{t_{k}};\sigma)=v(X_{t_k};\sigma)-v(X_{t_{k-1}};\sigma)$ and $\beta_2(X_{t_k};\sigma)=\partial_{\sigma_i\sigma_j}v(X_{t_k};\sigma)-\partial_{\sigma_i\sigma_j}v(X_{t_{k-1}};\sigma)$, by a similar argument and \eqref{eqn:v sigma derivatives}, we have that
    \begin{align*}
        &\mathbb{E}_{\alpha_0}\left[\left.Z_{t_k}(\theta_0,\sigma)\partial_{\sigma_i\sigma_j}Z_{t_k}(\theta_0,\sigma)\right|X_{t_{k-1}}\right]\\
        =&h\left[\beta_1(X_{t_{k-1}};\sigma)\mathcal{L}_{\alpha_0}\beta_2(X_{t_{k-1}};\sigma)+\beta_2(X_{t_{k-1}};\sigma)\mathcal{L}_{\alpha_0}\beta_1(X_{t_{k-1}};\sigma)\right.\\
        &\left.+g^2(X_{t_{k-1}};\sigma)\beta_1^\prime(X_{t_{k-1}};\sigma)\beta_2^\prime(X_{t_{k-1}};\sigma)\right]+\mathcal{R}(h^2;X_{t_{k-1}})\\
        =&h\left[\frac{2\partial_{\sigma_i}g(X_{t_{k-1}};\sigma)\partial_{\sigma_j}g(X_{t_{k-1}};\sigma)}{g^2(X_{t_{k-1}};\sigma)}-\frac{\partial_{\sigma_i\sigma_j}g(X_{t_{k-1}};\sigma)}{g(X_{t_{k-1}};\sigma)}\right]+\mathcal{R}(h^2;X_{t_{k-1}}),
    \end{align*}
    and
    \begin{equation*}
    \frac{1}{N^2h^2}\sum_{k=1}^N\mathbb{E}_{\alpha_0}\left[ \left.\left(Z_{t_k}(\theta_0,\sigma)\partial_{\sigma_i\sigma_j}Z_{t_k}(\theta_0,\sigma)\right)^2\right| X_{t_{k-1}}\right]=\frac{1}{N^2h^2}\sum_{k=1}^N \mathcal{R}(h,X_{t_{k=1}})\xrightarrow[\substack{N h\rightarrow\infty \\ h \rightarrow 0}]{\mathbb{P}_{\alpha_0}}0.
    \end{equation*}
    Then, by Lemma \ref{lemma:Jacob pointwise convergence}
    \begin{align}\label{eqn:Z Z sigma sigma}
         &\left.\frac{1}{Nh}\sum_{k=1}^N Z_{t_{k}}(\theta_0,\sigma)\partial_{\sigma_i\sigma_j}Z_{t_{k}}(\theta_0,\sigma)\right|_{\alpha=\alpha_0}\nonumber\\
    \xrightarrow[\substack{N h\rightarrow\infty \\ h \rightarrow 0}]{\mathbb{P}_{\alpha_0}}&\int \left(\frac{2\partial_{\sigma_i}g(x;\sigma_0)\partial_{\sigma_j}g(x;\sigma_0)}{g^2(x;\sigma_0)}-\frac{\partial_{\sigma_i\sigma_j}g(x;\sigma_0)}{g(x;\sigma_0)}\right)d\mu_0(x).
    \end{align}
    Combining \eqref{eqn:logg sigma sigma}, \eqref{eqn:Z sigma Z sigma} and \eqref{eqn:Z Z sigma sigma} yields $\left[\mathbf{C}_{\sigma}(\alpha_0)\right]_{j_1,j_2}$. \\
    For $\partial_{\theta\sigma}l_N(\alpha)/(N\sqrt{h})$, it can be easily shown that all the relevant terms tend to $0$ as $h\to 0,\; Nh\to \infty$. For $U_2/\sqrt{h}$, by \eqref{eqn:vvarphi expand} the leading coefficient becomes $\sqrt{h}/N$, which tends to $0$. For $U_3/\sqrt{h}$, we know that for the one-step conditional mean, both $Z_{t_k}(\alpha_0)$ and $\partial_{\sigma}Z_{t_k}(\alpha_0)$ are at least of order $h$, therefore the leading coefficient is again $\sqrt{h}/N$, which tends to $0$. The second part of the lemma is straightforward, since all  limits are continuous in $\alpha$.
\end{proof}
\begin{lemma}\label{lemma:lambda limit}
    If $h\to0,\;Nh\to0$ and $Nh^2\to 0$, then $\boldsymbol{\lambda}_N \xrightarrow{d} \mathcal{N}(\mathbf{0},\mathbf{C}(\alpha_0))$ under $\mathbb{P}_{\alpha_0}$.
\end{lemma}

\begin{proof}
We first compute the partial derivatives of $l_N(\theta,\sigma)$. From the proof of Lemma \ref{lemma:C alpha0 limit} and the form of $U_i$, we immediately have
    \begin{eqnarray*}
&&\frac{\partial_{\theta_i}l)_B(N,\theta,\sigma)}{\sqrt{Nh}}=\frac{1}{\sqrt{Nh}}\sum_{k=1}^N Z_{t_k}(\theta_0,\sigma)\frac{\partial_{\theta_i}f(X_{t_{k-1}};\theta)}{g(X_{t_{k-1}};\sigma)}, \\&&       \frac{\partial_{\sigma_i}l(N,\theta,\sigma)}{\sqrt{N}}=\frac{1}{\sqrt{N}}\sum_{k=1}^N\left[\frac{\partial_{\sigma_i}g(X_{t_{k}};\sigma)}{g(X_{t_k};\sigma)}+\frac{1}{h}Z_{t_k}(\alpha)\left(\partial_{\sigma_i}v(X_{t_k};\sigma)-\partial_{\sigma_i}v(X_{t_k};\sigma)+\mathcal{R}(h;X_{t_{k-1}}) \right)\right].
    \end{eqnarray*}
    We define
    \begin{align}
        A^{(i)}_{N,k}(\alpha_0)=&\frac{1}{\sqrt{Nh}}Z_{t_k}(\alpha_0)\frac{\partial_{\theta_i}f(X_{t_{k-1}};\theta_0)}{g(X_{t_{k-1}};\sigma_0)},\label{eqn:A_ink}\\
        B^{(i)}_{N,k}(\alpha_0)=&\frac{1}{\sqrt{N}}\left[\frac{\partial_{\sigma_i}g(X_{t_{k}};\sigma_0)}{g(X_{t_k};\sigma_0)}+\frac{1}{h}Z_{t_k}(\alpha_0)\left(\partial_{\sigma_i}v(X_{t_k};\sigma_0)-\partial_{\sigma_i}v(X_{t_k};\sigma_0)+\mathcal{R}(h;X_{t_{k-1}}) \right)\right].\label{eqn:B_ink}
    \end{align}
    Then,  $\mathbf{\lambda}_N=\sum_{k=1}^N\left[A^{(1)}_{N,k}(\alpha_0),\dots,A^{(d_1)}_{N,k}(\alpha_0),B^{(1)}_{N,k}(\alpha_0),\dots,B^{(d_2)}_{N,k}(\alpha_0)\right]^T$. To establish the two conditions of Lemma \ref{prop:triangular asymptotic normal}, by Remarks S1 and S2 in \cite{pilipovic_parameter_2024}, it is sufficient to prove the following lemma.
\end{proof}

\begin{lemma}\label{lemma:normality conditions}
    Let $A^{(i)}_{N,k}(\alpha_0)$ and $B^{(i)}_{N,k}(\alpha_0)$ be defined as \eqref{eqn:A_ink} and \eqref{eqn:B_ink}, respectively. Then, the followings hold for all $i,i_1,i_2=1,2,\dots,d_1$ and $j,j_1,j_2=1,2,\dots,d_2$, if $h\to 0,\;Nh\to \infty,\; Nh^2\to 0$.
    \begin{itemize}
        \item[(i)]$\mathbb{E}_{\alpha_0}\left[\sup _{1 \leq k \leq N}\left|A_{N, k}^{(i)}\left(\alpha_0\right)\right|\right] \longrightarrow 0$, and $\mathbb{E}_{\alpha_0}\left[\sup _{1 \leq k \leq N}\left|B_{N, k}^{(j)}\left(\alpha_0\right)\right|\right] \longrightarrow 0 $;
        \item[(ii)]$\sum_{k=1}^N \mathbb{E}_{\alpha_0}\left[A_{N, k}^{(i)}\left(\alpha_0\right) \mid X_{t_{k-1}}\right] \xrightarrow{\mathbb{P}_{\alpha_0}} 0$, and $\sum_{k=1}^N \mathbb{E}_{\alpha_0}\left[B_{N, k}^{(j)}\left(\alpha_0\right) \mid X_{t_{k-1}}\right] \xrightarrow{\mathbb{P}_{\alpha_0}} 0 $;
        \item[(iii)]$\sum_{k=1}^N \mathbb{E}_{\alpha_0}\left[A_{N, k}^{\left(i_1\right)}\left(\alpha_0\right) \mid X_{t_{k-1}}\right] \mathbb{E}_{\alpha_0}\left[A_{N, k}^{\left(i_2\right)}\left(\alpha_0\right) \mid X_{t_{k-1}}\right] \xrightarrow{\mathbb{P}_{\alpha_0}} 0$;
        \item[(iv)] $\sum_{k=1}^N \mathbb{E}_{\alpha_0}\left[B_{N, k}^{\left(j_1\right)}\left(\alpha_0\right) \mid X_{t_{k-1}}\right] \mathbb{E}_{\alpha_0}\left[B_{N, k}^{\left(j_2\right)}\left(\alpha_0\right) \mid X_{t_{k-1}}\right] \xrightarrow{\mathbb{P}_{\alpha_0}} 0$;
        \item[(v)]$\sum_{k=1}^N \mathbb{E}_{\alpha_0}\left[A_{N, k}^{\left(i_1\right)}\left(\alpha_0\right) \mid X_{t_{k-1}}\right] \mathbb{E}_{\alpha_0}\left[B_{N, k}^{\left(j_2\right)}\left(\alpha_0\right) \mid X_{t_{k-1}}\right] \xrightarrow{\mathbb{P}_{\alpha_0}} 0$;
        \item[(vi)]$\sum_{k=1}^N \mathbb{E}_{\alpha_0}\left[A_{N, k}^{\left(i_1\right)}\left(\alpha_0\right) A_{N, k}^{\left(i_2\right)}\left(\alpha_0\right) \mid X_{t_{k-1}}\right] \xrightarrow{\mathbb{P}_{\alpha_0}}\left[\mathbf{C}_{\theta}(\alpha_0)\right]_{i_1i_2}$;
        \item[(vii)] $\sum_{k=1}^N \mathbb{E}_{\alpha_0}\left[B_{N, k}^{\left(j_1\right)}\left(\alpha_0\right) B_{N, k}^{\left(j_2\right)}\left(\alpha_0\right) \mid X_{t_{k-1}}\right] \xrightarrow{\mathbb{P}_{\alpha_0}}\left[\mathbf{C}_{\sigma}(\alpha_0)\right]_{j_1j_2}$;
        \item[(viii)] $\sum_{k=1}^N \mathbb{E}_{\alpha_0}\left[A_{N, k}^{\left(i\right)}\left(\alpha_0\right) B_{N, k}^{\left(j\right)}\left(\alpha_0\right) \mid X_{t_{k-1}}\right] \xrightarrow{\mathbb{P}_{\alpha_0}}0$;
        \item[(ix)] $\sum_{k=1}^N \mathbb{E}_{\alpha_0}\left[\left(A_{N, k}^{\left(i_1\right)}\left(\alpha_0\right) A_{N, k}^{\left(i_2\right)}\left(\alpha_0\right)\right)^2 \mid X_{t_{k-1}}\right] \xrightarrow{\mathbb{P}_{\alpha_0}} 0$;
        \item[(x)] $\sum_{k=1}^N \mathbb{E}_{\alpha_0}\left[\left(B_{N, k}^{\left(j_1\right)}\left(\alpha_0\right) B_{N, k}^{\left(j_2\right)}\left(\alpha_0\right)\right)^2 \mid X_{t_{k-1}}\right] \xrightarrow{\mathbb{P}_{\alpha_0}} 0$;
        \item[(xi)]  $\sum_{k=1}^N \mathbb{E}_{\alpha_0}\left[\left(A_{N, k}^{\left(i\right)}\left(\alpha_0\right) B_{N, k}^{\left(j\right)}\left(\alpha_0\right)\right)^2 \mid X_{t_{k-1}}\right] \xrightarrow{\mathbb{P}_{\alpha_0}} 0$.
    \end{itemize}
\end{lemma}
\begin{proof}
Throughout this proof, to shorten the notation, we omit $x$ in $\mathcal{R}(h;x)$. 
    We start with (i). For $A^{(i)}_{N,k}$, since by Lemma \ref{eqn:Zu expectation} the expectation of $Z_{t_k}(\alpha_0)$ is of order $h^2$, then the expectations of $A^{(i)}_{N,k}(\alpha_0)$ are bounded by $\mathcal{R}(\frac{h^{3/2}}{\sqrt{N}})$, which converge to $0$. The reasoning for $B^{(i)}_{N,k}$ is similar: since we assume $\frac{\partial_{\sigma_i} g\left(X_{t_k} ; \sigma_0\right)}{g\left(X_{t_k} ; \sigma_0\right)}$ is bounded, all terms in $B^{(i)}_{N,k}(\alpha_0)$ are either $\mathcal{R}(\frac{1}{\sqrt{N}})$ or $\mathcal{R}(\sqrt{\frac{h}{N}})$, which converge to $0$.
    
    For (ii), we first calculate the expectation of $A^{(i)}_{N,k}(\alpha_0)$ and $ B^{(i)}_{N,k}(\alpha_0)$. For $A^{(i)}_{N,k}(\alpha_0)$,\\   $\mathbb{E}_{\alpha_0}\left[\left.A^{(i)}_{N,k}(\alpha_0)\right| X_{t_{k-1}}\right]=\mathcal{R}(h^{3/2}/\sqrt{N})$, so
   \begin{equation*}
        \sum_{k=1}^N \mathbb{E}_{\alpha_0}\left[A_{N, k}^{(i)}\left(\alpha_0\right) \mid X_{t_{k-1}}\right]=\mathcal{R}(\sqrt{Nh^3})\xrightarrow{\mathbb{P}_{\alpha_0}} 0.
    \end{equation*}
    For $B^{(i)}_{N,k}(\alpha_0)$, using \eqref{eqn:v sigma derivatives}, we have 
    \begin{align*}
        &\mathbb{E}_{\alpha_0}\left[\left.B^{(i)}_{N,k}(\alpha_0)\right| X_{t_{k-1}}\right] =\mathbb{E}_{\alpha_0}\left[\left.\frac{1}{\sqrt{N}h}\left(Z_{t_k}(\alpha_0)\partial_{\sigma_i}v(X_{t_k};\sigma_0)\right)+\frac{1}{\sqrt{N}}\frac{\partial_{\sigma_i}g(X_{t_k};\sigma_0)}{g(X_{t_k};\sigma_0)}\right|X_{t_{k-1}}\right]+\mathcal{R}\left(\frac{h}{\sqrt{N}}\right)\\
        =&\frac{1}{\sqrt{N}}\left(g^2(X_{{t_{k-1}}};\sigma_0)\partial_xg(X_{t_{k-1}};\sigma_0)\partial_{\sigma_ix}g(X_{t_{k-1}};\sigma_0)+\frac{\partial_{\sigma_i}g(X_{t_{k-1}};\sigma_0)}{g(X_{t_{k-1}};\sigma_0)}\right)+\mathcal{R}\left(\frac{h}{\sqrt{N}}\right)=\mathcal{R}\left(\frac{h}{\sqrt{N}}\right),
    \end{align*}
    implying
    \begin{equation*}
        \sum_{k=1}^N \mathbb{E}_{\alpha_0}\left[\left.B^{(i)}_{N,k}(\alpha_0)\right| X_{t_{k-1}}\right]=\mathcal{R}(\sqrt{Nh^2})\xrightarrow{\mathbb{P}_{\alpha_0}} 0,
    \end{equation*}
    which concludes (ii). Then (iii)-(v) follow immediately.
    
    For (vi), again from Lemma \ref{lemma:Z moment bounds}, we know that we have 
    \begin{eqnarray*}
        &&\sum_{k=1}^N \mathbb{E}_{\alpha_0}\left[A_{N, k}^{\left(i_1\right)}\left(\alpha_0\right) A_{N, k}^{\left(i_2\right)}\left(\alpha_0\right) \mid X_{t_{k-1}}\right]
        =\frac{1}{Nh}\sum_{k=1}^N\mathbb{E}_{\alpha_0}\left[\left. Z^2_{t_k}(\alpha_0)\frac{\partial_{\theta_i}f(X_{t_{k-1}};\theta_0)\partial_{\theta_j}f(X_{t_{k-1}};\theta_0)}{g^2(X_{t_{k-1}};\sigma_0)}\right|X_{t_{k-1}}\right]\\
        &=&\frac{1}{N}\sum_{k=1}^N\frac{\partial_{\theta_i}f(X_{t_{k-1}};\theta_0)\partial_{\theta_j}f(X_{t_{k-1}};\theta_0)}{g^2(X_{t_{k-1}};\sigma_0)}+\mathcal{R}\left(\frac{h}{N}\right)\xrightarrow{\mathbb{P}_{\alpha_0}}\left[\mathbf{C}_{\theta}(\alpha_0)\right]_{i_1,i_2}.
    \end{eqnarray*}
    Similarly, using \eqref{eqn:product rule} and \eqref{eqn:v sigma derivatives}, we have 
    \begin{align*}
        &\sum_{k=1}^N \mathbb{E}_{\alpha_0}\left[B_{N, k}^{\left(j_1\right)}\left(\alpha_0\right) B_{N, k}^{\left(j_2\right)}\left(\alpha_0\right) \mid X_{t_{k-1}}\right]\\
        =&\sum_{k=1}^N\mathbb{E}_{\alpha_0}\left[\frac{1}{Nh^2}Z^2_{t_k}(\alpha_0)\left(\partial_{\sigma_i}v(X_{t_k};\sigma_0)-\partial_{\sigma_i}v(X_{t_{k-1}};\sigma_0)\right)\left(\partial_{\sigma_j}v(X_{t_k};\sigma_0)-\partial_{\sigma_j}v(X_{t_{k-1}};\sigma_0)\right)\right]\\
        &+\frac{1}{N}\frac{\partial_{\sigma_i}g(X_{t_{k-1}};\sigma_0)\partial_{\sigma_j}g(X_{t_{k-1}};\sigma_0)}{g^2(X_{t_{k-1}})}+\mathcal{R}\left(\frac{h}{N}\right)    \xrightarrow{\mathbb{P}_{\alpha_0}}\left[\mathbf{C}_{\sigma}(\alpha_0)\right]_{j_1,j_2}.
    \end{align*}
    The proof of (viii)-(xi) follows the same logic, with repeated application of Lemma \ref{lemma:Z moment bounds}, \eqref{eqn:product rule} and \eqref{eqn:v sigma derivatives}.
\end{proof}

\section{Numerical results}\label{AppendixC}
\begin{table}
\tiny
    \centering
        \begin{tabular}{ccc}
        \toprule
             & {\bf OU} - $g(x)=\sqrt{2\theta }$ & {\bf IGBM} -$g(x)=\sqrt{2\theta a}x$ \\\midrule
        $(\tilde\mu,\tilde\theta)$& &$\left(\frac{\mu}{1+a},(1+a)\theta\right)$\\ \midrule
        Extra& $A^{\textrm{LT}}=1,A^{\textrm{S}}=e^{\frac{-\theta h}{2}}$&$A^{\textrm{LT}}=(e^{\theta ah}-e^{-\theta h}), A^{\mathrm{S}}=\left(1+e^{-\frac{\tilde\theta h}{2}}+e^{\theta ah-\frac{\tilde\theta h}{2}}-e^{-\theta h}\right)$\\ \midrule
        $X^{[2]}_{t_{k+1}}=\varphi^{[2]}_h(x)$&  $x+\sqrt{2\theta}\xi_k$&$xe^{\sigma \xi_k}$\\ \midrule
        $\mathbb{E}[X^{\bullet}_{t_{k+1}}|X_{t_k}]$&$\mu+e^{-\theta h}(X_{t_k}-\mu)$&$e^{-\theta h}X_{t_k}+\tilde\mu A^\bullet$\\ \midrule
        $\mathbb{E}[X^\bullet_{t}|X_0]$&$e^{-\theta t}X_0+\mu(1-e^{-\theta t})$&$A^{\bullet}\frac{1-e^{-\theta t}}{1-e^{-\theta h}}$\\ \midrule
        $v(x);v^{-1}(x)$&&$\frac{\ln(x)}{\sqrt{2\theta a}},e^{\sqrt{2\theta a}x}$\\ 
        \bottomrule
        \end{tabular}
    \vspace{.5cm}
        \begin{tabular}{ccc}
        \toprule
    & {\bf F diffusion} - $g(x)=\sqrt{2a\theta x(x+1)}$ & {\bf Wright-Fisher} $g(x)=\sqrt{2a\theta x(x-1)}, a<0$ \\ \midrule
    $(\tilde\mu,\tilde\theta)$ & $(\frac{2\mu-a}{2(1+a)},(1+a)\theta)$& $(\frac{2\mu+a}{2(1+a)},(1+a)\theta)$\\ \midrule
    Extra&$l(x,h)=\sqrt{1+\tilde{\mu}+e^{-\tilde{\theta}h}\left(x-\tilde{\mu}\right)}-\sqrt{\tilde{\mu}+e^{-\tilde{\theta}h}\left(x-\tilde{\mu}\right)}$& $A^{\textrm{LT}}=\tilde{\mu}\left(e^{a\theta h}-e^{-\theta h}\right)+\frac{1}{2}\left(1-e^{a\theta h}\right)$,\\
&&{$A^{\textrm{S}}=\tilde{\mu}\left[1-e^{-\theta h}-e^{-\frac{\tilde{\theta}h}{2}}(1-e^{a\theta h})\right]+\frac{e^{-\frac{\tilde{\theta}h}{2}}}{2}
\left(1-e^{a\theta h}\right)$}\\ \midrule
$\varphi^{[2]}_{h}(x)$& $\sinh^2\left(\frac{\sqrt{2a\theta}}{2}{\xi_k}-\log(\sqrt{1+x}-\sqrt{x})\right)$&$\sin^2\left(\frac{\sqrt{-2a\theta}}{2}\xi_k+\arcsin\sqrt{x}\right)$\\ \midrule
$\mathbb{E}[X^{\textrm{LT}}_{t_{k+1}}|X_{t_k}]$& $\frac{1}{2}e^{a\theta h}\left(\frac{1}{l^2(X_{t_k},h)}+l^2(X_{t_k},h)\right)-\frac{1}{2}$&$e^{-\theta h}X_{t_k}+A^{\textrm{LT}}$\\ \midrule
$\mathbb{E}[X^{\textrm{S}}_{t_{k+1}}|X_{t_k}]$& $\tilde{\mu}+e^{-\frac{\tilde{\theta}h}{2}}\left(\frac{1}{2}e^{a\theta h}\left(\frac{1}{l^2(X_{t_k},\frac{h}{2})}+l^2(X_{t_k},\frac{h}{2})\right)-\frac{1}{2}-\tilde{\mu}\right)$&$e^{-\theta h}X_{t_k}+ A^{\textrm{S}}$\\ \midrule
$\mathbb{E}[X^{\bullet}_{t}|X_{0}]$&& $e^{-\theta t}X_0+A^{\bullet}\frac{1-e^{-\theta t}}{1-e^{-\theta h}}$ \\ \midrule
$v(x); v^{-1}(x)$&$\frac{-2}{\sqrt{2\theta a}}\ln(\sqrt{1+x}-\sqrt{x});\quad \sinh^2(\frac{\sqrt{2\theta a}x}{2})$&$\frac{2\arcsin (\sqrt{x})}{\sqrt{-2a\theta}};\quad \sin^2\left(\frac{\sqrt{-2a\theta}}{2}x\right)$\\ 
        \bottomrule
        \end{tabular}
    \vspace{.5cm}
        \begin{tabular}{cccc}
        \toprule
    & {\bf Ginzburg-Landau} - $g(x)= \sigma x$& {\bf Verhulst} - $g(x)=\sigma x$\\ \midrule
$X^{[1]}_{t_{k+1}}=\varphi^{[1]}_h(x)$&$\frac{x}{\sqrt{2x^2\lambda h+1}}$&$\frac{\eta e^{\eta h} x}{\eta-\lambda x+\lambda e^{\eta h} x}$\\ \midrule
$X^{[2]}_{t_{k+1}}=\varphi_h^{[2]}(x)$&$xe^{\sigma \xi_k}$&$xe^{\sigma \xi_k}$\\ \midrule
    $v(x); v^{-1}(x)$&$\frac{\ln(x)}{\sigma}; \quad e^{\sigma x}$& $\frac{\ln(x)}{\sigma};\quad e^{\sigma x}$\\ \midrule  
    $\varphi^{-1}_h(x)$&$\sqrt{\frac{x^2}{1-2\lambda hx^2}}$& $\frac{\eta x}{\eta e^{\eta h}-\lambda x\left(e^{\eta h}-1\right)}$\\ \midrule
    $\frac{d\varphi^{-1}_h(x)}{dx}$&$\frac{1}{\left(1-2\lambda h x^2\right)^{3/2}}$&$\frac{\eta^2 e^{\eta h}}{\left(\eta e^{\eta h}-\lambda x\left(e^{\eta h}-1\right)\right)^2}$\\ \midrule
    \bottomrule
    \end{tabular}    \caption{
 Characterisation of the derived LT  \eqref{eqn:LT subsystem} and S \eqref{eqn:Strang subsystem} splitting schemes for the remaining four members of the Pearson diffusion class \eqref{eqn:Pearson diffusion} and additional considered models. Here, $\varphi^{[1]}_h$ and $\varphi^{[2]}_h$ represent the one-step solutions to the subequations \eqref{eqn:ODE1}-\eqref{eqn:SDE2}. The one-step conditional means are also reported, together with the whole conditional mean, when explicitly derivable. We refer to \cite{forman_pearson_2008} for the invariant distributions and additional properties of each Pearson diffusion process. Here, $\xi_k\sim \mathcal{N}(0,h)$, $k=0,\ldots, N-1$ are iid normal increments.}
    \label{TablePearson2}
\end{table}
Here, we first explicitly derive the LT and S splitting schemes for each of the considered models, investigating their properties and theoretically computing their one-step conditional means. Then, we numerically validate some of the theoretical properties derived, comparing the proposed splitting schemes with competing methods in terms of mean-square convergence order, and preservation of the state space.

\subsection{Splitting schemes: Derivation and properties for the Pearson diffusion models}
\label{PearsonAppendix}
As we previously analysed the CIR and the Student diffusion in Section \ref{resultsPearson}, we consider here the remaining members of the Pearson diffusion class, listing them as they are presented in \cite{forman_pearson_2008}. As written in Section \ref{resultsPearson}, the one-step solution of the ODE subequation \eqref{eqn:ODE1} is given by \eqref{ODEsolPearson}, with $\tilde\theta=\theta(a+1),\tilde\mu=(\mu-b/2)/(a+1)$, while the solution to the SDE subequation \eqref{eqn:SDE2}, depending on the specific $g(x)$, is reported in Table \ref{TablePearson2}, together with the one-step and global conditional means, when explicitly available. Composing these solutions lead to the LT and S splitting schemes below.

\subsubsection{Ornstein-Uhlenbeck (OU) process - $g(x)=\sqrt{2\theta }$} 
The OU transition density is well known, given by
\begin{equation*}
    X_t|X_0\sim \mathcal{N}\left(e ^{-\theta t}X_0+\mu\left(1-e ^{-\theta t}\right),1-e^{-2\theta t}\right).
\end{equation*}
 While exact simulation is possible, we still derive the LT and S splitting schemes
\begin{align*}
    X^{\bullet}_{t_{k+1}}=&\mu+e^{-\theta h}\left(X_{t_k}-\mu\right)+\sqrt{2\theta}
    A^{\bullet}\xi_k,\qquad \bullet=\textrm{LT, S},
\end{align*}
 to compute their properties, 
with $A^{\textrm{LT}}=1, A^{\textrm{S}}=e^{-\theta h/2}$. 
For $t=t_{k},\; k\leq N$, by back i\-te\-ra\-ting, we obtain
\[
    X^{\bullet}_{t}=e^{-\theta t}X_0+\mu\left(1-e^{-\theta t}\right)+\sqrt{2\theta}\sum_{j=0}^{k-1}e^{-\theta jh}A^{\bullet}\xi_j,
\]
leading to
\[
    X^{\bullet{}}_t|X_0\sim \mathcal{N}\left(e^{-\theta t}X_0+\mu\left(1-e^{-\theta t}\right),\frac{2\theta h \left(1-e^{-2\theta t}\right)}{1-e^{-2\theta h}}A^{\bullet}\right).\label{eqn:OU LT dist}
\]
By comparing them with the true transition density, we can see that both splitting schemes are normally distributed with unbiased mean, and biased variance, with biases, lower for S, vanishing as $h\to 0$.

\subsubsection{Inhomogeneous Geometric Brownian Motion (IGBM) -  $g(x)=\sqrt{2\theta ax^2}$}
This process, defined for $a>0$, was considered in \cite{buckwar_splitting_2022} with a different spitting decomposition. Here, the LT and S schemes are given by 
\[
X^{\textrm{LT}}_{t_{k+1}}=\left(\tilde\mu + e^{-\tilde\theta h}(X_{t_k}-\tilde\mu)\right)e^{\sigma\xi_k},\qquad X^{\textrm{S}}_{t_{k+1}}=\tilde\mu+e^{-\tilde\theta h/2}\left(e^{\sigma \xi_k}\left(\tilde\mu+e^{-\tilde\theta h/2}(X_{t_k}-\tilde\mu)-\tilde\mu\right)\right),
\]
with $\tilde\mu=\mu/(1+a)$, 
with one-step condition mean and conditional mean reported therein. Both schemes are positive when $\mu>0$ (i.e. 0 is an entrance boundary), so the positivity of the process is preserved.

\subsubsection{$F$ diffusion  - $g(x)=\sqrt{2\theta ax(x+1)}$}
The name of this process comes from the fact that when $a>0$ and $\mu\geq a$, there exists a unique ergodic solution on the positive half line, with the invariant distribution being a scaled $F-$distribution with $2\mu/a$ and $ 2(1+1/a)$ degrees of freedom and scale parameter $\mu/(1+a)$. If $\mu\geq 1$, the 0 boundary is of entrance type. Here, the splitting schemes are given by 
\begin{align}
X^{\textrm{LT}}_{t_{k+1}}=&\sinh^2\left(\frac{\sqrt{2a\theta}}{2}{\xi_k}-l(X_{t_k},h)\right),\label{eqn:F LT}\\
    X^{\textrm{S}}_{t_{k+1}}=&\tilde{\mu}+e^{-\tilde{\theta}h}\left(\sinh^2\left(\frac{\sqrt{2a\theta}}{2}{\xi_k}-l(X_{t_k},\frac{h}{2})
    \right)-\tilde{\mu}\right),\label{eqn:F Strang}
\end{align}
where $\tilde{\theta}=(1+a)\theta,\; \tilde{\mu}=(\mu-a/2)/(1+a)$ and 
\begin{equation*}
l(X_{t_k},h)=\sqrt{1+\tilde{\mu}+e^{-\tilde{\theta}h}\left(X_{t_k}-\tilde{\mu}\right)}-\sqrt{\tilde{\mu}+e^{-\tilde{\theta}h}\left(X_{t_k}-\tilde{\mu}\right)}.
\end{equation*}
For the 0 entrance boundary and the ergodic case, i.e. $\mu\geq a$, both $\varphi_h^{[1]}$ and $\varphi_h^{[2]}$ are positive, so the LT and S splitting schemes preserve the state space of the process if $X_{t_k}>0$. If $\mu<a$, the 0 boundary may be obtained and the ODE solution may become  negative, making the LT and S schemes not well defined. In that case, we could choose an adaptive time-step as done for the CIR process. 

For inference, if we denote $Y^\textrm{LT}\sim \mathcal{N}(-l(X_{t_k},h),\theta ah/2)$,  we have
\begin{equation}\label{eqn:F LT transition}
    f^{\textrm{LT}}_{X_{t_{k+1}}|X_{t_k}}(y|x)=\frac{1}{2\sqrt{y(1+y)}}\left(f_{Y^\textrm{LT}}(\operatorname{arcsinh}\sqrt{y})+f_{Y^\textrm{LT}}(\operatorname{arcsinh}(-\sqrt{y}))\right).
\end{equation}
If we now denote $Y^\textrm{S}\sim \mathcal{N}(-l(X_{t_k},h/2),\theta ah/2)$ and $\widetilde{y}=e^{\frac{\tilde{\theta}h}{2}}(y-\tilde{\mu})+\tilde{\mu}$,
we have
\begin{equation}\label{eqn:F ST den}
    f^{\textrm{S}}_{X_{t_{k+1}}|X_{t_k}}(y|x)=\left|\frac{e^{\tilde{\theta }h/2}}{2\sqrt{\widetilde{y}(1+\widetilde{y})}}\right|\left(f_{Y^{\textrm{S}}}(\operatorname{arcsinh}(\widetilde{y}))+f_{Y^{\textrm{S}}}(\operatorname{arcsinh}(-\widetilde{y}))\right).
\end{equation}
In Figure \ref{fig:F inf}, we report the inference results. We simulate the underlying paths using $\mu=10.0,\theta=2.0, a=2.0, T=10.0, h^\textrm{fine}=10^{-4}, \widetilde X_\textrm{fine}=X^{\textrm{LT}}$, fix $\theta$ and infer the unknown parameter $(\mu,a)$. Our splitting estimators again outperform other baseline methods, especially when the observation time step is large. 
\begin{figure}[t]
    \centering
    \includegraphics[width=.8\linewidth]{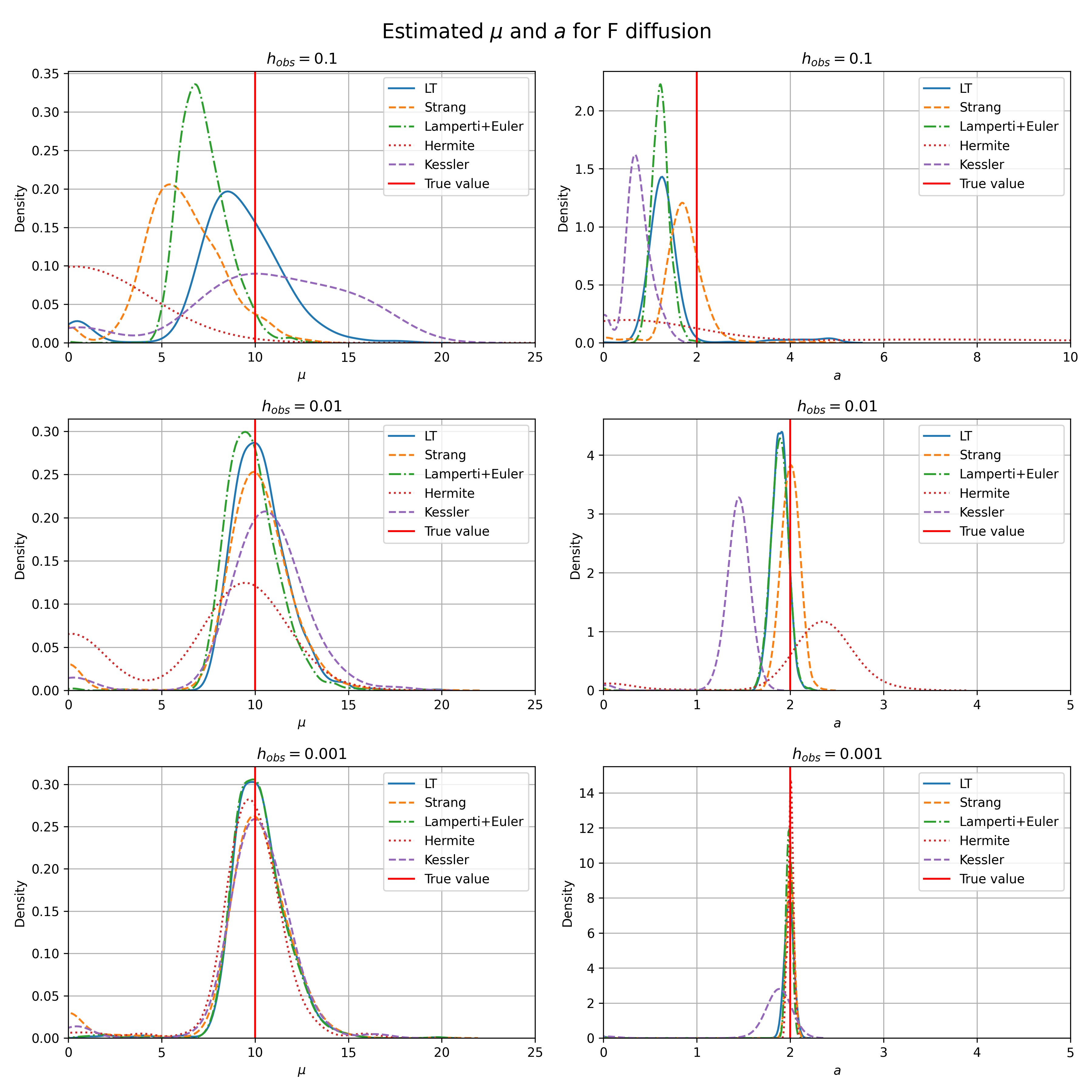}
    \caption{Inference results for the F diffusion: densities of the pseudo MLEs derived under different estimators. We fix $\theta=2.0$, and infer the unknown $\mu$ (left column) and $a$ (right column) using $M=1000$ simulated paths,  generated using $h^{\textrm{fine}}=10^{-4}$ with $\widetilde{X}^{\textrm{fine}}$ being LT, and then subsampled according to the  observation time step $h_{\textrm{obs}}$, here $ 0.1, 0.01, 0.001$, from top to bottom.}
    \label{fig:F inf}
\end{figure}

\subsubsection{Wright-Fisher - $g(x)=\sqrt{2\theta ax(x-1)}, a<0$} The Wright-Fisher diffusion, defined for $a<0$, is also known as  Jacobi diffusion because the eigenfunctions of the associated semigroup are Jacobi polynomials \cite{bingham_diffusion_2010}. Its state space is $[0,1]$. If $\mu>0$ and $\min(\mu,1-\mu)\geq -a$, that is, $\mu\in(0,1), a\in(-1/2,0)$, 
both $0$ and $1$ boundaries are of entrance type, a unique ergodic solution exists on the interval $(0,1)$, and the invariant distribution is the Beta distribution with shape parameter $-\mu/a,(\mu-1)/a$. This ergodic solution still exists for other types of boundaries that are instantaneously reflecting, but we do not discuss Feller's boundary classification here, referring to  \cite{forman_pearson_2008,karlin_second_1981} for details. 
If we denote $\tilde{\mu}=(2\mu+a)/[2(1+a)]$ and $\tilde{\theta}=\theta(1+a)$, the LT and S schemes are given by
\begin{align}
X^{\textrm{LT}}_{t_{k+1}}=&\sin^2\left(\frac{\sqrt{-2a\theta}}{2}{\xi_k}+\arcsin\sqrt{\tilde{\mu}+e^{-\tilde{\theta}h}\left(X_{t_k}-\tilde{\mu}\right)}\right),\label{eqn:WF LT}\\
X^{\textrm{S}}_{t_{k+1}}=&\tilde{\mu}+e^{-\tilde{\theta}\frac{h}{2}}\left(\sin^2\left(\frac{\sqrt{-2a\theta}}{2}{\xi_k}+\arcsin\sqrt{\tilde{\mu}+e^{-\tilde{\theta}\frac{h}{2}}\left(X_{t_k}-\tilde{\mu}\right)}\right)-\tilde{\mu}\right)\label{eqn:WF Strang}.
\end{align}
Also in this case, the LT and S splitting schemes preserve the state space of the process when the boundaries are of entrance type. Indeed, if we start from $X_{t_k}\in (0,1)$ for entrance boundaries $0$ and $1$, both splitting schemes are in $(0,1)$, as  the condition $\min(\mu,1-\mu)\geq -a$ implies $\tilde{\mu} \in(0,1), \tilde{\theta}\geq 0$, guaranteeing that  $X^{[1]}_{t_{k+1}}{\in (0,1)}$, and that $X^{[2]}_{t_{k+1}}$ is well defined and trivially in $(0,1)$. 
 We illustrate the state space  preservation in Figure \ref{fig:WF paths}, where we report $10$ trajectories simulated with the LT and S scheme in $[0,3]$ with step size $h=0.01$. 

\begin{figure}
    \centering
    \includegraphics[scale=0.55]{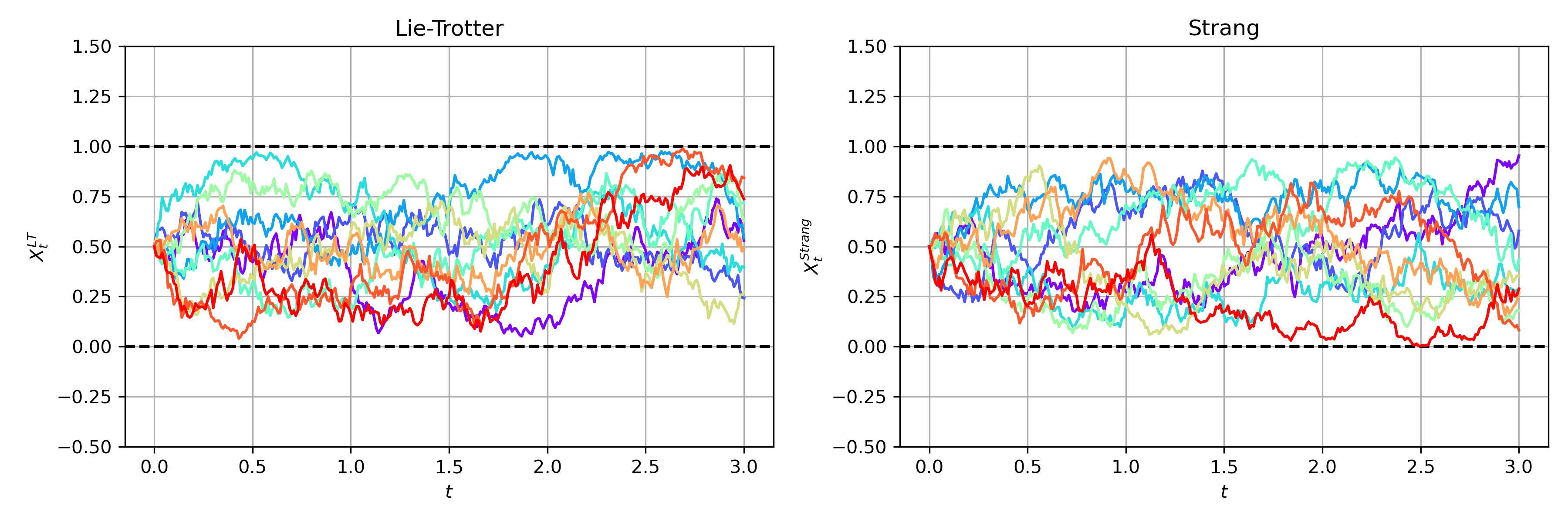}
    \caption{Ten trajectories of the Wright-Fisher diffusion simulated using the LT and S splitting schemes for entrance boundaries, 
    with $\mu=0.5,\theta=1.0,a=-0.3$.}
    \label{fig:WF paths}
\end{figure}
In Figure \ref{fig:WF X_T}, we report the asymptotic distribution of the LT scheme obtained at $T=15$, comparing it to the true beta distribution. We fix $\theta=1.0, h=0.05,T=15.0$, and vary the values of $\mu,a$. In all cases, the distribution from the LT scheme nicely matches that of the true invariant distribution with different shape parameters.   

\begin{figure}[t]
    \centering
    \includegraphics[width=0.7\textwidth]{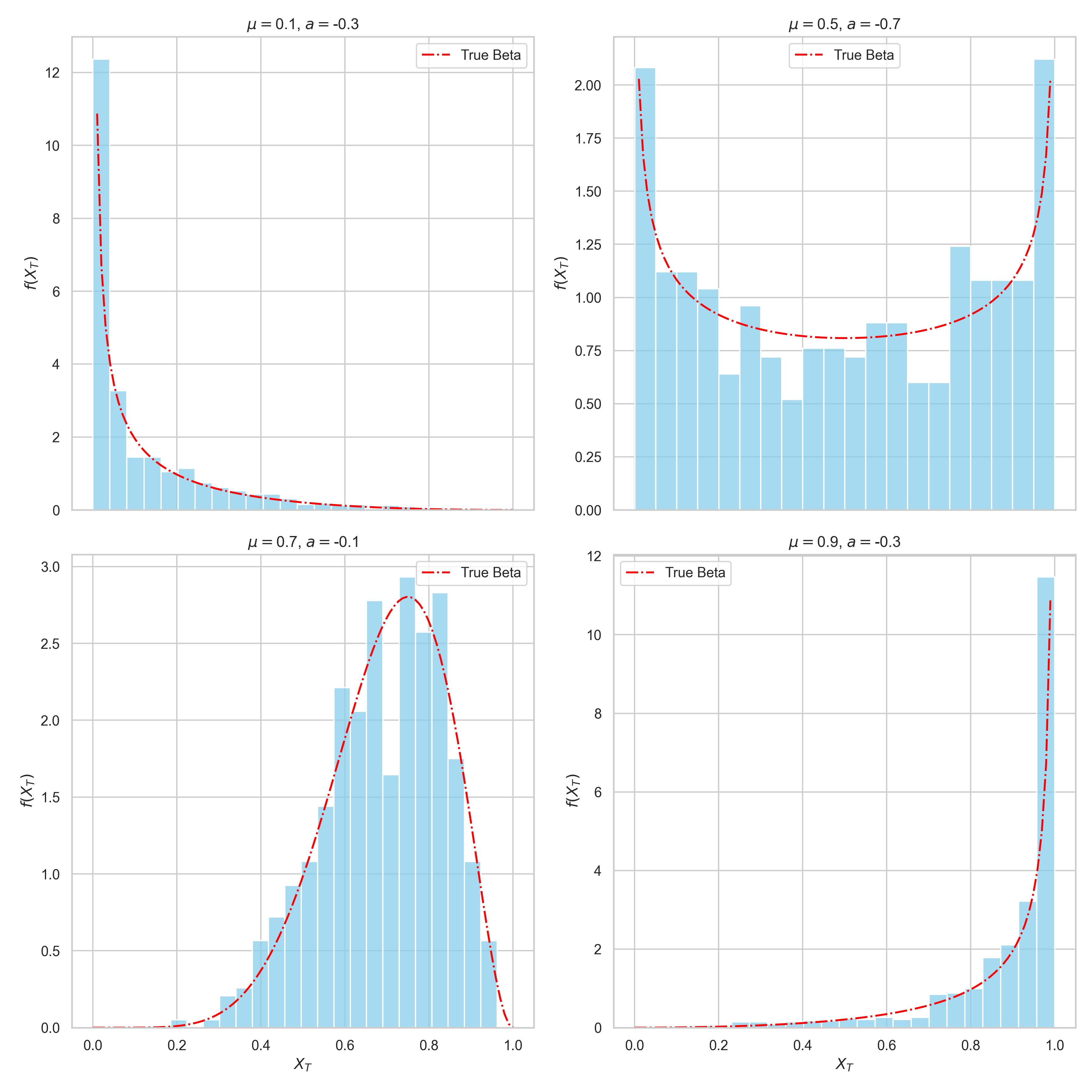}
    \caption{Wright-Fisher diffusion: Empirical distribution of the LT splitting scheme \eqref{eqn:WF LT} at $T=15$, compared with the true invariant beta distribution. We fix $\theta=1.0,h=0.05, M=1000$, and vary the values of $\mu$ and $a$.} 
    \label{fig:WF X_T}
\end{figure}

\begin{figure}
    \centering
    \includegraphics[width=0.8\linewidth]{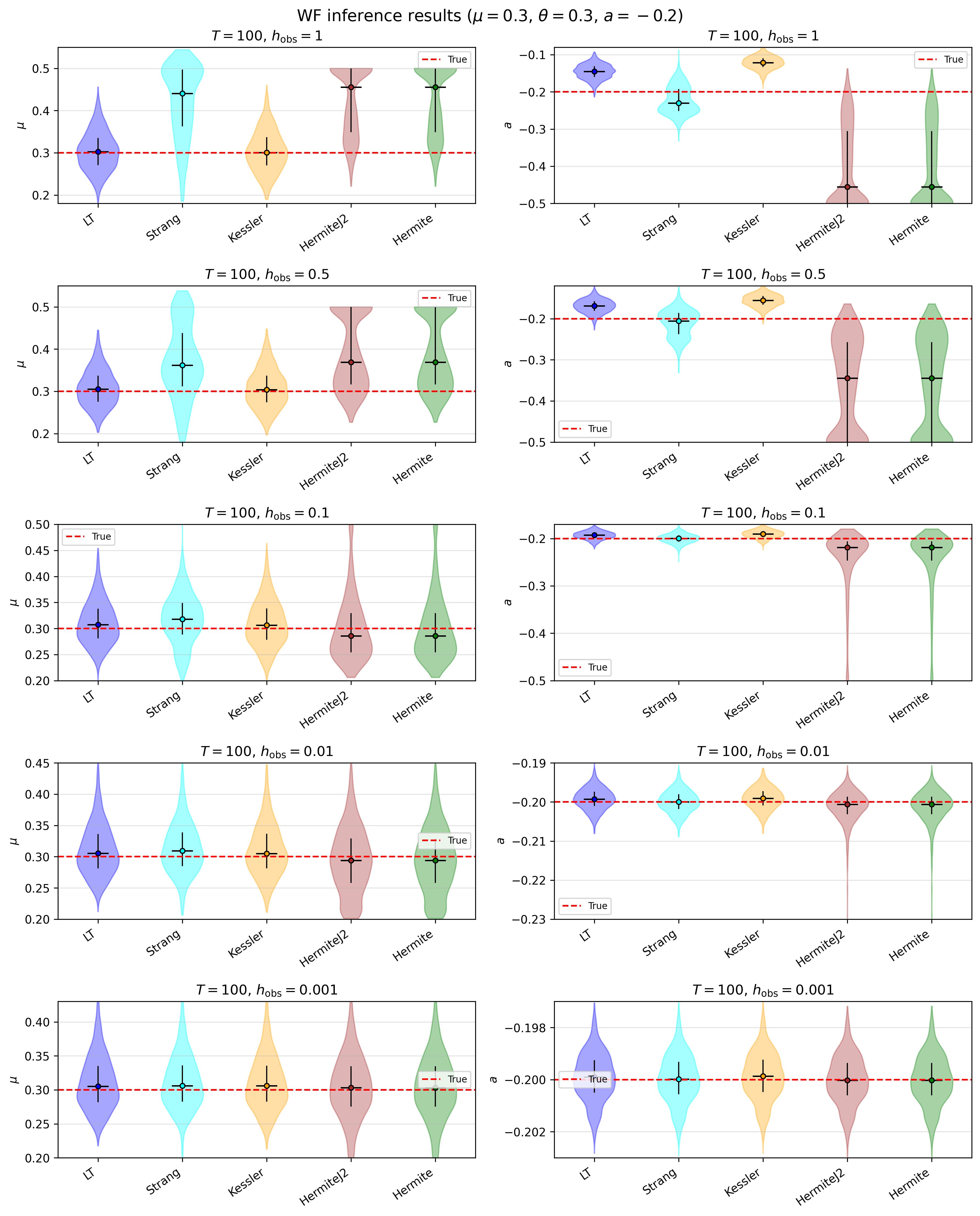}
    \caption{Inference result for the Wright-Fisher model: violin plots of the estimates derived under different estimators. We fix $\theta=0.3$, and infer the unknown $\mu$ (left columns) and $\theta$ (right columns) using $M=1000$ simulated paths,  generated using $h^{\textrm{fine}}=10^{-5}$ with SD, and then subsampled according to the observation time step $h_\textrm{obs}$, here from $1.0$ to $0.001$.}
    \label{fig:WF inf}
\end{figure}
Finally, in Figure \ref{fig:WF inf}, we report the inference results of the different estimators. We fix $\theta=0.3, a=-0.2, T=100.0$, and simulate using $h^{\textrm{fine}}=10^{-5}$ with the SD scheme. We then infer the unknown $\mu,a$ for different time steps $h_\textrm{obs}=1.0,0.5,0.1,0.01,0.001$. In this case, LT is again the most robust estimator across different time steps, together Kessler (with LT having a slight edge on the diffusion parameter). As for the AG model, S struggles with the estimation (of $\mu$ here) for large observation time steps, but still outperforms the two Hermite estimators as they collapse, and start converging only when $h_\textrm{obs}=0.1$.

\subsection{More on Pearson SDEs}

\subsubsection{Mean-square convergence order} Here, we report an
illustration of the mean-square convergence orders of different numerical schemes on the Student diffusion (Figure \ref{fig:T convergence}), F diffusion (Figure \ref{fig:T convergence}), and Wright-Fisher (Figure \ref{fig:WF convergence}), evaluated via the mean-squared error \eqref{eqn:MSE}, with $M=1000, h^\textrm{fine}=2^{-12}, \widetilde X^{\textrm{fine}}=X^\textrm{LT}$ for Student and F diffusions, and $h^\textrm{fine}=2^{-15}$ and the SD for the Wright-Fisher as reference. Except EuM, all other numerical schemes achieve a mean-square error of order one, with the S scheme yielding the smallest error in all models. Note that, by construction, EuM and Milstein do not preserve the $(0,1)$ state space of the Wright-Fisher. For this reason, during the simulation, we manually truncate the value at time $t_{k+1}$ at the boundary, if it exits the state space.

\begin{figure}
    \centering
    \includegraphics[width=0.7\textwidth]{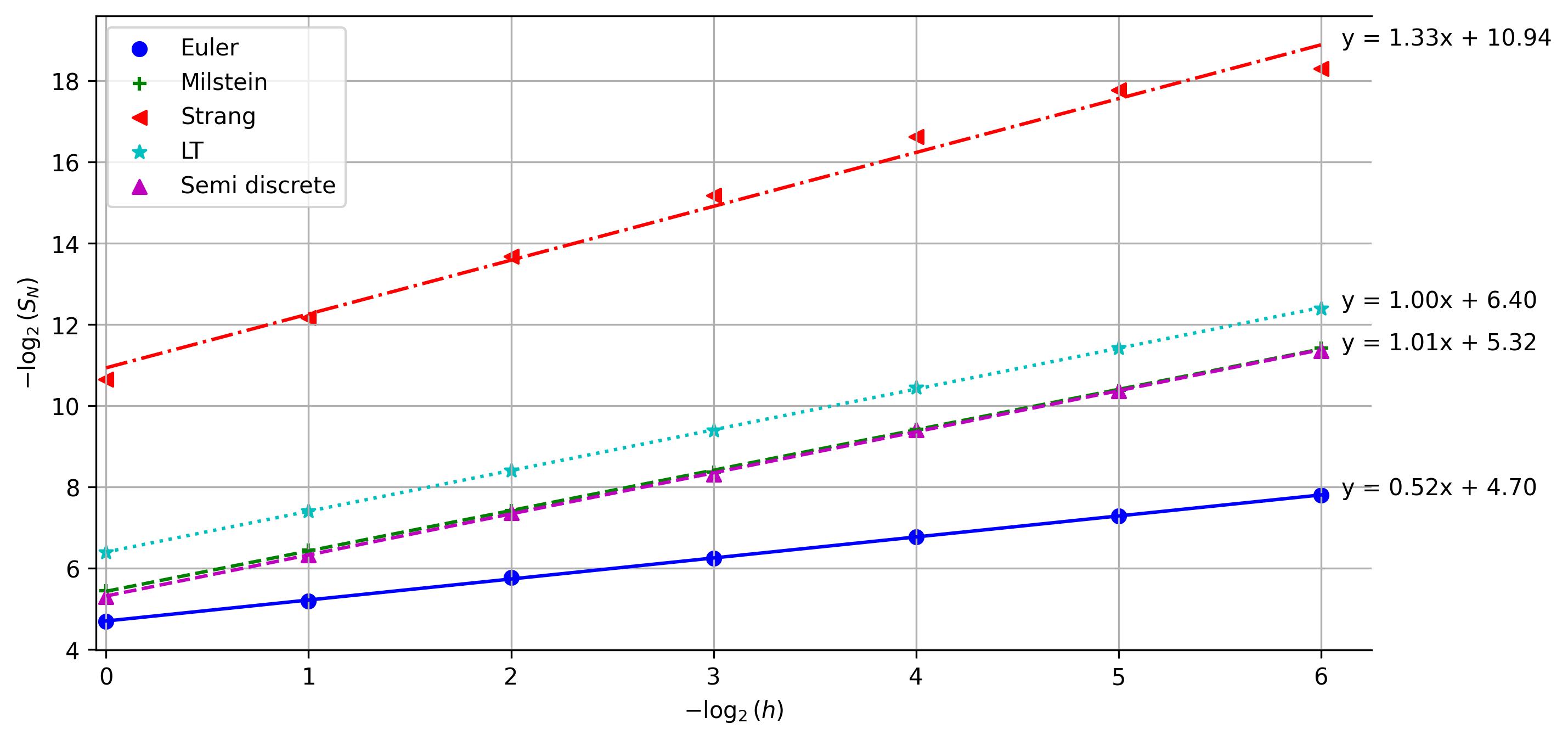}
    \caption{Illustration of the mean-square convergence order on the Student diffusion via the mean-squared error \eqref{eqn:MSE} when $\mu=10.0,\theta=2.0,a=5.0$.}
    \label{fig:T convergence}
\end{figure}

\begin{figure}[t]
    \centering
    \includegraphics[width=0.75\textwidth]{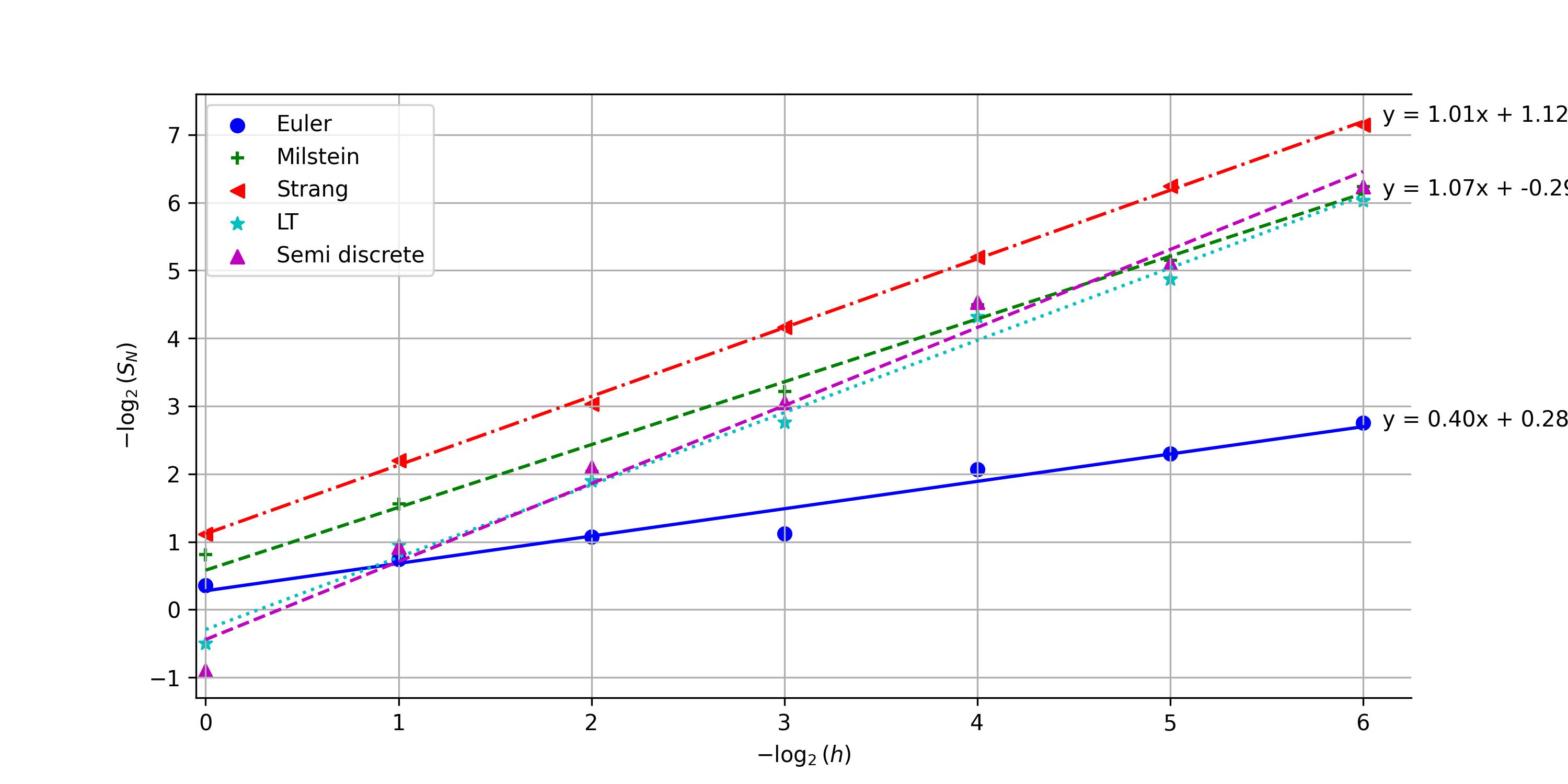}
    \caption{Illustration of the mean-square convergence order on the F diffusion  via the mean-squared error \eqref{eqn:MSE} when $\mu=2.5,\theta=0.5,a=1.0$.}
    \label{fig:F convergence}
\end{figure}

\begin{figure}[t]
    \centering
    \includegraphics[width=0.78\textwidth]{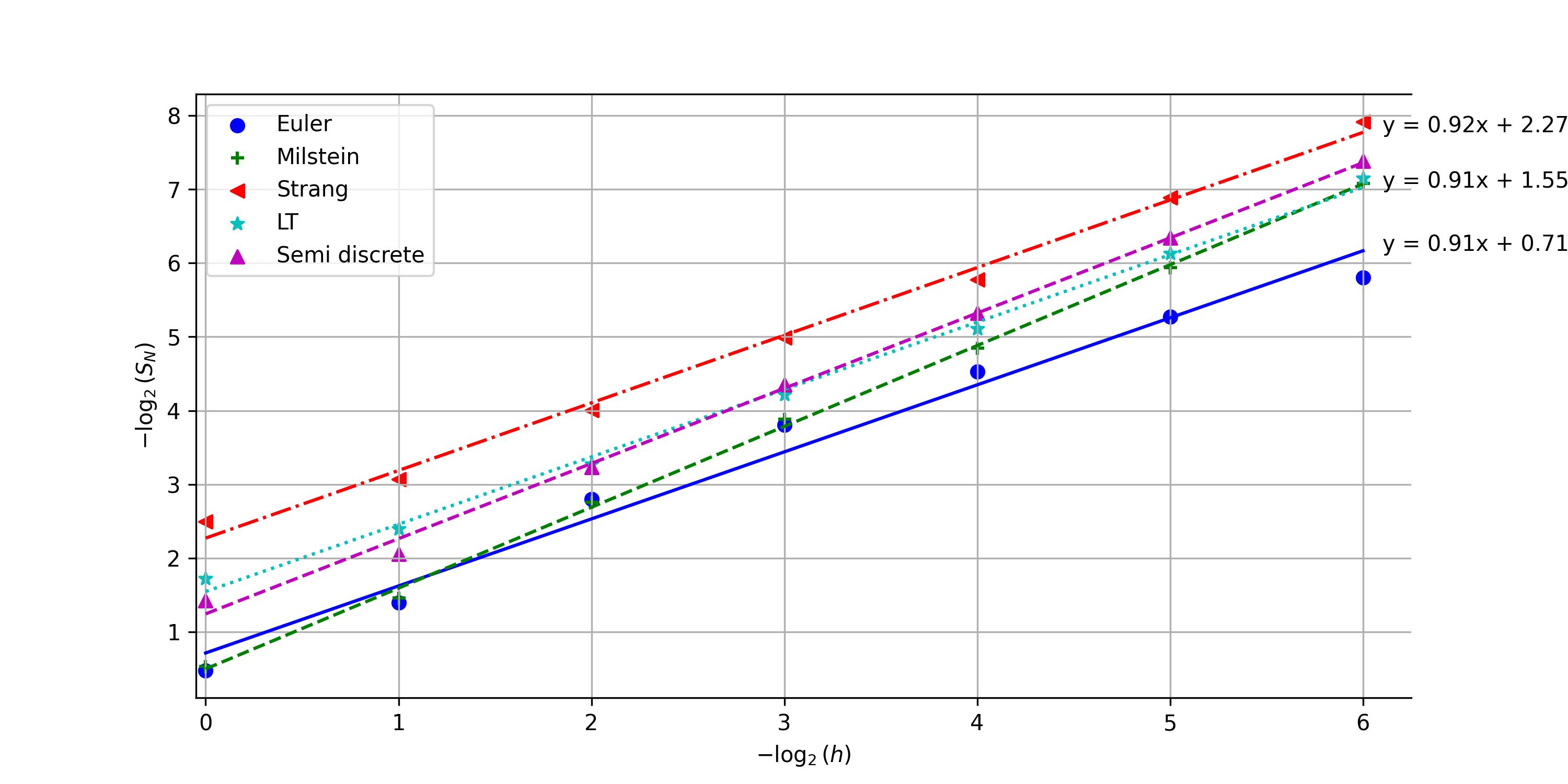}
    \caption{Illustration of the mean-square convergence order on the Wright-Fisher  via the mean-squared error \eqref{eqn:MSE} when $\mu=0.3,\theta=0.3,a=-0.2$}
    \label{fig:WF convergence}
\end{figure}

\subsubsection{Wasserstein distance for CIR process}
In Figure \ref{fig:CIR wasserstein}, we report the Wasserstein distance between the one-step transition density of the CIR process and its numerical approximation computed via Lamperti plus LT \cite{kelly_adaptive_2023}, LT, and S for different parameter values as a function of the time step $h$. In all setting, the S scheme achieves the lowest discrepancy, followed by LT. 

\begin{figure}[t]
    \centering
\includegraphics[width=.9\textwidth]{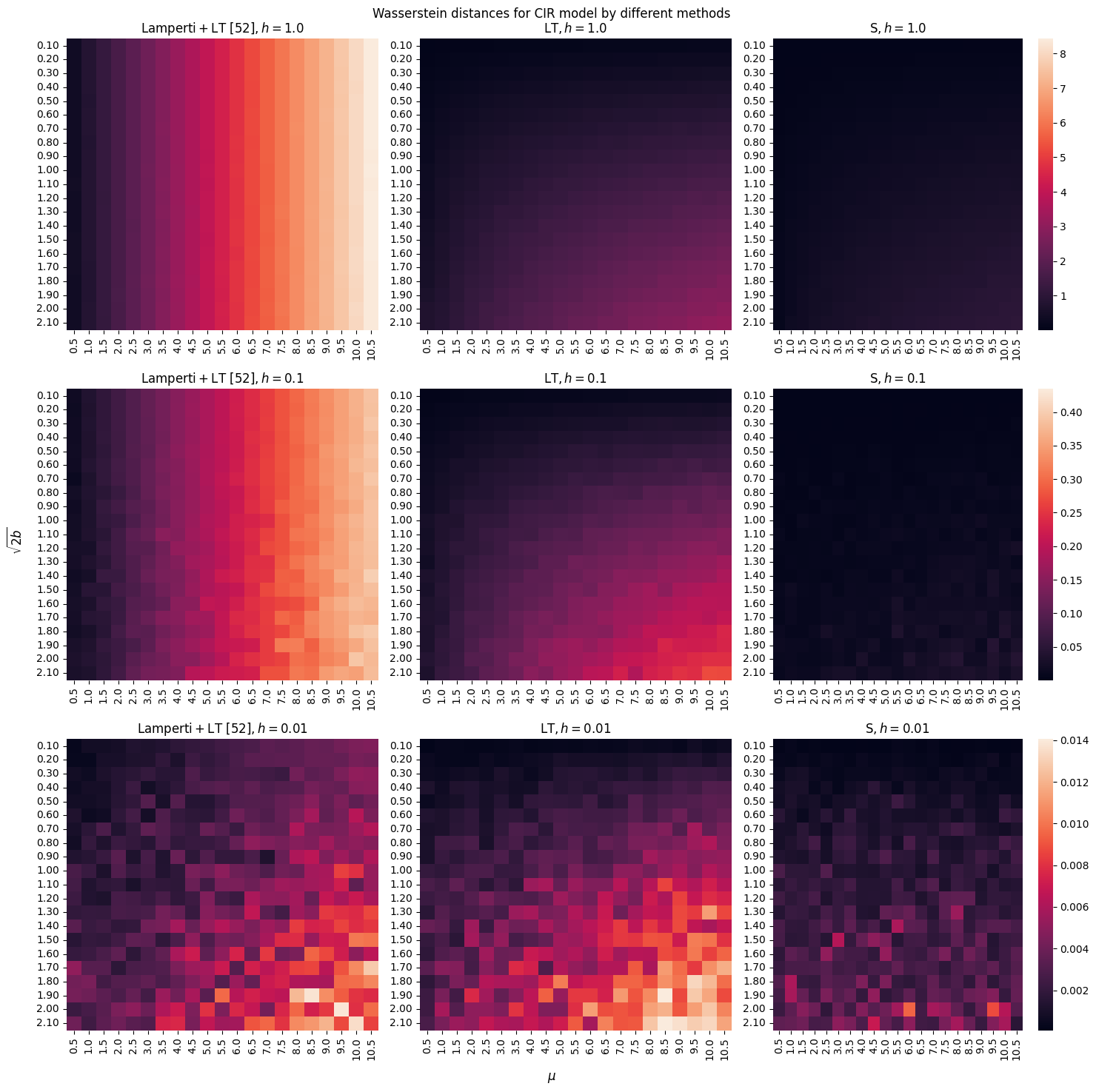}
    \caption{Wasserstein distance between the one-step transition density (from time $0$ to $h$) of the CIR process and of the different numerical schemes as a function of the parameters $\mu$ and $b$ for different time steps $h$.}
    \label{fig:CIR wasserstein}
\end{figure}

\subsection{Additional models}
Here, we consider two additional models for which our proposed splitting schemes can be derived, and inference performed (results not shown). 
\subsubsection{Stochastic Ginzburg-Landau equation}
The stochastic version of the Ginzburg–Landau model with multiplicative noise can be written as \cite{hutzenthaler_strong_2011}
\begin{equation}\label{eqn:Ginzburg-Landau}
    dX_t=\left((\eta+\frac{1}{2}\sigma^2)X_t-\lambda X_t^3\right)dt+\sigma X_t dW_t,
\end{equation}
where $\eta\geq 0,\lambda,\sigma>0$. Its analytical solution is \cite{kloeden_numerical_1992}
\begin{equation*}
    X_t=\frac{X_0 \exp \left(\eta t+\sigma W_t\right)}{\sqrt{1+2 X_0^2 \lambda \int_0^t \exp \left(2 \eta s+2 \sigma W_s\right) d s}},
\end{equation*}
for $t\in [0,T]$. Using the proposed splitting decomposition \eqref{eqn:ODE1}-\eqref{eqn:SDE2}, the corresponding solutions with initial conditions $X_{t_k}$ are given by
\[
X^{[1]}_{t_{k+1}}=\varphi^{[1]}_h(X_{t_k})=\sqrt{\frac{e^{2\eta h}K(X_{t_k})\eta }{1+\lambda e^{2\eta h}K(X_{t_k})}}, \qquad
X^{[2]}_{t_{k+1}}=\varphi^{[1]}_h(X_{t_k})=X_{t_k}e^{-\frac{\sigma^2h}{2}+\sigma \xi_k},
\]
with $K(X_{t_k})=X^2_{t_k}/(\eta -\lambda X^2_{t_k})$. Thus, the LT and S schemes are given by
\begin{align}
\label{GLLT}X^{\textrm{LT}}_{t_{k+1}}=&\sqrt{\frac{e^{2\eta h}K(X_{t_k})\eta }{1+\lambda e^{2\eta h}K(X_{t_k})}}e^{-\frac{\sigma^2h}{2}+\sigma \xi_k},\\
    X^{\textrm{S}}_{t_{k+1}}=&\sqrt{\frac{e^{\eta h}K(\sqrt{\frac{e^{\eta h}K(X_{t_k})\eta }{1+\lambda e^{\eta h}K(X_{t_k})}}e^{-\frac{\sigma^2h}{2}+\sigma 
    \xi_k})\eta }{1+\lambda e^{\eta h}K(\sqrt{\frac{e^{\eta h}K(X_{t_k})\eta }{1+\lambda e^{\eta h}K(X_{t_k})}}e^{-\frac{\sigma^2h}{2}+\sigma \xi_k})}}.
\end{align}
From \eqref{GLLT}, we immediately see that $X_{t_{k+1}}^{\textrm{LT}}|X_{t_k}\sim \textrm{LogN}(-\sigma^2h/2+\ln a,\sigma^2h)$, with $a$ being the square root term in $X^{[1]}_{t_{k+1}}$, and $\textrm{logN}$ denoting a LogNormal distribution, allowing for a direct derivation of the pseudo-log-likelihood. Moreover, this can also be used to compute the one-step conditional mean and variance of the LT scheme. 
However, the $K(x)$ term may be negative for some $x$, making the $X^{[1]}_{t_{k+1}}$ term, and thus $X^\bullet_{t_{k+1}}$, not well defined for some $x$. To tackle this, we consider a different SDE decomposition, taking advantage of another reducible SDE (case 2 in Chapter 4 of \cite{kloeden_numerical_1992}), see also Remark \ref{case2} with $d_1(x)$. In particular, 
we propose the following decomposition
\begin{align*}
    dX^{[1]}_t=&-\lambda (X^{[1]}_{t})^3dt,\\
    dX^{[2]}_t=&(\eta+\frac{\sigma^2}{2})X^{[2]}_tdt+\sigma X^{[2]}_tdW_t,
\end{align*}
with solutions
\begin{align*}
X^{[1]}_{t_{k+1}}=\frac{X_{t_k}}{\sqrt{2X_{t_k}^2\lambda h+1}},\qquad X^{[2]}_{t_{k+1}}=X_{t_k}e^{\eta h+\sigma\xi_k}.
\end{align*}
Then, the LT and S schemes are given by 
\begin{align*}
X^{\textrm{LT}}_{t_{k+1}}=\frac{X_{t_k}}{\sqrt{2X_{t_k}^2\lambda h+1}}e^{\eta h+\sigma\xi_k},\qquad
X^{\textrm{S}}_{t_{k+1}}=\frac{X_{t_k} e^{\eta h+\sigma \xi_k}}{\sqrt{1+\lambda h X_{t_k}^2+\lambda h X_{t_k}^2 e^{2(\eta h+\sigma \xi_k)}}}.
\end{align*}
Also in this case, $X^\textrm{LT}_{t_{k+1}}|X_{t_k}\sim \textrm{LogN}(\theta h+\ln(X^{[1]}_{t_{k+1}}),\sigma^2h)$. However, both subquestions' solutions are now positive and well defined, with the LT and S solutions preserving the state space of the process.

\begin{figure}
    \centering
\includegraphics[width=0.75\textwidth]{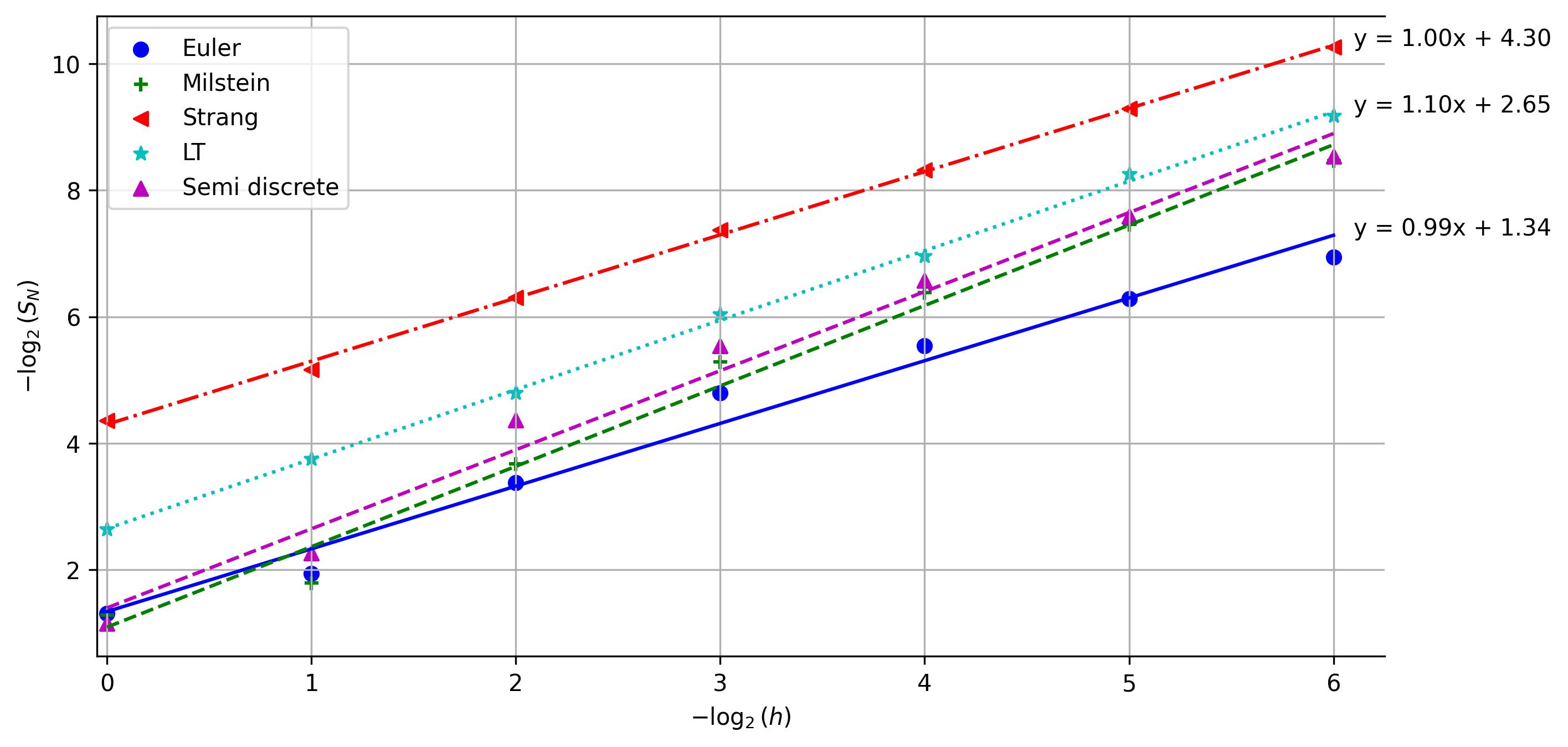}
    \caption{Illustration of the mean-square convergence order on the Ginzburg-Landau SDE \eqref{eqn:Ginzburg-Landau}
   via the mean-squared error \eqref{eqn:MSE}.}
    \label{fig:GL convergence}
\end{figure}
In Figure \ref{fig:GL convergence}, we report the numerical strong order of convergence for the LT and S splitting, with the reference scheme being the SD scheme. The pattern of the order of convergence is the same as for the Pearson diffusion class, with S achieving the best performance.

\subsubsection{Stochastic Verhulst equation}
This model is driven by the following SDE
\begin{equation}\label{eqn:Verhulst SDE}
    dX_t=\left((\eta+\frac{1}{2}\sigma^2)X_t-\lambda X_t^2\right)dt+\sigma X_t dW_t,
\end{equation}
where $\eta,\lambda,\sigma>0$, which resembles \eqref{eqn:Ginzburg-Landau}, except that the drift term is now quadratic. Also in this case, an explicit solution is available \cite{kloeden_numerical_1992}
\begin{equation}
    X_t=\frac{X_0 \exp \left(\eta t+\sigma W_t\right)}{1+X_0 \lambda \int_0^t \exp \left(\eta s+\sigma W_s\right) d s},
\end{equation}
for $t\in[0,T]$, with $0$ being an entrance boundary, i.e., the process is positive. Using the proposed splitting decomposition \eqref{eqn:ODE1}-\eqref{eqn:SDE2}, the corresponding solutions with initial conditions $X_{t_k}$ are given by
\[X^{[1]}_{t_{k+1}}=\frac{\eta e^{\eta h}X_{t_k}}{\eta+\lambda X_{t_k}(e^{\eta h}-1)},\qquad 
X^{[2]}_{t_{k+1}}=X_{t_k}e^{\sigma \xi_k},
\]
with the one-step LT and S schemes obtained via \eqref{eqn:LT subsystem} and \eqref{eqn:Strang subsystem}, respectively. Also in this case, the one-step LT conditional density is LogNormal distributed, allowing for a direct derivation of the pseudo-log-likelihood, and of the one-step conditional mean and variance. Note also that both schemes are positive, preserving thus the state space of the process.
\end{appendix}
\end{document}